\def\LambdaTRDeanonymize{1}
\def\LambdaTRIncludeAppendix{1}
\def\LambdaTRArxiv{1}
\newif\ifLambdaTRDeanonymize
\ifdefined\LambdaTRDeanonymize
  \LambdaTRDeanonymizetrue
\fi
\ifdefined\LambdaTRArxiv
  \documentclass[acmsmall,screen,nonacm]{acmart}
\else
  \ifLambdaTRDeanonymize
    \documentclass[acmsmall,screen,review,nonacm]{acmart}
  \else
    \documentclass[acmsmall,screen,review,anonymous,nonacm]{acmart}
  \fi
\fi
\usepackage{ff}
\newcommand{\jbox}[1]{{\mprset{atop}\inferrule{\fbox{$#1$}}{ }}}

\newcommand{\allsat}{\sat{\env}{\ctx}}
\NewDocumentCommand{\enva}{O{v}}{\extenv{\env}{x}{#1}}
\NewDocumentCommand{\envw}{}{\enva[w]}

\newcommand{\ctxa}{\ctx \exta \bindvar}
\newcommand{\IfR}{\R}
\newcommand{\ifbrctxbase}{\ctx}
\newcommand{\ifbrposprop}{\posprop[c]}
\newcommand{\ifbrnegprop}{\negprop[c]}
\newcommand{\IfCondR}{\res[\tyA][\ifbrposprop][\ifbrnegprop][\obj_c]}
\NewDocumentCommand{\ifbrctx}{m}{\ifbrctxbase \extb #1}
\NewDocumentCommand{\ifbrR}{O{2}}{\res[\tyB][\possynprop[#1]][\negsynprop[#1]][\obj_2]}
\NewDocumentCommand{\appfunrngres}{}{\R}
\NewDocumentCommand{\appfunrngresunfold}{}{\res[\tyC][\possynprop[f]][\negsynprop[f]][\obj_f]}
\NewDocumentCommand{\appargres}{O{\tyB}}{\res[#1][\possynprop[2]][\negsynprop[2]][\obj_2]}
\NewDocumentCommand{\appfunty}{O{\tyB}}{\funty{x}{#1}{\appfunrngres}}
\NewDocumentCommand{\appfunres}{}{\res[\appfunty][\possynprop[1]][\negsynprop[1]][\obj_1]}
\NewDocumentCommand{\appresinner}{}{\appfunrngres}
\NewDocumentCommand{\appres}{}{\osubst{\appresinner}[x][\obj_2]}
\NewDocumentCommand{\propaconseq}{O{\unexpr} O{\unexpr'} O{\R_a} O{\R_b}}{Then by definition,  $\expinterp{#4}{#2}[0]$. Then by \Cref{appendix:lem:hdxp}, $\expinterp{#4}{#1}$}
\NewDocumentCommand{\propaconseqsub}{O{\unexpr} O{\unexpr'} m m O{\R_b} O{n}}{Therefore, $\neestepsto{#3}{#4}[#6]$.
  Then by definition,
  \begin{align*}
    \expinterp{#5}{#4}[0]
  \end{align*}
  Then by \Cref{appendix:lem:hdxp}, \begin{align*}\expinterp{#5}{#3}[#6]\end{align*}}
\NewDocumentCommand{\presuppositions}{}{
  \begin{itemize}
  \item $\wfctx$ in all the judgements where $\ctx$ appears on the left hand of $\vdash$;
  \item $\subtype{\ctx}{\tyA}{\tyB}$ presupposes $\wftype{\tyA}$ and $\wftype{\tyB}$;
  \item $\subtype{\ctx}{\obj_1}{\obj_2}$ presupposes $\wfobj{\obj_1}$ and $\wfobj{\obj_2}$;
  \item $\subtype{\ctx}{\R_1}{\R_2}$ presupposes $\wfres{\R_1}$ and $\wfres{\R_2}$;
  \item $\restrictrel{\ctx}{\tyA}{\tyB}{\tyC}$ presupposes $\wftype{\tyA}$ and $\wftype{\tyB}$, and implies $\wftype{\tyC}$;
  \item $\removerel{\ctx}{\tyA}{\tyB}{\tyC}$ presupposes $\wftype{\tyA}$ and $\wftype{\tyB}$, and implies $\wftype{\tyC}$;
  \item $\proves{\ctx}{\prop}$ presupposes $\wfprop{\prop}$;
  \item $\typeof{\ctx}{\expr}{\R}$ presupposes $\wfres{\R}$.
  \item We further assume that every binder has a unique name.
  \end{itemize}
}
\NewDocumentCommand{\directworkrefs}{}{\cite{bonnaire-sergeant2016,chang2017,fu2021,kent2016}}

\newcommand{\showsyntaxfigure}{%
\begin{figure}[ht]
  \centering
  \begin{tabular}{>{\itshape}l<{} >{$}r<{$}@{\hskip 0.10cm} >{$}l<{$}}
    Base values & \bv \sdef &
    \true \mid \false \mid \nat \mid \str
    \\
    Primitive operations & \op \sdef &
    \isstrop \mid \isnatop \mid \addop \mid \notop
    \\
    Expressions & \expr \sdef &
    \bv \mid \op \mid x \mid
    \abs{x}{\expr} \mid \ap{\expr}{\expr} \mid \If{\expr}{\expr}{\expr} \mid \Let{\expr}{\expr}
    \\
    & & \\
    Types & \tyA , \tyB \sdef &
    \top \mid \bot \mid  \true \mid \false \mid \Nat \mid \Str \mid
    \uniontype{\tyA}{\tyB} \mid \funty{x}{\tyA}{\R}
    \\
    Propositions & \prop\sdef &
    \TT \mid \FF \mid \eqprop{x}{y} \mid \typeprop{x}{\tyA} \mid  \nottypeprop{x}{\tyA}
    \mid \conjprop{\prop}{\prop} \mid \disjprop{\prop}{\prop}
    \\
    Objects & \obj \sdef &
    x \mid \topobj \mid \botobj
    \\
    Type results & \R \sdef &
    \res[\tyA][\prop][\prop][\obj]
    \\
    Contexts & \ctx \sdef &
    \empctx \mid \ctx \cdot \prop \mid \ctx \exta \bindvar
  \end{tabular}
  \caption{Syntax}
  \trlabel{fig:stx}
\end{figure}
}

\newcommand{\showctxdomainfigure}{%
\begin{figure}[ht]
  \begin{align*}
    \dom{\empctx} &\mdef \varnothing \\
    \dom{\ctx \exta \bindvar} &\mdef \dom{\ctx} \cup \{x\} \\
    \dom{\ctx \extb \prop} &\mdef \dom{\ctx}
  \end{align*}
  \caption{Domain of a Context}
  \trlabel{fig:ctxdomain}
\end{figure}
}

\newcommand{\showwffigure}{%
\begin{figure}[!ht]
  \begin{mathpar}
  \jbox{\wfctx}
  \and
  \inferrule
    { }
    {\wfctx[\empctx]}
  \and
  \inferrule
    {\wfctx \\
     \wftype{\tyA}}
    {\wfctx[\ctx \cdot \bindvar]}
  \and
  \inferrule
    {\wfctx \\
     \wfprop{\prop}}
    {\wfctx[\ctx \cdot \prop]}
  \end{mathpar}
  \lighthrule
  \begin{mathpar}
  \jbox{\wfobj{\obj}}
  \and
  \inferrule
    {x \in \dom{\ctx}}
    {\wfobj{x}}
  \and
  \inferrule
    { }
    {\wfobj{\botobj}}
  \and
  \inferrule
    { }
    {\wfobj{\topobj}}
  \end{mathpar}
  \lighthrule
  \begin{mathpar}
  \jbox{\wfprop{\prop}}
  \and
  \inferrule
    { }
    {\wfprop{\TT}}
  \and
  \inferrule
    { }
    {\wfprop{\FF}}
  \and
  \inferrule
    {x \in \dom{\ctx} \\
     \wftype{\tyA}}
    {\wfprop{\typeprop{x}{\tyA}}}
  \and
  \inferrule
    {x \in \dom{\ctx} \\
     \wftype{\tyA}}
    {\wfprop{\nottypeprop{x}{\tyA}}}
  \and
  \inferrule
    {x \in \dom{\ctx} \\
     y \in \dom{\ctx}}
    {\wfprop{\eqprop{x}{y}}}
  \and
  \inferrule
    {\wfprop{\prop_1} \\
     \wfprop{\prop_2}}
    {\wfprop{\conjprop{\prop_1}{\prop_2}}}
  \and
  \inferrule
    {\wfprop{\prop_1} \\
     \wfprop{\prop_2}}
    {\wfprop{\disjprop{\prop_1}{\prop_2}}}
  \end{mathpar}
  \lighthrule
  \begin{mathpar}
  \jbox{\wftype{\tyA}}
  \and
  \inferrule
    { }
    {\wftype{\true}}
  \and
  \inferrule
    { }
    {\wftype{\false}}
  \and
  \inferrule
    { }
    {\wftype{\Nat}}
  \and
  \inferrule
    { }
    {\wftype{\Str}}
  \and
  \inferrule
    { }
    {\wftype{\top}}
  \and
  \inferrule
    { }
    {\wftype{\bot}}
  \and
  \inferrule
    {\wftype{\tyA} \\
     \wftype{\tyB}}
    {\wftype{\uniontype{\tyA}{\tyB}}}
  \and
  \inferrule
    {\wftype{\tyA} \\
     \wfres[\ctx\cdot\bindvar]{\synR}}
    {\wftype{\funty{x}{\tyA}{\synR}}}
  \end{mathpar}
  \lighthrule
  \begin{mathpar}
  \newcommand{\ctxw}{\ctx}
  \jbox{\wfres{\R}}
  \and
  \inferrule
    {\wftype[\ctxw]{\tyA} \\
     \wfobj[\ctxw]{\obj} \\
     \wfprop[\ctxw]{\posprop} \\
     \wfprop[\ctxw]{\negprop}}
    {\wfres{\res}}
  \end{mathpar}
  \caption{Well-Formedness Judgments}
  \trlabel{fig:wf-judgments}
\end{figure}
}

\newcommand{\showproofsystemfigure}{%
\begin{figure}[ht]
  \begin{mathpar}
  \jbox{\proves{\ctx}{\prop}}
  \and
  \inferrule[L-True]
    {\ }
    {\proves{\ctx}{\TT}}
  \and
  \inferrule[L-False]
    {\proves{\ctx}{\FF} \\
     \wfprop{\prop}}
    {\proves{\ctx}{\prop}}
  \and
  \inferrule[L-Atom]
    {\lookup{\ctx}{\prop}}
    {\proves{\ctx}{\prop}}
  \and
  \inferrule[L-Ann]
    {(\bindvar) \in \ctx}
    {\proves{\ctx}{\typeprop{x}{\tyA}}}
  \and
  \inferrule[L-AndI]
    {\proves{\ctx}{\prop_1} \\
     \proves{\ctx}{\prop_2}}
    {\proves{\ctx}{\conjprop{\prop_1}{\prop_2}}}
  \and
  \inferrule[L-AndE1]
    {\proves{\ctx}{\conjprop{\prop_1}{\prop_2}}}
    {\proves{\ctx}{\prop_1}}
  \and
  \inferrule[L-AndE2]
    {\proves{\ctx}{\conjprop{\prop_1}{\prop_2}}}
    {\proves{\ctx}{\prop_2}}
  \and
  \inferrule[L-OrI1]
    {\proves{\ctx}{\prop_1} \\
     \wfprop{\prop_2}}
    {\proves{\ctx}{\disjprop{\prop_1}{\prop_2}}}
  \and
  \inferrule[L-OrI2]
    {\proves{\ctx}{\prop_2} \\
     \wfprop{\prop_1}}
    {\proves{\ctx}{\disjprop{\prop_1}{\prop_2}}}
  \and
  \inferrule[L-OrE]
    {\wfprop{\prop} \\
     \proves{\ctx}{\disjprop{\prop_1}{\prop_2}} \\\\
     \proves{\ctx \extb \prop_1}{\prop} \\
     \proves{\ctx \extb \prop_2}{\prop}}
    {\proves{\ctx}{\prop}}
  \and
  \inferrule[L-Bot]
    {\proves{\ctx}{\typeprop{x}{\bot}}}
    {\proves{\ctx}{\FF}}
  \and
  \inferrule[L-Sub]
    {\proves{\ctx}{\typeprop{x}{\tyA}} \\
     \subtype{\ctx}{\tyA}{\tyB}}
    {\proves{\ctx}{\typeprop{x}{\tyB}}}
  \and
  \inferrule[L-SubNot]
    {\proves{\ctx}{\nottypeprop{x}{\tyA}} \\
     \subtype{\ctx}{\tyB}{\tyA}}
    {\proves{\ctx}{\nottypeprop{x}{\tyB}}}
  \and
  \inferrule[L-Refl]
    {x \in \dom{\ctx}}
    {\proves{\ctx}{\eqprop{x}{x}}}
  \and
  \inferrule[L-Transp]
    {\wfprop[{\ctx \exta \bindvar[{z}][{\top}]}]{\prop} \\\\
     \proves{\ctx}{\substprop{\prop}{z}{x}} \\
     \proves{\ctx}{\eqprop{x}{y}}}
    {\proves{\ctx}{\substprop{\prop}{z}{y}}}
  \and
  \inferrule[L-Restrict]
    {\proves{\ctx}{\typeprop{x}{\tyA}} \\
     \proves{\ctx}{\typeprop{x}{\tyB}} \\\\
     \restrictrel{\ctx}{\tyA}{\tyB}{\tyC}}
    {\proves{\ctx}{\typeprop{x}{\tyC}}}
  \and
  \inferrule[L-Remove]
    {\proves{\ctx}{\typeprop{x}{\tyA}} \\
     \proves{\ctx}{\nottypeprop{x}{\tyB}} \\\\
     \removerel{\ctx}{\tyA}{\tyB}{\tyC}}
    {\proves{\ctx}{\typeprop{x}{\tyC}}}
  \end{mathpar}
  \caption{Proof System}
  \trlabel{fig:proofsystem}
\end{figure}
}

\newcommand{\showrestrictremovefigure}{%
\begin{figure}[ht]
  \begin{mathpar}
  \jbox{\restrictrel{\ctx}{\tyA}{\tyB}{\tyC}}
  \and
  \inferrule[Restrict-Base]
    {\tyA, \tyB \in \{\true,\false,\Nat,\Str\} \quad \tyB \neq \tyA}
    {\restrictrel{\ctx}{\tyA}{\tyB}{\bot}}
  \and
  \inferrule[Restrict-BaseFun]
    {\tyA \in \{\true,\false,\Nat,\Str\} \\
     \wftype{\funty{x}{\tyB}{\R}}}
    {\restrictrel{\ctx}{\tyA}{\funty{x}{\tyB}{\R}}{\bot}}
  \and
  \inferrule[Restrict-Union]
    {\restrictrel{\ctx}{\tyA}{\tyB}{\tyC} \\
     \restrictrel{\ctx}{\tyD}{\tyB}{\tyE}}
    {\restrictrel{\ctx}{\uniontype{\tyA}{\tyD}}{\tyB}{\uniontype{\tyC}{\tyE}}}
  \and
  \inferrule[Restrict-Sym]
    {\restrictrel{\ctx}{\tyB}{\tyA}{\tyC}}
    {\restrictrel{\ctx}{\tyA}{\tyB}{\tyC}}
  \and
  \inferrule[Restrict-Sub]
    {\subtype{\ctx}{\tyA}{\tyB}}
    {\restrictrel{\ctx}{\tyA}{\tyB}{\tyA}}
  \and
  \inferrule[Restrict-Other]
    {\wftype{\tyA} \\
     \wftype{\tyB}}
    {\restrictrel{\ctx}{\tyA}{\tyB}{\tyB}}
  \end{mathpar}
  \lighthrule
  \begin{mathpar}
  \jbox{\removerel{\ctx}{\tyA}{\tyB}{\tyC}}
  \and
  \inferrule[Remove-Union]
    {\removerel{\ctx}{\tyA}{\tyB}{\tyC} \\
     \removerel{\ctx}{\tyD}{\tyB}{\tyE}}
    {\removerel{\ctx}{\uniontype{\tyA}{\tyD}}{\tyB}{\uniontype{\tyC}{\tyE}}}
  \and
  \inferrule[Remove-Sub]
    {\subtype{\ctx}{\tyA}{\tyB}}
    {\removerel{\ctx}{\tyA}{\tyB}{\bot}}
  \and
  \inferrule[Remove-Other]
    {\wftype{\tyA} \\
     \wftype{\tyB}}
    {\removerel{\ctx}{\tyA}{\tyB}{\tyA}}
  \end{mathpar}
  \caption{Rules for Restrict and Remove}
  \trlabel{fig:restrict-remove}
\end{figure}
}

\newcommand{\tcbaseval}{
  \inferrule[T-BaseVal]
    {\ }
    {\typeof{\ctx}{\bv}{\typeofbv{\bv}}}
}

\newcommand{\tcprimop}{
  \inferrule[T-PrimOp]
    {\tyA = \typeofop{\op}}
    {\typeof{\ctx}{\op}{\res[\tyA][\TT][\FF][\topobj]}}
}

\newcommand{\tcvar}{
  \inferrule[T-Var]
    {\proves{\ctx}{\typeprop{x}{\tyA}}}
    {\typeof{\ctx}{x}{\res[\tyA][\nottypeprop{x}{\false}][\typeprop{x}{\false}][x]}}
}

\newcommand{\tcabs}{
  \inferrule[T-Abs]
    {\typeof{\ctxa}{\synexpr_b}{\synR}}
    {\typeof{\ctx}{\abs{x}{\synexpr_b}}{\res[\funty{x}{\tyA}{\synR}][\TT][\FF][\topobj]}}
}

\newcommand{\tcapp}{
  \inferrule[T-App]
    {\typeof{\ctx}{e_1}{\appfunres} \\\\
     \typeof{\ctx}{e_2}{\appargres}}
    {\typeof{\ctx}{\app{e_1}{e_2}}{\appres}}
}

\newcommand{\tcif}{
  \inferrule[T-If]
    {\typeof{\ctx}{e_1}{\IfCondR} \\\\
     \typeof{\ifbrctx \ifbrposprop}{e_2}{\IfR} \\
     \typeof{\ifbrctx \ifbrnegprop}{e_3}{\IfR}}
    {\typeof{\ctx}{\If{e_1}{e_2}{e_3}}{\IfR}}
}

\newcommand{\tclet}{
  \inferrule[T-Let]
    {\prop_1 =
       \osubst{(\eqprop{x}{z})}[z][\obj] \\
     \prop_2 =
       \disjprop
         {(\conjprop{(\nottypeprop{x}{\false})}{\posprop})}
         {(\conjprop{(\typeprop{x}{\false})}{\negprop})} \\\\
     \typeof{\ctx}{\expr}{\res[\tyA][\posprop][\negprop][\obj]} \\
     \typeof{\ctx \exta \bindvar \exta \prop_1 \exta \prop_2}{\expr_b}{\R}}
    {\typeof{\ctx}{\Let{\expr}{\expr_b}}{\osubst{\R}[x][\obj]}}
}

\newcommand{\tcsubsume}{
  \inferrule[T-Subsume]
    {\typeof{\ctx}{e}{\synR_1} \\
     \subtype{\ctx}{\synR_1}{\synR_2}}
    {\typeof{\ctx}{e}{\synR_2}}
}

\newcommand{\showtypingrulesfigure}{%
\begin{figure}[ht]
  \begin{mathpar}
  \jbox{\typeof{\ctx}{\expr}{\R}}
  \and
  \tcbaseval
  \and
  \tcprimop
  \and
  \tcvar
  \and
  \tcabs
  \and
  \tcapp
  \and
  \tclet
  \and
  \tcif
  \and
  \tcsubsume
  \end{mathpar}
  \caption{Typing Rules}
  \trlabel{fig:typingrules}
\end{figure}
}

\newcommand{\showbasetypesfigure}{%
\begin{figure}[ht]
  \centering
  \begin{minipage}[t]{0.5\linewidth}
    \begin{align*}
      \typeofbv{\nat} &\mdef \res[\Nat][\TT][\FF][\topobj] \\
      \typeofbv{\str} &\mdef \res[\Str][\TT][\FF][\topobj] \\
      \typeofbv{\true} &\mdef \res[\true][\TT][\FF][\topobj]  \\
      \typeofbv{\false} &\mdef \res[\false][\FF][\TT][\topobj]
    \end{align*}
  \end{minipage}%
  \begin{minipage}[t]{0.5\linewidth}
    \begin{align*}
      \typeofop{\isnatop} &\mdef \predfun{\Nat} \\
      \typeofop{\isstrop} &\mdef \predfun{\Str} \\
      \typeofop{\notop} &\mdef \predfun{\false} \\
      \typeofop{\addop} &\mdef \simpfunty{\Nat}{\Nat}
    \end{align*}
  \end{minipage}
  \caption{Types of Base Values and Primitive Operations}
  \trlabel{fig:basetypes}
\end{figure}
}

\newcommand{\showsubtypingfigure}{%
\begin{figure}[ht]
  \begin{mathpar}
  \jbox{\subtype{\ctx}{\tyA}{\tyB}}
  \and
  \inferrule[S-Top]
    {\wftype{\tyA}}
    {\subtype{\ctx}{\tyA}{\top}}
  \and
  \inferrule[S-Bot]
    {\wftype{\tyA}}
    {\subtype{\ctx}{\bot}{\tyA}}
  \and
  \inferrule[S-Refl]
    {\wftype{\tyA}}
    {\subtype{\ctx}{\tyA}{\tyA}}
  \and
  \inferrule[S-Trans]
    {\subtype{\ctx}{\tyA}{\tyB} \\
     \subtype{\ctx}{\tyB}{\tyC}}
    {\subtype{\ctx}{\tyA}{\tyC}}
  \and
  \inferrule[S-Fun]
    {\subtype{\ctx}{\tyA}{\tyB} \\
     \subtype{\ctxa}{\R_b}{\R_a}}
    {\subtype{\ctx}{\funty{x}{\tyB}{\R_b}}{\funty{x}{\tyA}{\R_a}}}
  \and
  \inferrule[S-UnionSub]
    {\subtype{\ctx}{\tyB}{\tyA} \\
     \subtype{\ctx}{\tyC}{\tyA}}
    {\subtype{\ctx}{\tyB \union \tyC}{\tyA}}
  \and
  \inferrule[S-UnionSuper1]
    {\subtype{\ctx}{\tyA}{\tyB} \\
     \wftype{\tyC}}
    {\subtype{\ctx}{\tyA}{\tyB \union \tyC}}
  \and
  \inferrule[S-UnionSuper2]
    {\subtype{\ctx}{\tyA}{\tyC} \\
     \wftype{\tyB}}
    {\subtype{\ctx}{\tyA}{\tyB \union \tyC}}
  \and
  \inferrule[S-Expl]
    {\proves{\ctx}{\FF} \\
     \wftype{\tyA} \\
     \wftype{\tyB}}
    {\subtype{\ctx}{\tyA}{\tyB}}
  \end{mathpar}
  \lighthrule
  \begin{mathpar}
  \jbox{\subtype{\ctx}{\obj_1}{\obj_2}}
  \and
  \inferrule[SO-Equiv]
    {\proves{\ctx}{\eqprop{x}{y}}}
    {\subtype{\ctx}{x}{y}}
  \and
  \inferrule[SO-Bot]
    {\wfobj{\obj}}
    {\subtype{\ctx}{\botobj}{\obj}}
  \and
  \inferrule[SO-Top]
    {\wfobj{\obj}}
    {\subtype{\ctx}{\obj}{\topobj}}
  \and
  \inferrule[SO-Expl]
    {\proves{\ctx}{\FF} \\
     \wfobj{\obj_1} \\
     \wfobj{\obj_2}}
    {\subtype{\ctx}{\obj_1}{\obj_2}}
  \end{mathpar}
  \lighthrule
  \begin{mathpar}
  \jbox{\subtype{\ctx}{\R_1}{\R_2}}
  \and
  \inferrule[SR-Result]
    {\subtype{\ctx}{\tyA_1}{\tyA_2} \\
     \proves{\ctx \extb \posprop[1]}{\posprop[2]} \\
     \proves{\ctx \extb \negprop[1]}{\negprop[2]} \\
     \subtype{\ctx}{\obj_1}{\obj_2}}
    {\subtype{\ctx}{\res[\tyA_1][\posprop[1]][\negprop[1]][\obj_1]}{\res[\tyA_2][\posprop[2]][\negprop[2]][\obj_2]}}
  \end{mathpar}
  \caption{Subtyping Rules}
  \trlabel{fig:subtyping}
\end{figure}
}

\newcommand{\showobjectsubstitutionfigure}{%
\begin{figure}[ht]
  \[
    \osubst{\res[\tyA][\posprop][\negprop][\obj_r]}[x][\obj] =
    \res%
      [\osubst{\tyA}[x][\obj]]%
      [\osubst{\posprop}[x][\obj]]%
      [\osubst{\negprop}[x][\obj]]%
      [\osubst{\obj_r}[x][\obj]]
  \]
  \lighthrule*
  \begin{minipage}[t]{0.67\textwidth}
  \begin{align*}
    \osubst{\tyB} &=
      \tyB \quad \metaif B\in\{\top,\bot,\true,\false,\Nat,\Str\}
    \\
    \osubst{(\uniontype{\tyA_1}{\tyA_2})} &=
      \uniontype{\osubst{\tyA_1}}{\osubst{\tyA_2}}
    \\
    \poserase{(\funty{y}{\tyA}{\R})} &=
      \funty{y}{\osubst{\tyA}[x][\highlightmathbox*{\botobj}]}{\poserase{\synR}} \quad \metaif x \neq y
    \\
    \negerase{(\funty{y}{\tyA}{\R})} &=
      \funty{y}{\osubst{\tyA}[x][\highlightmathbox*{\topobj}]}{\negerase{\synR}} \quad \metaif x \neq y
    \\
    \osubst{(\funty{x}{\tyA}{\R})} &=
      \funty{x}{\osubst{\tyA}}{\synR}
    \\
    \osubst{\tyA}[x][y] &=
      \substR{\tyA}{x}{y} \quad \metaotherwise
  \end{align*}
  \end{minipage}%
  \begin{minipage}[t]{0.33\textwidth}
    \begin{align*}
      \osubst{\bot} &= \bot \\
      \osubst{\top} &= \top \\
      \osubst{x} &= \obj \\
      \osubst{y} &= y \quad \metaif y \neq x \\
      \\
      \osubst{\prop}[x][y] &= \substR{\prop}{x}{y}
    \end{align*}
  \end{minipage}
  \lighthrule
  \begin{minipage}[t]{0.5\textwidth}
  \begin{align*}
    \poserase{\TT} &= \TT \\
    \poserase{\FF} &= \FF \\
    \poserase{(\typeprop{x}{\tyB})} &= \highlightmathbox{\TT} \\
    \poserase{(\typeprop{z}{\tyB})} &= \typeprop{z}{\poserase{\tyB}} \quad \metaif z \neq x \\
    \poserase{(\nottypeprop{x}{\tyB})} &= \highlightmathbox{\TT} \\
    \poserase{(\nottypeprop{z}{\tyB})} &=
      \nottypeprop{z}{\osubst{\tyB}[x][\highlightmathbox*{\botobj}]} \quad \metaif z \neq x \\
    \poserase{(\eqprop{y}{z})} &= \highlightmathbox{\TT} \qquad \metaif x\in\{y,z\} \\
    \poserase{(\eqprop{y}{z})} &= \eqprop{y}{z} \qquad \metaif x\not\in\{y,z\} \\
    \poserase{(\conjprop{\prop_1}{\prop_2})} &= \conjprop{\poserase{\prop_1}}{\poserase{\prop_2}} \\
    \poserase{(\disjprop{\prop_1}{\prop_2})} &= \disjprop{\poserase{\prop_1}}{\poserase{\prop_2}}
  \end{align*}
  \end{minipage}%
  \begin{minipage}[t]{0.5\textwidth}
  \begin{align*}
    \negerase{\TT} &= \TT \\
    \negerase{\FF} &= \FF \\
    \negerase{(\typeprop{x}{\tyB})} &= \highlightmathbox{\FF} \\
    \negerase{(\typeprop{z}{\tyB})} &= \typeprop{z}{\negerase{\tyB}} \quad \metaif z \neq x \\
    \negerase{(\nottypeprop{x}{\tyB})} &= \highlightmathbox{\FF} \\
    \negerase{(\nottypeprop{z}{\tyB})} &=
      \nottypeprop{z}{\osubst{\tyB}[x][\highlightmathbox*{\topobj}]} \quad \metaif z \neq x \\
    \negerase{(\eqprop{y}{z})} &= \highlightmathbox{\FF} \qquad \metaif x\in\{y,z\} \\
    \negerase{(\eqprop{y}{z})} &= \eqprop{y}{z} \qquad \metaif x\not\in\{y,z\} \\
    \negerase{(\conjprop{\prop_1}{\prop_2})} &= \conjprop{\negerase{\prop_1}}{\negerase{\prop_2}} \\
    \negerase{(\disjprop{\prop_1}{\prop_2})} &= \disjprop{\negerase{\prop_1}}{\negerase{\prop_2}}
  \end{align*}
  \end{minipage}
  \caption{Object Substitution/Erasure}
  \trlabel{fig:objsubst}
\end{figure}
}

\newcommand{\showruntimesyntaxfigure}{%
\begin{figure}[ht]
  \centering
  \begin{tabular}{>{\itshape}l<{} >{$}r<{$}@{\hskip 0.10cm} >{$}l<{$}@{\hskip 0.5cm}}
    Runtime Expressions & \unexpr \sdef &
    \bv \mid \op \mid x \mid
    \highlightmathbox{\unabs} \mid \ap{\unexpr}{\unexpr} \mid \If{\unexpr}{\unexpr}{\unexpr} \mid \Let{\unexpr}{\unexpr}
    \\
    Values & u,v \sdef &
    \bv \mid \op \mid \unabs
    \\
    Evaluation Contexts & \ectx \sdef &
    \ehole \mid \ap{\ectx}{\unexpr} \mid \ap{v}{\ectx} \mid \If{\ectx}{\unexpr}{\unexpr}
    \mid \Let{E}{\unexpr}
    \\
    Value Environments & \env \sdef &
    \empenv \mid \extenv{\env}{x}{v}
  \end{tabular}
  \caption{Runtime Syntax}
  \trlabel{fig:runtime-stx}
\end{figure}
}

\newcommand{\showerasefigure}{%
\begin{figure}[ht]
  \begin{align*}
    \eraseann{\_} \colon \mathsf{Expression} &\rightarrow \mathsf{Unannotated\ Expression} \\
    \eraseann{x} &\mdef x \\
    \eraseann{\op} &\mdef \op \\
    \eraseann{\bv} &\mdef \bv \\
    \eraseann{\abs{x}{\expr}} &\mdef \unabs[x][{\eraseann{\expr}}] \\
    \eraseann{\app{\expr_1}{\expr_2}} &\mdef {\app{\eraseann{\expr_1}}{\eraseann{\expr_2}}} \\
    \eraseann{\If{\expr_1}{\expr_2}{\expr_3}} &\mdef {\If{\eraseann{\expr_1}}{\eraseann{\expr_2}}{\eraseann{\expr_3}}} \\
    \eraseann{\Let{\expr_1}{\expr_2}} &\mdef {\Let{\eraseann{\expr_1}}{\eraseann{\expr_2}}}
  \end{align*}
  \caption{Projection from syntax to runtime syntax}
  \trlabel{fig:erase}
\end{figure}
}

\newcommand{\showexecopfigure}{%
\begin{figure}[h]
  \begin{align*}
    \execop{\isnatop}{\nat} &\mdef \true \\
    \execop{\isnatop}{v} &\mdef \false
      \quad \metaif v \notveq \nat \\
    \execop{\isstrop}{\str} &\mdef \true \\
    \execop{\isstrop}{v} &\mdef \false
      \quad \metaif v \notveq \str \\
    \execop{\notop}{\false} &\mdef \true \\
    \execop{\notop}{v} &\mdef \false
      \quad \metaif v \notveq \false \\
    \execop{\addop}{\nat} &\mdef \nat + 1
  \end{align*}
  \caption{Primitive operation interpretation}
  \trlabel{fig:execop}
\end{figure}
}

\newcommand{\showreductionfigure}{%
\begin{figure}[h]
  \begin{align*}
    \fbox{$\stepsto{\unexpr}{\unexpr}$}
    \hspace{6em}
    \stepsto{\ap{\op}{v}&}{\execop{\op}{v}} &\textsc{(E-PrimOp)}\\
    \stepsto{\ap{(\unabs[x][\unexpr_b])}{v}&}{\substexpr{\unexpr_b}{x}{v}} &\textsc{(E-Beta)}\\
    \stepsto{\If{v}{\unexpr_1}{\unexpr_2}&}{\unexpr_1} \quad \metaif v \neq \false &\textsc{(E-IfTrue)}\\
    \stepsto{\If{\false}{\unexpr_1}{\unexpr_2}&}{\unexpr_2} &\textsc{(E-IfFalse)} \\
    \stepsto{\Let{v}{\unexpr_2}&}{\substexpr{\unexpr_2}{x}{v}} &\textsc{(E-Let)}
  \end{align*}
  \begin{mathpar}
  \jbox{\eestepsto{\unexpr}{\unexpr}}
  \and
  \inferrule
    {\stepsto{\unexpr_1}{\unexpr_2}}
    {\estepsto{\unexpr_1}{\unexpr_2}}
  \end{mathpar}
  \caption{Reduction rules}
  \trlabel{fig:newrules}
\end{figure}
}

\newcommand{\showfullopsemfigure}{%
\begin{figure}[t]
  \begin{center}
  \begin{tabular}{>{\itshape}l<{} >{$}r<{$}@{\hskip 0.10cm} >{$}l<{$}@{\hskip 0.5cm}}
    Runtime Expressions & \unexpr \sdef &
    \bv \mid \op \mid x \mid
    \highlightmathbox{\unabs} \mid \ap{\unexpr}{\unexpr} \mid \If{\unexpr}{\unexpr}{\unexpr} \mid \Let{\unexpr}{\unexpr}
    \\
    Values & u,v \sdef &
    \bv \mid \op \mid \unabs
    \\
    Evaluation Contexts & \ectx \sdef &
    \ehole \mid \ap{\ectx}{\unexpr} \mid \ap{v}{\ectx} \mid \If{\ectx}{\unexpr}{\unexpr}
    \mid \Let{E}{\unexpr}
    \\
  \end{tabular}
  \end{center}

  \lighthrule

  \begin{minipage}[t]{0.5\textwidth}
  \begin{mathpar}
  \inferrule
    {\stepsto{\unexpr_1}{\unexpr_2}}
    {\estepsto{\unexpr_1}{\unexpr_2}}
  \end{mathpar}
  \begin{align*}
    \stepsto{\ap{\op}{v}&}{\execop{\op}{v}}
    \\
    \stepsto{\ap{(\unabs[x][\unexpr_b])}{v}&}{\substexpr{\unexpr_b}{x}{v}}
    \\
    \stepsto{\If{v}{\unexpr_1}{\unexpr_2}&}{\unexpr_1} \quad \metaif v \neq \false
    \\
    \stepsto{\If{\false}{\unexpr_1}{\unexpr_2}&}{\unexpr_2}
    \\
    \stepsto{\Let{v}{\unexpr_2}&}{\substexpr{\unexpr_2}{x}{v}}
  \end{align*}
  \end{minipage}%
  \begin{minipage}[t]{0.5\textwidth}
  \begin{align*}
    \execop{\isnatop}{\nat} &\mdef \true \\
    \execop{\isnatop}{v} &\mdef \false
      \quad \metaif v \notveq \nat \\
    \execop{\isstrop}{\str} &\mdef \true \\
    \execop{\isstrop}{v} &\mdef \false
      \quad \metaif v \notveq \str \\
    \execop{\notop}{\false} &\mdef \true \\
    \execop{\notop}{v} &\mdef \false
      \quad \metaif v \notveq \false \\
    \execop{\addop}{\nat} &\mdef \nat + 1
  \end{align*}
  \end{minipage}

  \caption{Operational Semantics}
  \trlabel{fig:opsem-full}
\end{figure}
}

\newcommand{\showexpinterp}{%
  \begin{align*}
    \expinterp{\semR}{\unexpr}[0] \mdef& \metatrue \\
    \expinterp{\semR}{\unexpr_0}[n+1] \mdef& \begin{cases}
      \expinterp{\semR}{\unexpr_1}[n] \, &\metaif \eestepsto{\unexpr_0}{\unexpr_1} \\
      \valinterp{\semR}{v}[n+1] \, &\metaif \unexpr_0 = v
    \end{cases}
  \end{align*}
}

\newcommand{\showvalresinterp}{%
  \begin{align*}
    \valinterpres{\R}{v}[0] \mdef& \metatrue \\
    \valinterpres{{\res}}{v}[n+1] \mdef& (\propertyA{\obj}{v}) \metaand \\
                                 &\valinterp{\tyA}{v}[n+1] \metaand \\
                                 &\left(\begin{gathered}
                                     \left(\propertyCfst{v}{}{\env}{n+1}\right)\\
                                     \metaor\\
                                     \left(\propertyCsnd{v}{}{\env}{n+1}\right)
                                 \end{gathered}\right)
  \end{align*}
}

\newcommand{\showtyperesultinterpfigs}{%
\begin{figure}[h]
  \small
  \begin{minipage}[t]{0.40\textwidth}
    \showexpinterp
  \end{minipage}
  \hfill
  \begin{minipage}[t]{0.54\textwidth}
    \showvalresinterp
  \end{minipage}
  \caption{Interpretations of Type Results}
  \trlabel{fig:tresinterp}
\end{figure}
}

\newcommand{\showvaltyinterpfig}{%
\begin{figure}[h]
  \begin{align*}
    \valinterp{\tyA}{v}[0] \mdef& \metatrue \\
    \valinterp{\top}{v}[n+1] \mdef& \metatrue \\
    \valinterp{\bot}{v}[n+1] \mdef& \metafalse \\
    \valinterp{\False}{v}[n+1] \mdef& v \veq \false\\
    \valinterp{\True}{v}[n+1] \mdef& v \veq \true \\
    \valinterp{\Nat}{v}[n+1] \mdef& v \in \natset \\
    \valinterp{\Str}{v}[n+1] \mdef& v \in \strset \\
    \valinterp{\uniontype{\tyA}{\tyB}}{v}[n+1] \mdef& \valinterp{\tyA}{v}[n+1] \metaor \valinterp{\tyB}{v}[n+1] \\
    \valinterp{\funty{x}{\tyA}{\semR}}{v}[n+1] \mdef&
     \left(v =_{{\alpha}} \unabs \metaor v =_{\alpha} p \right) \metaand
      \\ & (\forall w, \betw{k}[0][n+1],\: \valinterp{\tyA}{w}[k] \metaimpl
            \expinterp{\R}{\ap{v}{w}}[k][\extenv{\env}{x}{w}])
  \end{align*}
  \caption{Interpretation of Types}
  \trlabel{fig:valinterptype}
\end{figure}
}

\newcommand{\showpropinterpfig}{%
\begin{figure}[h]
  \begin{align*}
    \propinterp{\prop}[0] \mdef& \metatrue \\
    \propinterp{\TT}[n+1] \mdef& \metatrue \\
    \propinterp{\FF}[n+1] \mdef& \metafalse \\
    \propinterp{\eqprop{x}{y}}[n+1] \mdef& \env(x) =_{\alpha} \env(y) \\
    \propinterp{\typeprop{x}{\tyA}}[n+1] \mdef& \valinterp{\tyA}{\env(x)}[n+1] \\
    \propinterp{\nottypeprop{x}{\tyA}}[n+1] \mdef& \notvalinterp{\tyA}{\env(x)}[m][n+1]\\
    \propinterp{\conjprop{\prop_1}{\prop_2}}[n+1] \mdef& \propinterp{\prop_1}[n+1] \metaand \propinterp{\prop_2}[n+1] \\
    \propinterp{\disjprop{\prop_1}{\prop_2}}[n+1] \mdef& \propinterp{\prop_1}[n+1] \metaor \propinterp{\prop_2}[n+1]
  \end{align*}
  \caption{Interpretation of Propositions}
  \trlabel{fig:propinterp}
\end{figure}
}
\newcommand{\showctxinterpfig}{%
  \begin{figure}[h]
  \begin{align*}
    \ctxinterp{\ctx}{\env}[0] \mdef& \dom{\ctx} \subseteq \dom{\env} \\
    \ctxinterp{\empctx}{\env}[n+1] \mdef& \metatrue \\
    \ctxinterp{\ctx \exta \bindvar}{\env}[n+1] \mdef& \ctxinterp{\ctx}{\env}[n+1] \metaand \valinterp{\tyA}{\env(x)}[n+1] \\
    \ctxinterp{\ctx \exta \prop}{\env}[n+1] \mdef& \ctxinterp{\ctx}{\env}[n+1] \metaand \propinterp{\prop}[n+1]
  \end{align*}
  \caption{Interpretation of Contexts}
  \trlabel{fig:ctxinterp}
\end{figure}

}
\newcommand{\showlogicalrelationsfigure}{%
  \showtyperesultinterpfigs
  \showvaltyinterpfig
  \showpropinterpfig
 \showctxinterpfig
}

\RequirePackage{todonotes}
\newif\ifLambdaTRHideNotes
\ifdefined\LambdaTRHideNotes
  \LambdaTRHideNotestrue
\fi
\ifLambdaTRHideNotes
  \NewDocumentCommand{\samnote}{m}{}
\else
  \NewDocumentCommand{\samnote}{m}{\todo[inline,color=green!30]{#1}}
\fi

\AtEndPreamble{
  \theoremstyle{acmplain}
  \newtheorem*{lemmatwo}{Lemma 2}
  \theoremstyle{acmdefinition}
  \newtheorem{remark}[theorem]{Remark}
}

\usepackage{epigraph}
\title{Revisiting Soundness for Occurrence Typing, Semantically}
\author{Yuquan Fu}
\orcid{0009-0009-1322-8614}
\email{yuqfu@iu.edu}
\affiliation{
  \institution{Indiana University}
  \city{Bloomington}
  \state{Indiana}
  \country{USA}
}
\author{Carlo Angiuli}
\orcid{0000-0002-9590-3303}
\email{cangiuli@iu.edu}
\affiliation{
  \institution{Indiana University}
  \city{Bloomington}
  \state{Indiana}
  \country{USA}
}
\author{Sam Tobin-Hochstadt}
\orcid{0000-0003-1302-6499}
\email{samth@iu.edu}
\affiliation{
  \institution{Indiana University}
  \city{Bloomington}
  \state{Indiana}
  \country{USA}
}

\begin{document}

\begin{abstract}
Over the past two decades, numerous systems have brought some of the benefits of
dependent typing to a wide variety of new programming languages, often by
restricting which terms can appear inside types.
Such techniques are known as refinement types, occurrence
typing, liquid types, and path dependent types, among others. However, the
restrictions adopted by these systems often break the substitution property,
because they explicitly disallow the ability to substitute arbitrary terms for
variables inside types. This leads to significant complexity in the design and
metatheory of these systems, increasing the possibility of significant errors.

We consider a specific line of work on occurrence typing, namely, the calculus
underlying Typed Racket due to Tobin-Hochstadt and Felleisen
\citep{tobin-hochstadt2010}. We show that the fundamental challenge of
substitution into types resulted in multiple flaws in the formalism and the
syntactic type soundness theorem of this work. These flaws are replicated in
several other papers building on this work, and also surface as a soundness bug
in Typed Racket itself. We identify and repair these problems, revising the core
calculus of Typed Racket and giving a \emph{semantic type soundness} proof using
step-indexed logical relations, formalized in Lean. We argue that this approach
is simpler than it may seem, and easily scales to handle the complexity of the
occurrence typing in Typed Racket.
\end{abstract}

\setlength{\epigraphwidth}{0.8\textwidth}
\setlength{\epigraphrule}{0pt}

\maketitle

\epigraph{Simple type systems have the property that you can replace
  any subexpression of a well-typed term by another of the same
  type!}{\href{https://personal.cis.strath.ac.uk/conor.mcbride/pub/}{Conor McBride}}

\section{Introduction}
\label{sec:introduction}

Over the past two decades, numerous systems have sought to integrate typed and
untyped programming, often with specially-designed type systems aimed at
accommodating the idioms of a related and pre-existing untyped programming
language \citep{siek2006,matthews2007,tobin-hochstadt2008}. Typed Racket is one
such system \citep{tobin-hochstadt2008,tobin-hochstadt2010,kent2016,fu2021},
aimed at integrating with the Racket programming language \citep{plt-tr1}. Typed
Racket uses \emph{occurrence typing} to allow different occurrences of variables
to have different types, based on control flow and type
tests~\cite{tobin-hochstadt2008,tobin-hochstadt2010}, as illustrated by the
following simple example:

\begin{lstlisting}
(: f (-> Any Number))
(define (f x)
  (if (number? x)
      (+ 1 x)
      0))
\end{lstlisting}

Here, the declared type of \texttt{x}, \texttt{Any}, allows \texttt{f} to be
called on any argument. The first occurrence of \texttt{x} in \texttt{f}'s body,
within \texttt{(number? x)}, continues to have type \texttt{Any}. In the
``then'' branch of the conditional, however, the fact that the \texttt{(number?
x)} test was successful refines the type of \texttt{x} to \texttt{Number}, making
\texttt{(+ 1 x)} well-typed.

The key mechanism enabling this is the fact that the types of functions can
include \emph{propositions} about the types of
variables~\cite{tobin-hochstadt2010}. For example, the type of \texttt{number?}
is:
\[
  \texttt{number?} :
  \funty{x}{\top}
  {\res[\Bool][\typeprop{x}{\Nat}][\nottypeprop{x}{\Nat}][\topobj]}
\]
That is, \texttt{number?} takes any value, dubbed $x$, and returns a boolean (written $\Bool$),
but its result type also records two \emph{latent propositions}: if the call
returns true, then the argument $x$ has type $\Nat$ (a number), and if it returns false,
$x$ does not have type $\Nat$.\footnote{The final component of the type,
$\topobj$, is called an \emph{object}, and is not needed for this example. We
will discuss objects in \cref{sec:background}.}

Because types in Typed Racket may contain terms such as $x$, Typed Racket is in
some sense \emph{dependently typed}; however, this dependency comes with two
major restrictions: only \emph{variables} can occur in types, and only in the
context of asserting \emph{typing propositions} about those variables. However,
these variables may come from scopes other than the innermost scope
\cite{tobin-hochstadt2010}. For example, the following curried function tests
the type of its outer parameter \texttt{x}, not its inner parameter \texttt{y}:

\begin{lstlisting}
(lambda ([x : Any])
  (lambda ([y : Any])
    (number? x)))
\end{lstlisting}

The inner function is type-checked in a context that binds \texttt{x}, and is
assigned a type whose latent propositions refer to the free variable $x$, rather
than the variable $y$ bound in the type itself:
\[
  \funty{y}{\top}
  {\res[\Bool][\typeprop{x}{\Nat}][\nottypeprop{x}{\Nat}][\topobj]}
\]
As a result, types are context-sensitive: they can only be understood relative
to a typing or runtime context providing an interpretation of their free
variables.

When type dependency is combined with restrictions on what terms are allowed to
appear in types, some basic properties of type systems begin to fail. In
particular, substitution of values for variables (as in call-by-value
$\beta$-reduction) no longer preserves typing, or even the well-formedness of
types, because the type system handles variables in a special way that does not
extend to arbitrary terms or even values. The preservation property also fails,
as it generally relies on substitution.

This failure of preservation is a well-known problem in the programming languages community, and is
often addressed by adding runtime-only typing rules as a technical device, as
famously exemplified by the ``stupid cast'' rule of Featherweight Java
\citep{igarashi2001}. A similar approach is adopted in the original POPL 2008
paper on occurrence typing for Typed Racket, in order to address problems with
typing partially-reduced \texttt{if} expressions~\citep{tobin-hochstadt2008}.
The subsequent ICFP 2010 paper \citep{tobin-hochstadt2010} on Typed Racket, however,
introduces a more sophisticated type system with dependency, as shown above, which cannot be handled in this way,
and instead adopts a big-step operational semantics to avoid needing to discuss
properties of intermediate states whatsoever.

Adopting a big-step operational semantics allows us to prove type soundness at
base type without worrying about ill-typed intermediate states, but does not
solve the problem at function type: $\lambda$-expressions are values that may
still contain ill-typed subexpressions due to substitutions into their bodies.
The ICFP 2010 paper also switches from substitution to closures for this reason,
but still suffers from problems related to the context-sensitivity of types.
Consider its typing rule for closures \citep[Section 6.2]{tobin-hochstadt2010}:
\[
	\inferrule[\textsc{T-Clos}]
	{\exists \ctx.\; \satp{\env}{\ctx} \text{ and }
		\typeof{\ctx}{\abs{x}[\tyA]{\expr}}
		{\res[\tyB][\posprop][\negprop][\obj]}}
	{\typeof{}
		{\langle \env,\abs{x}[\tyA]{\expr}\rangle}
		{\highlightmathbox{\res[\tyB][\posprop][\negprop][\obj]}}}
\]
The conclusion of this rule is ill-scoped, because the conclusion takes place in
an empty context, whereas the type of the conclusion may contain any of the
variables in $\Gamma$, as the premise suggests.

The fundamental issue is that, although \citet{tobin-hochstadt2010} understood
that a more semantic approach to type soundness is needed for their terms, for
which they provide a big-step operational semantics $\env \vdash \expr
\Downarrow v$, and for their propositions, for which they provide a
model-theoretic interpretation relative to a runtime environment
$\satp{\env}{\prop}$, they did not provide, or indeed recognize the need for, a semantic interpretation of
types. Instead, the interpretation of types simply falls back on the syntactic
type system, as in the above \textsc{T-Clos} rule.

In fact, even their key type soundness lemma features a similarly ill-scoped
type in its conclusion:

\begin{lemmatwo}[\citet{tobin-hochstadt2010}]
If $\Gamma \vdash e : \tau \mathbin{;} \psi_{+} \mid \psi_{-} \mathbin{;} o$,
$\rho \models \Gamma$, and $\rho \vdash e \Downarrow v$, then [\dots]\
\(
  {}\vdash v :
  \highlightmathbox{\tau} \mathbin{;} \psi'_{+} \mid \psi'_{-} \mathbin{;} o'
\)
for some \(\psi'_{+}\), \(\psi'_{-}\), and \(o'\).
\end{lemmatwo}

These fundamental problems in the metatheory of \citet{tobin-hochstadt2010}
extend to several subsequent papers building on their work, by a variety of
different authors \directworkrefs{}. In all of these papers, even ones adopting
significantly different approaches such as abstract machines, the proofs,
statements, and even definitions are incorrect in similar ways.

We comprehensively address these issues by instead establishing a \emph{semantic
  type soundness} theorem for a core calculus for Typed Racket, using
step-indexed logical relations \citep{appel2001,ahmed2006,dreyer2009} as
advocated e.g.\ by \citet{timany2024}. The resulting metatheory is quite
straightforward and has been both proven on paper and formalized in
Lean~\cite{demoura2015}, making us confident that it is finally correct. Our
core calculus is a reformulation of the type system introduced by
\citet{tobin-hochstadt2010}, which clarifies its scoping invariants and repairs
several other problems we discovered along the way, including a non-well-founded
definition and a genuine unsoundness in the Typed Racket implementation.

More broadly, after developing the key syntactic elements of these kinds of type
systems, researchers frequently run into difficulties validating those type
systems metatheoretically, and have devised many complex new approaches for
addressing these metatheoretic challenges
\citep{tobin-hochstadt2010,castagna2022,flanagan2006,knowles2010,chugh2012}. In
many cases these approaches work, but in the course of our work we have
discovered that many papers in this area---and \emph{not only} those
\directworkrefs{} building on the work of
\citet{tobin-hochstadt2010}---independently contain issues related to the ones
at the heart of this work \cite{flanagan2006,knowles2010,rondon2008,vekris2015}.
These include papers published at POPL, ICFP, PLDI, ESOP, and ECOOP.

Our conclusion is twofold: these complexities can be more comprehensively and
straightforwardly addressed by the use of semantic type soundness techniques,
and by paying closer attention to variable scoping in types. These papers often
involve systems with little complexity in the runtime behavior of the programs,
meaning that the definitions of their logical relations can be extremely
straightforward, and their proofs not more challenging than traditional
progress-and-preservation proofs \citep{wright1994}.

\paragraph{Contributions}

Our work has three primary contributions:

\begin{enumerate}
\item A detailed analysis of why proving type soundness for Typed Racket was
challenging, and why it was done incorrectly, including the role of type
dependency in that challenge;
\item A corrected version of the core occurrence typing calculus for Typed
Racket and its metatheory, including a semantic proof of type soundness
mechanized in Lean; and
\item An enumeration of how similar difficulties affect other papers on related
languages, both those that descend from Typed Racket and others.
\end{enumerate}

\paragraph{Outline}

The remainder of this paper is organized as follows.
In \cref{sec:background}, we review the basics of Typed Racket and its core
calculus $\lambda_{TR}$, how its metatheory was developed in the ICFP 2010 paper
\citep{tobin-hochstadt2010}, and what was incorrect. We also outline how similar
problems reoccur in subsequent work building on and incorporating this
formalism.
In \cref{sec:formalism}, we present our revised core calculus, $\lamtrl$,
including its proof system, type system, and small-step operational semantics,
along with a refined analysis of its variable scoping.
In \cref{sec:logical-relations-soundness}, we define step-indexed logical
relations for $\lamtrl$ and outline the statements and proofs of semantic type
soundness.
In \cref{sec:related_work}, we analyze a broader range of related systems in
greater detail.
We conclude in \cref{sec:conclusion}.

\section{Why Occurrence Typing Needs More From Type Soundness}
\label{sec:background}

We begin by recalling how Typed Racket's static type system works in practice,
how its core calculus is defined, and how occurrence typing allows Typed Racket
to type check idiomatic Racket programs by combining logical reasoning about
type-based propositions with a restricted form of dependent typing. We then
describe the metatheoretic approach taken by \citet{tobin-hochstadt2010}, how it
attempts to take a more semantic approach to proving type soundness but falls
short, and the ways in which it is ultimately incorrect. Finally, we outline how
subsequent papers building on that work inherited its flaws, even when defining
the system in a substantially different way.

\subsection{A Brief Introduction to Typed Racket}

To understand how Typed Racket's type system works, we will begin with an
example that includes all of the concepts we will need for this paper.

\begin{lstlisting}
(: number-or-boolean? (-> Any Boolean :
                          #:+ ((U Number Boolean) @ 0)
                          #:- (! (U Number Boolean) @ 0)))
(define (number-or-boolean? x)
  (or (number? x) (boolean? x)))

(: size (-> (U Number String) Number))
(define (size z)
  (if (number-or-boolean? z)
      (log z)
      (string-length z)))
\end{lstlisting}

The \texttt{size} function processes a number or string using the \texttt{log}
or \texttt{string-length} functions respectively. It
relies on a type test using the predicate \texttt{number-or-boolean?} to
determine which function to use. This predicate is defined in the program; the
type system does not require predicates to be built-in. The notation \texttt{(U
Number Boolean)} is an untagged union type; it represents the set of values that
are either numbers or booleans. The function type notation \texttt{(-> Any
Boolean : \#:+ P \#:- Q)} represents the type of the predicate function, which
takes any value and returns a boolean. If that boolean is \emph{true}, then we
learn the truth of the proposition \texttt{P}, that the value is a number or a
boolean; if it is \emph{false}, then we learn the truth of the proposition
\texttt{Q}, that the value is neither a number nor a boolean. In these
propositions, \texttt{@ 0} refers to the predicate's parameter \texttt{x}.

The key rules for type checking this program are those for function application
and \texttt{if} expressions, reproduced below:\footnote{These rules are taken
from \cref{sec:formalism} of this paper but are effectively identical to the
earlier rules \citep{tobin-hochstadt2010}.}
\begin{mathpar}
  \tcapp \and \tcif
\end{mathpar}

To understand these rules, first note that each typing judgment not only
specifies the type of a term, but also two \emph{propositions} and an
\emph{object}. The two propositions are statements that are true if the term
evaluates to true (resp., false) and the object describes a \emph{variable} that
the term is equal to. For example, the variable \texttt{x} has object $x$,
whereas an expression such as \texttt{(+ x 1)} has a trivial object denoted by
$\top$, because there is no variable that this expression is equal to.
Propositions include statements like $\typeprop{x}{\Nat}$ (which is true if $x$
is a number) as well as logical connectives and equality.

With this in mind, the \textsc{T-App} rule types the operand $e_2$ at an
arbitrary type $B$, and the operator $e_1$ at a function type whose target is a
type result (i.e., a type, two propositions, and an object) $R$. The type of the
application is the result $R$ with the name of its formal parameter, $x$,
replaced by the operand's object. If the operand has a trivial object, then all
information in $R$ about the formal parameter is lost.

The \textsc{T-If} rule types the condition $e_1$ at a type $A$, which is
ignored, and also with two propositions $\psi_{c+}$ and $\psi_{c-}$ describing
what is learned when the condition is true or false, respectively. Those two
propositions are then added to the context $\Gamma$ when typing the two branches
of the \texttt{if}. Thus the context contains not only variable bindings with
their types, but also arbitrary propositions from which more precise type
information can be derived.

Going back now to our example, the application in the condition of the
\texttt{if} expression is typed using the \textsc{T-App} rule. The type of the
predicate states that it takes \texttt{Any} input, with the result type being
\texttt{Bool} and two propositions stating that the argument is (resp., is not)
a number or boolean. Then the name of the formal parameter, $x$, is substituted
with the \emph{object} of the operand, $z$, to produce the result of the
application. As mentioned above, objects are (trivial or) the name of a variable
which is equal to the result of the expression being considered; the variable
$z$ thus has itself as an object for obvious reasons.

The two propositions resulting from the application are therefore
$\typeprop{z}{\uniontype{\Nat}{\Bool}}$ and
$\nottypeprop{z}{\uniontype{\Nat}{\Bool}}$, respectively. From the initial type
declaration for \texttt{z}, we also have $\typeprop{z}{\uniontype{\Nat}{\Str}}$.
These propositions are combined using a proof system (expressed by the judgment
$\Gamma \vdash \psi$) to conclude the more precise proposition
$\typeprop{z}{\Nat}$ (resp., $\typeprop{z}{\Str}$) in each branch, as required
to type check the application of \texttt{log} (resp.,
\texttt{string-length}).

\subsection{The ICFP 2010 Metatheory and its Problems}
\label{ssec:icfp-wrong}

We now turn to the challenges of proving type soundness for the system just
sketched, and why the approach taken by \citet{tobin-hochstadt2010} is
insufficient. The most fundamental issue is that substitution of values (or
terms) for variables does not preserve typing or even well-formedness of types,
because the type system and proof system fundamentally focus on properties of
variables. This challenge appears in multiple ways.

At the most basic level, the definition of the \textsc{T-App} rule above
requires substituting \emph{something} for the formal parameter $x$ in the
result type, but the restrictions on what can appear in types means that we
cannot simply substitute in the operand $e_2$. Instead, we substitute in the
\emph{object of the operand}, which is a variable if the operand is a variable
or equal to a variable, and $\top$ otherwise.

However, the problem of substitution reoccurs when proving type soundness for
the system. First, as recognized already in the earlier POPL 2008 paper on
occurrence typing~\cite{tobin-hochstadt2008}, simply substituting the value
\texttt{7} for the parameter \texttt{z} in the body of the function
\texttt{size} leads to the term

\begin{lstlisting}
(if (number-or-boolean? 7)
    (log 7)
    (string-length 7))
\end{lstlisting}

\noindent
whose else branch is ill-typed, and thus the metatheory must allow for
intermediate terms of this form. This is handled by extra runtime-only typing
rules in the POPL 2008 paper, and by switching to a big-step semantics in the
ICFP 2010 paper.

Second, the correctness of the proof system for propositions relies on the fact
that the propositions are well-formed, which is also not preserved under
substitution. For example, if we substitute \texttt{7} for \texttt{z} in the
proposition $\typeprop{z}{\uniontype{\Nat}{\Bool}}$, we get
$\typeprop{7}{\uniontype{\Nat}{\Bool}}$, which is ill-formed since $7$ is not a
variable. Thus, the correctness of the proof system cannot be characterized by
using substitution. Instead, the ICFP 2010 paper introduces a model-theoretic
account of propositions, in which they are interpreted relative to a runtime
environment, or a mapping from variables to values. This is used to prove the
correctness of the proof system for propositions.

However, although this model-theoretic approach makes progress towards a correct
account of the semantics, it is not sufficient. When characterizing the
model-theoretic truth of the proposition $x \in A$ in environment $\env$, the
ICFP 2010 paper refers back to the syntactic typing judgment $\empctx \vdash
\env(x) : A$, but this resort to syntactic typing reintroduces the same problems
that the semantic account was attempting to resolve.

In particular, defining the semantics with reference to the syntactic typing
judgment requires including a typing rule for closures. This results in
ill-scoped types, because the term \emph{inside} the closure is typed in a
context that includes the free variables being closed over, but the closure
itself must have a type that does \emph{not} include those free variables, as
already discussed in \cref{sec:introduction}. Moving away from closures is no
better, as substituting into the body of a $\lambda$-expression may result in
ill-typed terms, also as previously discussed. And once again, we cannot close
the resulting type by substituting in values for free variables, as types are
not closed under substitution.

The core problem is expressed in the following table:
\citet{tobin-hochstadt2010} develop both semantic and syntactic counterparts for
terms and propositions, but not for types.
\[
  \begin{array}{c|c|c|c}
    & \text{terms} & \text{types} & \text{propositions} \\
    \hline
    \text{syntax} &
      \unabs[x][\expr] &
      \typeof{\ctx}{\expr}{\R} &
      \proves{\ctx}{\prop} \\
    \text{semantics} &
      \env \vdash \expr \Downarrow v &
      \textbf{???} &
      \satp{\env}{\prop}
  \end{array}
\]

\paragraph{Other errors in Typed Racket}

In the course of our work, we have also discovered other errors in
\citet{tobin-hochstadt2010}. Most importantly, the definition of substitution of
objects for variables does not correctly handle contravariance in the domains of
function types. Additionally, the definitions of several metafunctions are
defined formally in terms of typing judgments in a way that is not well-founded.

The following example demonstrates the problem with object substitution.

\begin{lstlisting}
((lambda ([x : Any])
   (lambda ([f : (Any -> Boolean : #:+ (Number @ x) )])
     (if (f x) (add1 x) 1)))
 (ann "" Any))
\end{lstlisting}

The operand of this application is a curried higher-order function of two
arguments, whose second argument $f$ is a predicate that reveals information
about the first argument $x$. If the predicate produces a true value, then $x$
must be a number. This $\lambda$-expression is well-typed. Once it is applied to
the string \texttt{""}, however, the resulting function type can no longer refer
to $x$ since it is no longer in scope. The type system \emph{should} erase the
entire proposition $\typeprop{x}{\Nat}$ to $\FF$, since no possible function
could demonstrate that $x$ is a number when $x$ is in fact the empty string.
However, both the ICFP 2010 paper and the actual Typed Racket implementation
instead erase the proposition to $\TT$ instead, meaning that the type checker
allows applying the above term to the function \texttt{(lambda ([w : Any])
true)}, even though doing so immediately leads to a runtime error.

\subsection{Downstream Consequences of these Errors}

As a result of these problems, the statements of type soundness by
\citet{tobin-hochstadt2010} are not even well-formed, as explained in
\cref{sec:introduction}. In addition, follow-on work adding refinement types to
Typed Racket~\citep{kent2016}, adding dynamic dispatch to Typed
Racket~\citep{fu2021}, and adapting Typed Racket to the Clojure
language~\cite{bonnaire-sergeant2016} all straightforwardly adopt
\citeauthor{tobin-hochstadt2010}'s \citep{tobin-hochstadt2010} approach to
semantics and thus inherit its flaws.

\citet{chang2017} also build on the work of \citet{tobin-hochstadt2010}, but
with significant changes. Rather than using a big-step semantics, they use an
abstract machine based on the CESK machine~\cite{felleisen1987}. To prove type
soundness, they define a typing judgment for abstract machine states, which must
type continuations containing evaluated terms. This runs into very similar
problems as in the earlier paper, as seen in the following typing rule
\cite[Figure~20]{chang2017}:
\[
\inferrule[TM-CESK]
  {\mathcal{E}(\rho,\sigma) \vdash^\pi e : \highlightmathbox*{\tau_2} ; \psi_2^+ \mid \psi_2^- ; o_2 \\
   \highlightmathbox*{\tau_2} ; \psi_2^+ \mid \psi_2^- ; o_2 ; \sigma \vdash_\kappa
     \kappa^\pi : \tau ; \psi^+ \mid \psi^- ; o}
  {\vdash_M
     \langle e, \rho, \sigma, \kappa^\pi \rangle : \tau ; \psi^+ \mid \psi^- ; o}
\]

The first premise types $e$ in a context generated from the environment $\rho$
and store $\sigma$ by a metafunction $\mathcal{E}$, but the resulting type
$\tau_2$ may contain free variables and is thus ill-scoped in the second premise
where no such context is provided.  Additionally, $\mathcal{E}$ has similar
problems to the typing of closures, in that it relies on syntactic typing to
give an account of the meanings of values.

Across all of these papers, the core difficulty is that restricted dependent
types pose two kinds of obligations. First, type dependency necessitates
significantly greater care with scoping and binding. Second, in systems where
dependent types are integrated into general-purpose programming languages with
effects of any kind, not every expression can be substituted into a type. The
resulting restrictions break the substitution and preservation properties,
meaning that the metatheoretic arguments that work for simpler systems
\emph{and} for full-spectrum dependent type theories---both of which satisfy the
substitution property---no longer apply. The situation of needing to develop
novel metatheory while simultaneously needing greater care and precision is a
recipe for mistakes, as we have seen in this section.

In fact, the authors first discovered these problems when attempting to extend
the core calculus and metatheory of \citet{tobin-hochstadt2010}, which led us to
want a more precise account of its scoping. Attempts to build that account---now
in \cref{ssec:formalism:wf-judgments}---led us to realize that the scoping
issues in the \textsc{T-Clos} rule pointed to more fundamental errors.
Rebuilding the core calculus more precisely then led to the discovery of the
other errors described here.

Other papers on restricted dependent types have run into similar challenges, and
have addressed them with different but also complex approaches and, in some
cases, with similar issues; we discuss these in \cref{sec:related_work}.

\section{The \texorpdfstring{$\lamtrl$}{λ\_TR\textasciitilde} Formalism}
\label{sec:formalism}

In this section we present the syntax, type system, and operational semantics of
$\lamtrl$, a core calculus for Typed Racket. The $\lamtrl$ system revises and
corrects the original core calculus behind Typed Racket, $\lambda_{TR}$, which
was first developed by \citet{tobin-hochstadt2010} and later extended to support
refinement types \citep{kent2016}, dynamic dispatch \citep{fu2021}, and various
features of Typed Clojure \citep{bonnaire-sergeant2016}.

In addition to the basic functionality of $\lambda_{TR}$, our presentation
includes the \emph{local binding} extension described by
\citet{tobin-hochstadt2010} and the \emph{equality propositions} introduced by
\citet{kent2016}. As we will discuss in \Cref{sec:conclusion}, we expect that
the other aforementioned extensions can be rebuilt on top of $\lamtrl$ with only
minor modifications, but focus on the core calculus for simplicity.

\subsection{Syntax and Overview}

\showsyntaxfigure

The syntax of $\lamtrl$ is given in \Cref{fig:stx}. The basic \emph{expression}
syntax is very straightforward: we include constants for true $\true$, false
$\false$, numbers (written $\nat$), and strings (written $\str$), as well as
basic primitives, variables, type-annotated $\lambda$-expressions, applications,
conditionals, and $\mathbf{let}$-binding.

The \emph{types} include $\top$ (\texttt{Any} in Typed Racket, a type containing
all values), $\bot$ (\texttt{Nothing}, or the empty type), types containing
exactly the values $\true$ and $\false$, the base types $\Nat$ and $\Str$, union
types $\uniontype{\tyA}{\tyB}$, and function types $\funty{x}{\tyA}{\R}$. In
function types, the \emph{type result} $\R$ is a four-tuple
$\res[\tyB][\posprop][\negprop][o]$ where $\tyB$ is the result type of the
function in the ordinary sense, $\posprop$ and $\negprop$ are propositions that
are guaranteed to hold when the function returns true or false respectively, and
the \emph{object} $o$ tracks whether the result is guaranteed to be equal to the
value of an in-scope variable. (We will discuss objects in more detail in
\cref{ssec:formalism:type-system}.) We will also write $\Bool$ for the type of
booleans, which is defined as $\uniontype{\true}{\false}$.

The \emph{propositions} include standard propositional connectives, the
\emph{positive typing proposition} $\typeprop{x}{\tyA}$ stating that the
variable $x$ has type $\tyA$, the \emph{negative typing proposition}
$\nottypeprop{x}{\tyA}$ stating that $x$ does not have type $\tyA$, and the
\emph{equality proposition} $\eqprop{x}{y}$ stating that the variables $x$ and
$y$ are aliases. Note that $\typeprop{x}{\tyA}$, $\nottypeprop{x}{\tyA}$, and
$\eqprop{x}{y}$ \emph{only apply to variables}, in order to keep the proof
system decidable. For instance, $\nottypeprop{\true}{\Nat}$ is not a valid
proposition.

Finally, \emph{contexts} are lists of two kinds of hypotheses that can be
interleaved: bindings $\bindvar$, which add the variable $x$ into scope, and
propositions $\prop$, which add the hypothesis that $\prop$ is true.

\begin{remark}
For readers familiar with earlier work on $\lambda_{TR}$, we note that our
notations are chosen to follow those of \citet{fu2021}, except that we rename
the ``null object'' $\emptyset$ to $\topobj$. In addition,
\citet{tobin-hochstadt2010} write $\tau_x$ and $\overline{\tau}_x$ for the
positive and negative typing propositions respectively, and $x{:}\sigma
\xrightarrow[o]{\psi_+\mid\psi_-} \tau$ for the function type
$\funty{x}{\sigma}{\res[\tau]}$.
\end{remark}

\begin{figure}[ht]
\begin{center}
\begin{tikzpicture}[
  judgment/.style={draw, rounded corners=2pt, align=center, inner sep=4pt, outer sep=3pt},
  dependency/.style={->, thick}
]
  \node[judgment] (wf) {
    \begin{tabular}{@{}c@{}}
      \underline{\textbf{Well-formedness}} \\[2pt]
      $\wfctx$ \\
      $\wfobj{\obj}$ \\
      $\wfprop{\prop}$ \\
      $\wftype{\tyA}$ \\
      $\wfres{\R}$
    \end{tabular}
  };
  \node[judgment, right=0.6cm of wf] (proof) {
    \begin{tabular}{@{}c@{}}
      \underline{\textbf{Proof/subtyping systems}} \\[2pt]
      $\proves{\ctx}{\prop}$ \\
      $\subtype{\ctx}{\tyA}{\tyB}$ \\
      $\subtype{\ctx}{\obj_1}{\obj_2}$ \\
      $\subtype{\ctx}{\R_1}{\R_2}$ \\
      $\restrictrel{\ctx}{\tyA}{\tyB}{\tyC}$ \\
      $\removerel{\ctx}{\tyA}{\tyB}{\tyC}$
    \end{tabular}
  };
  \node[judgment, right=0.6cm of proof] (typing) {
    \underline{\textbf{Typing judgment}} \\[2pt]
    $\typeof{\ctx}{\expr}{\R}$
  };
  \draw[dependency] (wf) -- (proof);
  \draw[dependency] (proof) -- (typing);
\end{tikzpicture}
\end{center}
\caption{Layers of $\lamtrl$}
\label{fig:formalism-layers}
\end{figure}

There are many judgments in $\lamtrl$, which we separate into three distinct
layers as depicted in \Cref{fig:formalism-layers}. In
\Cref{ssec:formalism:wf-judgments}, we will describe the well-formedness
judgments, which clarify and correct some variable scoping issues in prior work.
In \Cref{ssec:formalism:proof-system}, we turn our attention to the most
complicated layer, the proof and subtyping system, which uses propositional
assumptions to derive subtyping relationships that hold in the current scope. In
\Cref{ssec:formalism:type-system}, we finally use the previous two layers to
define the main typing judgment, which is perhaps surprisingly straightforward
once the proof and subtyping system is in place.

\subsection{Well-Formedness and Scoping}
\label{ssec:formalism:wf-judgments}

As outlined in \cref{sec:background}, many of the problems in prior work can be
traced back to ill-scoped operations on types with free variables, including the
\textsc{T-Clos} rule of \citet{tobin-hochstadt2010} and the \textsc{TM-CESK}
rule of \citet{chang2017}. When trying to clarify the correct variable scoping
conditions, we realized that this problem is complicated by $\lambda_{TR}$ and
its descendants using a single mechanism---typing propositions---to capture both
variable scoping and typing.

We illustrate this phenomenon with the expression
$\abs{x}[\tyA]{\If{\app{\isnatop}{x}}{e}{e'}}$. To type the body of this
$\lambda$-abstraction in $\lambda_{TR}$, we extend the context with
$\typeprop{x}{\tyA}$. In the `then' branch we have learned that $x$ is a number,
so $e$ is typed in context $\typeprop{x}{\tyA},\typeprop{x}{\Nat}$, recording
that the variable $x$ has both type $\tyA$ (by declaration) and type $\Nat$ (by
occurrence typing). Suppose now that $e = \abs{x}[\tyB]{e''}$; then $e''$ is
typed in context $\typeprop{x}{\tyA},\typeprop{x}{\Nat},\typeprop{x}{\tyB}$,
where the third $x$ is a \emph{new} $x$ shadowing the previous two.

We emphasize that this is \emph{not} just the standard problem of variable
shadowing. We cannot prohibit $\Gamma$ from containing multiple entries for $x$,
because that is the central mechanism behind occurrence typing in
$\lambda_{TR}$; nor can we assert that later $x$s shadow earlier ones, for the
same reason! To put a finer point on the problem, it is quite unclear how to
formulate $\lambda_{TR}$ or related systems with a de Bruijn representation of
variables, e.g., in order to formalize their metatheory.

\showwffigure

We distinguish two kinds of context extension in $\lamtrl$: descending under a
binder $\abs{x}[\tyA]{\dots}$ extends the context by $\bindvar$, whereas
learning information about an existing variable $x$ extends the context by a
typing proposition $\typeprop{x}{\tyA}$. Armed with this distinction, it is
straightforward to mutually define the \emph{well-formedness} (or
well-scopedness) judgments $\wfctx$, $\wfobj{\obj}$, $\wfprop{\prop}$,
$\wftype{\tyA}$, and $\wfres{\R}$ for our type system in
\Cref{fig:wf-judgments}, where $\dom{\ctx}$ is the set of variables \emph{bound}
in $\ctx$:
\begin{align*}
  \dom{\empctx} &\mdef \varnothing \\
  \dom{\ctx \exta \bindvar} &\mdef \dom{\ctx} \cup \{x\} \\
  \dom{\ctx \extb \prop} &\mdef \dom{\ctx}
\end{align*}

There are three rules governing well-formedness of contexts ($\wfctx$): the
empty context $\empctx$ is well-formed; the context extension $\ctx \cdot
\bindvar$ is well-formed when $\ctx$ is well-formed and $\tyA$ is well-formed in
the context $\ctx$ (i.e., $\wftype{\tyA}$); and the context extension $\ctx
\cdot \prop$ is well-formed when $\ctx$ is well-formed and $\prop$ is
well-formed in $\ctx$ (i.e., $\wfprop{\prop}$), which in the case of $\prop =
(\typeprop{x}{\tyB})$ requires that $\Gamma$ already contain an entry of the
form $\bindvar$.

The remainder of the rules in \Cref{fig:wf-judgments} simply express that all
term variables occurring in objects, propositions, types, and type results must
be in scope, where the result $\R$ of function types $\funty{x}{\tyA}{\R}$ has
$x$ in scope. We have patterned these judgments and rules after the $\wfctx$ and
$\wftype{\tyA}$ rules of full-spectrum dependent type theories in the tradition
of \citet{martin-lof1984}. Accordingly, in the rest of this paper we
adopt the type-theoretic convention of \emph{presuppositions}
\citep{martin-lof1996}, namely that in every mention of a judgment
$\ctx\vdash\mathcal{J}$ we implicitly require that $\ctx$ range only over
well-formed contexts; in subtyping judgments $\subtype{\ctx}{\tyA}{\tyB}$ we
implicitly require that $\tyA,\tyB$ are well-formed in $\ctx$; in typing
judgments $\typeof{\ctx}{\expr}{\R}$ we implicitly require that $\R$ is
well-formed in $\ctx$; and so forth.

\subsection{The Proof and Subtyping Systems}
\label{ssec:formalism:proof-system}

We now turn our attention to the intertwined proof and subtyping systems of
$\lamtrl$, which do most of the heavy lifting in occurrence typing systems. We
begin with the \emph{provability} judgment $\proves{\ctx}{\prop}$ defined in
\Cref{fig:proofsystem}, which states that the proposition $\prop$ holds in
context $\ctx$. Many of these rules are just the standard rules for truth,
falsity, conjunction, and disjunction in propositional logic. Note that, unlike
prior work \directworkrefs{}, we require that the unconstrained propositions in
the conclusions of \RULE{L-False}, \RULE{L-OrI1}, \RULE{L-OrI2}, and
\RULE{L-OrE} are well-formed in the current context.

\showproofsystemfigure

The remaining rules fall into one of several categories. \RULE{L-Refl} and
\RULE{L-Transp} govern the equality proposition $\eqprop{x}{y}$ introduced by
\citet{kent2016}, which states that the variables $x$ and $y$ are aliases:
\RULE{L-Refl} states that this relation is reflexive, and \RULE{L-Transp} states
that whenever $\eqprop{x}{y}$, then all propositions that hold for $x$ also hold
for $y$. This implies, for example, that equality is also symmetric and
transitive. (We note that Kent et al.'s \RULE{L-Transport} rule uses the
undefined notations $\psi(x)$ and $\psi(y)$ in its premise and conclusion,
whereas our rule is explicitly about substitution.)

\RULE{L-Sub} and \RULE{L-SubNot} connect typing propositions to subtyping.
\RULE{L-Sub} states that if $x$ has type $\tyA$ which is a subtype of $\tyB$
(written $\subtype{\ctx}{\tyA}{\tyB}$), then $x$ also has type $\tyB$;
\RULE{L-SubNot} states that if $x$ does \emph{not} have type $\tyA$ which is a
supertype of $\tyB$, then $x$ also does not have type $\tyB$. Next, \RULE{L-Bot}
states that $\typeprop{x}{\bot}$ is a contradiction, and \RULE{L-Ann} mediates
between our two context extensions by stating that variables $\bindvar$ in
$\Gamma$ have their declared types.

\paragraph{Restrict and remove}

The final two proof system rules, \RULE{L-Restrict} and \RULE{L-Remove}, refer
to a pair of auxiliary judgments $\restrictrel{\ctx}{\tyA}{\tyB}{\tyC}$ and
$\removerel{\ctx}{\tyA}{\tyB}{\tyC}$, defined in \Cref{fig:restrict-remove}. One
should think of these judgments as declarative specifications of metafunctions
$\restrictop{\tyA}{\tyB}$, which computes the intersection of the types $\tyA$
and $\tyB$, and $\removeop{\tyA}{\tyB}$, which computes the set difference
$\tyA\setminus\tyB$. These operations let us combine information from multiple
typing propositions about a single variable. \RULE{L-Restrict} states that if
$\typeprop{x}{\tyA}$ and $\typeprop{x}{\tyB}$ then $\typeprop{x}{\tyC}$, where
$\restrictrel{\ctx}{\tyA}{\tyB}{\tyC}$, and \RULE{L-Remove} states that if
$\typeprop{x}{\tyA}$ and $\nottypeprop{x}{\tyB}$ then $\typeprop{x}{\tyC}$,
where $\removerel{\ctx}{\tyA}{\tyB}{\tyC}$.

\showrestrictremovefigure

The rules for $\mathsf{restrict}$ and $\mathsf{remove}$ in
\Cref{fig:restrict-remove} are thus fairly self-explanatory as
overapproximations of $\tyA\cap\tyB$ and $\tyA\setminus\tyB$ respectively: that
is, whenever $\restrictrel{\ctx}{\tyA}{\tyB}{\tyC}$, then all values that have
both type $\tyA$ and type $\tyB$ also have type $\tyC$, and whenever
$\removerel{\ctx}{\tyA}{\tyB}{\tyC}$, then all values that have type $\tyA$ but
not type $\tyB$ also have type $\tyC$. (See \cref{lem:prfsound}.)

The rules for $\mathsf{restrict}$ include that the intersection of any base type
with any other base type (\RULE{Restrict-Base}) or with any function type
(\RULE{Restrict-BaseFun}) is the empty type $\bot$; that intersection is
symmetric (\RULE{Restrict-Sym}); and as a fallback (\RULE{Restrict-Other}), we
can always soundly use $\tyB$ as the intersection of $\tyA$ and $\tyB$ for the
purposes of \RULE{L-Restrict}. The rules for $\mathsf{remove}$ include the
principles that $\tyA\setminus\tyB$ is the empty type whenever $\tyA$ is a
subtype of $\tyB$ (\RULE{Remove-Sub}), and that we can soundly overapproximate
$\tyA\setminus\tyB$ with $\tyA$ for the purposes of \RULE{L-Remove}
(\RULE{Remove-Other}).

\begin{remark}
We note that the proof and subtyping systems, including
$\mathsf{restrict}$/$\mathsf{remove}$, are \emph{not intended to be complete}.
Of course, the rules do have to be \emph{sound} in
a sense that we will discuss in \cref{sec:logical-relations-soundness}.
\end{remark}

\begin{remark}
The prior work on $\lambda_{TR}$ formulates $\mathsf{restrict}$ and
$\mathsf{remove}$ as metafunctions rather than judgments. For example,
\citet[Figure 9]{tobin-hochstadt2010} include the clauses:
\begin{align*}
  \restrictop{\tau}{\sigma} &= \bot \qquad\metaif
    \not\exists v.\vdash v:\tau; \psi_1; o_1 \text{ and}
    \vdash v:\sigma; \psi_2; o_2 \\
  \restrictop{\tau}{\sigma} &= \tau\, \qquad\metaif \vdash\tau <: \sigma \\
  \restrictop{\tau}{\sigma} &= \sigma \qquad\metaotherwise
\end{align*}

The condition on the first clause is a semantic condition about the
non-existence of certain values. Although this is a sensible \emph{correctness}
condition for $\mathsf{restrict}$, it is not clear how to interpret it as part
of the definition of a metafunction used in the definition of the type system
itself.

In addition, the condition on the third clause causes the definition of
$\mathsf{restrict}$ to depend on which subtyping judgments are \emph{not}
derivable, but the definition of $\mathsf{restrict}$ affects which propositions
are provable (via \RULE{L-Restrict}), which in turn affects which subtyping
judgments hold (via \RULE{S-Expl}, which we will see shortly)! It is unclear
whether the resulting subtyping judgment is well-founded, being defined
(indirectly) in terms of its own negation. We avoid both of these problems by
reformulating $\mathsf{restrict}$ and $\mathsf{remove}$ as judgments defined
mutually with subtyping and provability.
\end{remark}

\paragraph{Subtyping}

The three mutually-defined subtyping judgments of $\lamtrl$ are presented in
\Cref{fig:subtyping}. First, the judgment $\subtype{\ctx}{\tyA}{\tyB}$ states
that $\tyA$ is a \emph{subtype} of $\tyB$ in context $\ctx$. Most of its rules
are standard. We have added a rule \RULE{S-Expl} asserting that all subtyping
relationships hold in inconsistent contexts, thereby correcting a discrepancy
between $\lambda_{TR}$ and the Typed Racket implementation: in the latter,
\emph{all terms} have all types whenever $\FF$ is derivable, whereas in the
former, only \emph{variables} have all types whenever $\FF$ is derivable. For a
concrete example, $\If{\false}{\ap{\true}{\true}}{\true}$ is well-typed in Typed
Racket and in $\lamtrl$ (by \RULE{S-Expl} and \RULE{T-Subsume}) but not in
$\lambda_{TR}$.

\showsubtypingfigure

Subtyping for function types
$\subtype{\ctx}{\funty{x}{\tyB}{\R_b}}{\funty{x}{\tyA}{\R_a}}$ is contravariant
in the source and covariant in the target. Because these targets are type
results $\R = \res[\tyC][\posprop][\negprop][\obj]$ and not just types, we must
also define a \emph{subtyping judgment for type results}
$\subtype{\ctx}{\R_a}{\R_b}$, also in \Cref{fig:subtyping}. A type result is a
subtype of another type result if the first type is a subtype of the second, the
propositions of the first imply the propositions of the second (i.e.,
$\proves{\ctx \extb \posprop[1]}{\posprop[2]}$ and $\proves{\ctx \extb
\negprop[1]}{\negprop[2]}$), and the object of the first is at least as precise
as the object of the second.

We express that final condition using the \emph{object subtyping judgment}
$\subtype{\ctx}{\obj_1}{\obj_2}$, whose rules are also in \Cref{fig:subtyping}.
The object $\obj_1$ is at least as precise as $\obj_2$ if $\obj_1=\botobj$ or
$\obj_2=\topobj$ or $\obj_1,\obj_2$ are aliased variables (or the same
variable). The \RULE{SO-Expl} rule is needed in concert with \RULE{S-Expl} and
the earlier \RULE{L-False} in order for all type result subtyping relationships
to hold in inconsistent contexts.

\subsection{The Typing Judgment}
\label{ssec:formalism:type-system}

We are finally prepared to discuss the typing judgment
$\typeof{\ctx}{\expr}{\res}$ itself, defined in \Cref{fig:typingrules}. This
judgment captures the following properties of $\expr$:
\begin{itemize}
\item The expression $\expr$ has type $\tyA$;
\item If $\expr$ is true (i.e., evaluates to anything other than $\false$) then
the proposition $\posprop$ holds. Otherwise, if $\expr$ evaluates to $\false$,
then $\negprop$ holds.
\item If the object $\obj$ is a variable $x$, then the values of $\expr$ and $x$
will be equal at runtime. When instead $\obj=\topobj$, we have no information of
this kind. (The third kind of object, $\obj=\botobj$, is a technical device used
in the definition of object substitution below.)
\end{itemize}

\showtypingrulesfigure

\showbasetypesfigure

The typing rules for base values (\RULE{T-BaseVal}) and primitive operations
(\RULE{T-Primop}) are defined in terms of the metafunctions $\typeofbvsym$ and
$\typeofopsym$ respectively, given in \Cref{fig:basetypes}. Number constants
$\nat$, for example, have type $\Nat$; a ``then-proposition'' $\posprop = \TT$
because they always evaluate to non-$\false$; an ``else-proposition'' $\negprop
= \FF$, again because they are never equal to $\false$; and object $\topobj$,
because they are not necessarily equal to any variable in scope. The constants
$\true$ and $\false$ have types $\true$ and $\false$ respectively, and the then-
and else-propositions of $\false$ are $\FF$ and $\TT$ respectively.

The primitive operation $\isnatop$ has function type $\predfun{\Nat}$,
indicating that it can take an input $x$ of any type and always returns a
boolean $\Bool$ (abbreviating the type $\uniontype{\true}{\false}$). When it
returns $\true$ we learn the proposition $\typeprop{x}{\Nat}$, i.e., that its
input is a number; when it returns $\false$ we learn instead that
$\nottypeprop{x}{\Nat}$. The object of $\isnatop$ is again $\topobj$ because it
is not necessarily equal to any variable in scope. The other primitive
operations follow a similar pattern.

The \RULE{T-Var} rule states that whenever the proof system derives
$\proves{\ctx}{\typeprop{x}{\tyA}}$, then
$\typeof{\ctx}{x}{\res[\tyA][\nottypeprop{x}{\false}][\typeprop{x}{\false}][x]}$.
That is, $x$ has type $\tyA$, if it evaluates to non-$\false$ then it is
non-$\false$ (and similarly for $\false$), and the value of $x$ is always equal
to the value of $x$. The \RULE{T-Abs} and \RULE{T-Subsume} rules are similarly
unremarkable.

The \RULE{T-If} rule is the core rule of occurrence typing, stating that
$\typeof{\ctx}{\If{e_1}{e_2}{e_3}}{\IfR}$ when $\typeof{\ctx}{e_1}{\IfCondR}$
and $e_2$ and $e_3$ have type $\R$ under the assumptions $\ifbrposprop$ and
$\ifbrnegprop$ respectively. In concert with the proof system, this rule
propagates typing information from the condition expression $e_1$ to the two
branches $e_2$ and $e_3$. In a minor cosmetic improvement to prior work, we
require that $e_2,e_3$ have exactly the same result type $\R$, rather than
allowing them to have distinct propositions which are then joined in the
conclusion of the rule; the original $\RULE{T-If}$ rule of
\citet{tobin-hochstadt2010} is easily recovered from ours (or vice versa) via \RULE{T-Subsume}.

The \RULE{T-App} rule resembles the $\Pi$-elimination rule of dependent type
theory \citep{martin-lof1984}: if $\expr_1$ has type $\funty{x}{\tyA}{\R}$ and
$\expr_2$ has type $\tyA$, then $\app{\expr_1}{\expr_2}$ should have type $\R$
with $x$ instantiated with $\expr_2$. For example, $\isnatop$ has type
$\predfun{\Nat}$, so $\app{\isnatop}{y}$ has type
$\res[\Bool][\typeprop{y}{\Nat}][\nottypeprop{y}{\Nat}][\topobj]$. The
difficulty arises when we apply functions to non-variables: we might imagine
that $\app{\isnatop}{0}$ has type
$\res[\Bool][\typeprop{0}{\Nat}][\nottypeprop{0}{\Nat}][\topobj]$, but
\emph{this is not even a type} because $\typeprop{0}{\Nat}$ and
$\nottypeprop{0}{\Nat}$ are not propositions. That is, propositions (and thus
type results) are not closed under substitution of terms for variables!

\paragraph{Object substitution}

The solution to this problem, dating back to the original $\lambda_{TR}$
calculus, is twofold. First, the purpose of objects is to allow us to handle
some non-variable terms as if they were variables, when one of the variables in
scope is equal to that term. For example, although $\app{\mathsf{id}}{y}$ is not
a variable, we know that its value will always equal the value of the variable
$y$, so we can soundly assign $\app{\isnatop}{\app{\mathsf{id}}{y}}$ the type
$\res[\Bool][\typeprop{y}{\Nat}][\nottypeprop{y}{\Nat}][\topobj]$. In fact, the
conclusion of the \RULE{T-App} rule substitutes the \emph{object of $\expr_2$}
for $x$ in $R$ in the type of $\app{\expr_1}{\expr_2}$, an operation we write
$\osubst{\R}[x][\obj_2]$; since variables have themselves as objects, this
strictly generalizes our desired behavior on variables.

\begin{remark}
The example $\app{\isnatop}{\app{\mathsf{id}}{y}}$ is admittedly contrived, but
objects are an important part of Typed Racket's ability to type realistic
programs with $\mathbf{let}$-binding \citep{kent2016}, particularly when the
grammar of objects is extended to include structure selectors applied to
variables, such as $\app{\mathbf{car}}{y}$ and $\app{\mathbf{cdr}}{y}$
\citep{tobin-hochstadt2010}.
\end{remark}

Of course, many terms will not be equal to any variable in scope, and such terms
are assigned the object $\topobj$. The object substitution $\poserase{\R}$
\emph{erases} all propositions about $x$ in $\R$, replacing them by \emph{weaker
propositions that do not mention $x$}. For example, both
$\poserase{(\typeprop{x}{\Nat})}$ and $\poserase{(\nottypeprop{x}{\Nat})}$ are
defined to be $\TT$, so the \RULE{T-App} rule assigns $\app{\isnatop}{0}$ the
type $\res[\Bool][\TT][\TT][\topobj]$.

\Citet{tobin-hochstadt2010} define erasure $\poserase{\R}$ to take atomic
propositions mentioning $x$ to $\TT$, leave unchanged any atomic propositions
not mentioning $x$, and otherwise proceed recursively through other propositions
and types.\footnote{Except on the left-hand side of $\psi_1\supset\psi_2$
propositions, but such propositions are omitted from later work.} However, this definition is
unsound: it does not satisfy $\subtype{\ctx}{\tyA}{\poserase{\tyA}}$ when $x$
occurs in the domain of a function type in $\tyA$, because function subtyping is
contravariant in that position. As discussed in \cref{ssec:icfp-wrong}, this
unsoundness is also reflected in the actual implementation of Typed Racket, as
we discovered during our work.

Our solution is to introduce a \emph{negative erasure} operation $\negerase{\R}$
which instead takes atomic propositions mentioning $x$ to $\FF$. Positive and
negative erasure flip polarity on the left-hand side of arrows and in the types
of negative typing propositions not mentioning $x$:
\begin{align*}
  \poserase{(\funty{y}{\tyA}{\R})} &=
  \funty{y}{\negerase{\tyA}}{\poserase{\synR}} \quad \metaif x \neq y
  \\
  \negerase{(\funty{y}{\tyA}{\R})} &=
  \funty{y}{\poserase{\tyA}}{\negerase{\synR}} \quad \metaif x \neq y
  \\
  \poserase{(\nottypeprop{z}{\tyB})} &=
  \nottypeprop{z}{\negerase{\tyB}} \quad \metaif z \neq x
  \\
  \negerase{(\nottypeprop{z}{\tyB})} &=
  \nottypeprop{z}{\poserase{\tyB}} \quad \metaif z \neq x
\end{align*}

The soundness conditions for erasure then state that
$\subtype{\ctx}{\negerase{\tyA}}{\tyA}$ and
$\subtype{\ctx}{\tyA}{\poserase{\tyA}}$, and it is easy to see that the above
rules for function types preserve these conditions. (See
\cref{lem:subtyping-erase}.) The full rules for object substitution into type
results, types, propositions, and objects---including positive and negative
erasure---are located in \Cref{fig:objsubst}. (Note that the $\abssubst{}{x}{y}$
notation refers to ordinary capture-avoiding substitution of $y$ for $x$.)

\showobjectsubstitutionfigure

\paragraph{Local binding}

Finally, our \RULE{T-Let} rule is very similar to the one introduced to Typed
Racket by \citet{kent2016}. To type the expression $\Let{\expr}{\expr_b}$, we
start by typing $\typeof{\ctx}{\expr}{\res[\tyA][\posprop][\negprop][\obj]}$. To
type the body $\expr_b$, we extend the context by a variable $\bindvar$ and a
pair of propositions $\osubst{(\eqprop{x}{z})}[z][\obj]$ and
\(
  \disjprop
  {(\conjprop{(\nottypeprop{x}{\false})}{\posprop})}
  {(\conjprop{(\typeprop{x}{\false})}{\negprop})}
\).
The latter proposition simply gives $e_b$ access to the then- and
else-propositions of $e$: either $\nottypeprop{x}{\false}$ and $\posprop$ hold,
or $\typeprop{x}{\false}$ and $\negprop$ hold. When $e$'s object is a variable
$y$, the former proposition is $\eqprop{x}{y}$, which allows the proof system to
transfer facts about $y$ (and hence about $e$) to $x$ as well via
\RULE{L-Transp}. When $e$'s object is $\topobj$, there is no such knowledge to
transfer and this proposition automatically computes to $\TT$ by the rules of
object substitution. In either case, supposing that $e_b$ has type $\R$ under
the above hypotheses, we assign $\Let{\expr}{\expr_b}$ the type
$\osubst{\R}[x][\obj]$ because (as in \RULE{T-App}) $x$ is no longer in scope
once we leave the body of the $\mathbf{let}$.

\subsection{Operational Semantics}

The operational semantics of $\lamtrl$ is a very standard deterministic small-step
call-by-value operational semantics defined using evaluation contexts, in the
style of \citet{felleisen1992}. The only wrinkle is, once again, the fact that
types are not closed under substitution of non-variables for variables.
Specifically, because our $\lambda$-expressions contain type annotations, the
operation $\substexpr{e}{x}{v}$ is not well-defined whenever $e$ contains a
$\lambda$-subexpression (and thus a type subexpression). Defining
$\substexpr{(\abs{y}[A]{e'})}{x}{v}$ as $\abs{y}[A]{(\substexpr{e'}{x}{v})}$
leaves free occurrences of $x$ in $A$; defining it instead as
$\abs{y}[\substexpr{\tyA}{x}{v}]{(\substexpr{e'}{x}{v})}$ results in an
ill-formed type subexpression $\substexpr{\tyA}{x}{v}$.

\showfullopsemfigure

Fortunately, because type annotations are not used at runtime, we can avoid this
problem altogether by defining our operational semantics on \emph{runtime
expressions} $\unexpr$, which are identical to the expressions $e$ in
\Cref{fig:stx} except for $\lambda$ no longer being annotated. We write
$\eraseann{e}$ for the evident \emph{projection} from expressions to runtime
expressions. See \Cref{fig:opsem-full} for our operational semantics.

\section{Semantic Type Soundness via Step-Indexed Logical Relations}
\label{sec:logical-relations-soundness}

The goal of this section is to state and prove \emph{type soundness}
(\cref{thm:type-soundness}) for the $\lamtrl$ formalism presented in
\cref{sec:formalism}. As a corollary, every closed term of type $\Bool$ either
diverges or evaluates to $\true$ or $\false$ (\cref{cor:bool}). As advocated by
\citet{timany2024}, we establish type soundness using the technique of
\emph{step-indexed logical relations} introduced by \citet{appel2001} and
subsequently refined by \citet{ahmed2006}, \citet{dreyer2009}, and many others.
Our proof of type soundness is fully formalized in the Lean proof assistant
\citep{demoura2015} (\cref{ssec:lean}).

The definitions and lemmas involved in our logical relations proof are quite
standard, which we regard as a key strength of our approach in comparison to the
ad hoc methods of recovering syntactic type soundness in related work,
including:
\begin{itemize}
\item additional typing rules for \emph{otherwise ill-typed intermediate
states}, as in \citet{tobin-hochstadt2008} and \citet{castagna2022};

\item additional typing rules for \emph{final states} that entangle syntax and
semantics, as in the closure typing rule of $\lambda_{TR}$
\citep{tobin-hochstadt2010} and its successors
\citep{bonnaire-sergeant2016,fu2021,kent2016}, which quantifies over
semantically-valid environments, and the machine state typing rule of
\citet{chang2017}, which derives typing contexts from runtime environments;

\item \emph{non-well-founded} typing rules that entangle syntax and semantics by
having the typing judgment refer to itself in negative position, as in
\citet{flanagan2006}, \citet{rondon2008}, and \citet{vekris2015}; and

\item techniques to circumvent that non-well-foundedness by bootstrapping the
syntax through a denotational semantics of contracts for the simply-typed lambda
calculus \citep{knowles2010} or defining an infinite tower of increasingly
precise type systems \citep{chugh2012}.
\end{itemize}

We believe that our semantic soundness proof is not only conceptually cleaner
than the above but is also much easier to extend, as programming language
researchers have demonstrated time and again that step-indexed logical relations
are a robust technique for establishing soundness results \citep{timany2024}.
Moreover, as discussed in \cref{sec:background}, our work fixes ill-scoped
treatments of variables in type expressions found in many of the above papers
\citep{tobin-hochstadt2010,bonnaire-sergeant2016,fu2021,kent2016,chang2017}.

\subsection{Defining the Logical Relations}
\label{sec:logical-relations}

We mutually define $\mathbb{N}$-indexed relations for
closed runtime expressions of a given type result $\expinterp{\R}{\unexpr}$,
closed values of a given type result $\valinterpres{\R}{v}$,
closed values of a given type $\valinterp{\tyA}{v}$, and
the truth of propositions $\propinterp{\prop}$.
Each of these relations is indexed by a \emph{value environment} $\env$
assigning values to variables, which is used to interpret the free variables in
types, propositions, objects, and type results; these environments are the
subject of the remaining relation $\ctxinterp{\ctx}{\env}$.
\[
  \textit{Value Environments} \quad \env \sdef \empenv \mid \extenv{\env}{x}{v}.
\]

These relations capture the notion that a particular expression satisfies the
obvious semantic property associated to each typing judgment. The step-index $n$
is simply introduced to account for non-termination: that property either holds
within $n$ steps, or is automatically considered to hold if the step budget is
exhausted. For example, $\expinterp{\res[\Nat][\TT][\TT][\topobj]}{\unexpr}$
holds when $\unexpr$ either evaluates to a number in fewer than $n$ steps, or continues reducing for at least $n$ steps.

\paragraph{Contexts}

We begin by defining the logical relation for value environments satisfying a
given context. The relation $\ctxinterp{\ctx}{\env}[n]$ expresses that $\env$
assigns a value to every variable $\bindvar$ in the context $\Gamma$, that each
of these values satisfies the semantic property of the corresponding type, and
that this choice of values satisfies all the propositions in $\Gamma$.
\begin{align*}
  \ctxinterp{\ctx}{\env}[0] \mdef& \dom{\ctx} \subseteq \dom{\env} \\
  \ctxinterp{\empctx}{\env}[n+1] \mdef& \metatrue \\
  \ctxinterp{\ctx \exta \bindvar}{\env}[n+1] \mdef& \ctxinterp{\ctx}{\env}[n+1] \metaand \valinterp{\tyA}{\env(x)}[n+1] \\
  \ctxinterp{\ctx \exta \prop}{\env}[n+1] \mdef& \ctxinterp{\ctx}{\env}[n+1] \metaand \propinterp{\prop}[n+1]
\end{align*}

As a technical matter we have incorporated the first condition into the $n=0$
clause only, but as a straightforward consequence one can show that for any $n$,
$\ctxinterp{\ctx}{\env}[n]$ implies $\dom{\ctx} \subseteq \dom{\env}$.

\paragraph{Propositions}

The logical relation for propositions is defined in \Cref{fig:propinterp}. The
relation $\propinterp{\prop}[n]$ states that the proposition $\prop$,
\emph{which may contain free variables}, holds at step count $n$ when each
variable is interpreted according to the environment $\env$. Our soundness
theorem (\cref{lem:prfsound}) will then state that if $\proves{\ctx}{\prop}$ is
derivable, then for all $\env$ satisfying $\allsat$, $\propinterp{\prop}$ holds.

\showpropinterpfig

Most of \Cref{fig:propinterp} is straightforward, but a few remarks are in
order. First, every proposition is said to hold at $n=0$. The positive typing
proposition $\typeprop{x}{\tyA}$ is said to hold at $n+1$ if $\env(x)$ satisfies
the logical relation of $\tyA$ at the same $n+1$, but the negative typing
proposition $\nottypeprop{x}{\tyA}$ is said to hold at $n+1$ if $\env(x)$
\emph{does not} satisfy the logical relation of $\tyA$ at any $1\leq m \leq
n+1$. This condition, which can be read off of the Kripke interpretation of
negation in synthetic accounts of step-indexed logical relations
\citep{dreyer2009}, is needed for monotonicity/downward-closure of the logical
relations (\cref{lem:srlmono}).

Finally, the equality proposition $\eqprop{x}{y}$ holds when $\env(x)$ and
$\env(y)$ are \emph{exactly the same} (i.e., $\alpha$-equivalent)---not
contextually equivalent or any other such condition---because the equality
proposition is meant to capture variable aliasing.

\paragraph{Types}

The logical relation for values of a given type, $\valinterp{\tyA}{v}$, is
defined in \Cref{fig:valinterptype}. Note that $v$ must be closed, but $\tyA$
may contain free variables (in the then- and else-propositions or object of function
types) which are interpreted according to the environment $\env$.

\showvaltyinterpfig

All the definitions here are standard: every relation holds at $n=0$, the
relation at $\uniontype{\tyA}{\tyB}$ is the union of the relations at $\tyA$ and
$\tyB$, and the relation at $\funty{x}{\tyA}{\R}$ holds if $v$ is a
$\lambda$-term or a primitive operation, and whenever it is applied to a value
in the relation for $\tyA$ at any smaller or equal step count, the result is in
the relation for $\R$ at that step count in an extended environment $\envw$ to
account for free occurrences of $x$ in $\R$.

\paragraph{Type results}

Finally, we define the logical relations for runtime expressions and values of a
given type result, $\expinterp{\R}{\unexpr}$ and $\valinterpres{\R}{v}$. The
relation $\expinterp{\R}{\unexpr}$ lifts the value relation to runtime
expressions via the small-step operational semantics:
\begin{align*}
  \expinterp{\semR}{\unexpr}[0] \mdef& \metatrue \\
  \expinterp{\semR}{\unexpr_0}[n+1] \mdef&
    \expinterp{\semR}{\unexpr_1}[n] \quad\metaif \eestepsto{\unexpr_0}{\unexpr_1} \\
  \expinterp{\semR}{v}[n+1] \mdef& \valinterp{\semR}{v}[n+1]
\end{align*}

The value relation $\valinterp{\R}{v}$ is more complicated, but simply records
the intended meaning of the four components of type results $\R =
\res[\tyA][\posprop][\negprop][o]$ as already outlined in
\cref{ssec:formalism:type-system}. As in the other logical relations, all four
of these components may contain free variables, which are interpreted according
to the environment $\env$.
\begin{align*}
  \valinterpres{\res[\tyA][\posprop][\negprop][o]}{v}[0] \mdef& \metatrue \\
  \valinterpres{\res[\tyA][\posprop][\negprop][o]}{v}[n+1] \mdef&
    \valinterp{\tyA}{v}[n+1] \metaand \\
    &\left(
      (\propertyCfst{v}{}{\env}{n+1}) \metaor
      (\propertyCsnd{v}{}{\env}{n+1})
     \right) \metaand \\
    & \left(\propertyA{\obj}{v}\right)
\end{align*}

Specifically, $\valinterpres{\res[\tyA][\posprop][\negprop][o]}{v}[n+1]$ is
defined to mean that (1) $v$ is in the type relation at $\tyA$; (2) either $v$
is exactly $\false$ and the else-proposition $\negprop$ holds at $\env$, or $v$
is non-$\false$ and the then-proposition $\posprop$ holds at $\env$; and (3)
either the object $o$ is $\topobj$, or the object is a variable $y$ such that
$\env(y)$ is exactly $v$. The third possibility---that the object is
$\botobj$---is impossible for values, in exactly the same way that no value has
type $\botobj$.

\subsection{Proving Semantic Type Soundness} \label{sec:soundness}

We now sketch the key lemmas and theorems in our semantic type soundness proof.
Full paper proofs of these and other necessary statements are provided in the
accompanying supplementary material, as is a Lean formalization of the same
content.

The first step in our proof is not actually about logical relations at all; it
is a syntactic admissibility result about subtyping for the positive and
negative erasure operations alluded to in \cref{ssec:formalism:type-system},
stating that the positive erasure of a type, type result, object, or proposition
is at least as weak as it, and the negative erasure at least as strong.

\begin{lemma}[Subtyping for Erasure]\label[lemma]{lem:subtyping-erase}
Suppose $\wfctx$. Then:
  \begin{enumerate}
  \item If $\wftype{\tyA}$, then $\subtype{\ctx}{\negerase{\tyA}}{\tyA}$ and $\subtype{\ctx}{\tyA}{\poserase{\tyA}}$.
  \item If $\wfres{\R}$, then $\subtype{\ctx}{\negerase{\R}}{\R}$ and $\subtype{\ctx}{\R}{\poserase{\R}}$.
  \item If $\wfobj{\obj}$, then $\subtype{\ctx}{\negerase{\obj}}{\obj}$ and $\subtype{\ctx}{\obj}{\poserase{\obj}}$.
  \item If $\wfprop{\prop}$, then if $\proves{\ctx}{\negerase{\prop}}$ then $\proves{\ctx}{\prop}$, and if $\proves{\ctx}{\prop}$ then $\proves{\ctx}{\poserase{\prop}}$.
  \end{enumerate}
\end{lemma}
\begin{proof}
  By simultaneous induction on the structure of $\tyA$, $\R$, $\obj$, and $\prop$.
\end{proof}

These eight properties must be established simultaneously because positive and
negative erasure for each of these grammatical categories refers to all of the
others. \Cref{lem:subtyping-erase} will become relevant in the \RULE{T-App} and
\RULE{T-Let} cases of our proof of the fundamental theorem of logical relations,
which are the two places in the type system that use object substitution.

Next, we establish standard key properties of step-indexed logical relations:
head expansion (\cref{lem:hdxp}), or closure under converse evaluation, head
reduction (\cref{lem:hdrd}), or closure under evaluation, and monotonicity
(\cref{lem:srlmono}), or downward-closure of step-indices.
\Cref{lem:hdxp,lem:hdrd} follow directly from the definition of the logical
relations for type results, but are used many times throughout later proofs.
\Cref{lem:srlmono} is straightforward but must be proven simultaneously for
every relation because they all refer to one another.

\begin{lemma}[Head Expansion]\label[lemma]{lem:hdxp}\leavevmode
  \begin{enumerate}
    \item
      If $\expinterp{\R}{\unexpr'}[m]$ and $\neestepsto{\unexpr}{\unexpr'}[k]$, then $\expinterp{\R}{\unexpr}[m+k]$.
    \item
      If $\valinterp{\R}{v}[m]$ and $\neestepsto{\unexpr}{v}[k]$, then $\expinterp{\R}{\unexpr}[m+k]$.
  \end{enumerate}
\end{lemma}

\begin{lemma}[Head Reduction]\label[lemma]{lem:hdrd}
  If $\expinterp{\R}{\unexpr}$, then either:
  \begin{enumerate}
  \item there exists an $\unexpr'$ such that $\neestepsto{\unexpr}{\unexpr'}[n]$
  (i.e., $\unexpr$ exhausts the step budget), or
  \item there exists a value $v$ such that $\neestepsto{\unexpr}{v}[m]$ for some
  $0\leq m<n$, and $\valinterp{\R}{v}[n-m]$.
  \end{enumerate}
\end{lemma}

\begin{lemma}[Monotonicity]\label[lemma]{lem:srlmono}
  For all $m\leq n$,
  \begin{enumerate}
    \item \label{lem:srlmono-ctx} If $\ctxinterp{\ctx}{\env}$, then $\ctxinterp{\ctx}{\env}[m]$.
    \item \label{lem:srlmono-prop} If $\propinterp{\prop}$, then $\propinterp{\prop}[m]$.
    \item \label{lem:srlmono-type} If $\valinterp{\tyA}{v}$, then $\valinterp{\tyA}{v}[m]$.
    \item \label{lem:srlmono-exp} If $\expinterp{\R}{\unexpr}$, then $\expinterp{\R}{\unexpr}[m]$.
    \item \label{lem:srlmono-res} If $\valinterp{\R}{v}$, then $\valinterp{\R}{v}[m]$.
  \end{enumerate}
\end{lemma}

The next step is to show that the rules of the proof and subtyping judgments are
validated by the logical relations. Recall from
\cref{ssec:formalism:proof-system} that these judgments are all defined
mutually, but are \emph{not} mutual with the main typing judgment itself; as a
result, we must establish these results simultaneously, but can do so before
starting our proof of the fundamental theorem itself. The first two clauses
state the correctness of the subtyping judgments; the third and fourth state the
correctness of the $\mathsf{restrict}$ and $\mathsf{remove}$ judgments described
in \cref{ssec:formalism:proof-system}; the final clause states the correctness
of the proof system, as in Lemma 1 of \citet{tobin-hochstadt2010}.

\begin{theorem}[Correctness of the Proof and Subtyping Systems]\label{lem:prfsound}
  Suppose $\allsat$. Then:
  \begin{enumerate}
  \item
    If $\unexpr$ is closed, $\subtype{\ctx}{\R_a}{\R_b}$, and $\expinterp{\R_a}{\unexpr}$, then $\expinterp{\R_b}{\unexpr}$.
  \item If $\subtype{\ctx}{\tyA}{\tyB}$ and $\valinterp{\tyA}{v}$, then $\valinterp{\tyB}{v}$.
  \item If $\restrictrel{\ctx}{\tyA}{\tyB}{\tyC}$, $ \valinterp{\tyA}{v}$, and $\valinterp{\tyB}{v}$, then $\valinterp{\tyC}{v}$.
  \item If $\removerel{\ctx}{\tyA}{\tyB}{\tyC}$, $ \valinterp{\tyA}{v}$, and $\notvalinterp{\tyB}{v}$, then $\valinterp{\tyC}{v}$
  \item
    if $\proves{\ctx}{\prop}$, then $\propinterp{\prop}$.
  \end{enumerate}
\end{theorem}
\begin{proof}
  By simultaneous rule induction.
\end{proof}

When combined with \cref{lem:subtyping-erase}, \cref{lem:prfsound} implies that
the logical relations are \emph{closed under object substitutions
$\osubst{}[x][\obj]$} in the sense that whenever $\allsat$ and
$\valinterpres{\R}{v}$, for every object $\obj$ that is either $\topobj$ or a
variable $y$ with $\env(x) \alphaeq \env(y)$, we have
$\valinterpres{\osubst{\R}[x][\obj]}{v}$. This is in essence the semantic
``compatibility'' or substitution property that enables us to carry out the
\RULE{T-App} and \RULE{T-Let} cases of the fundamental theorem.

\begin{theorem}[Fundamental Theorem] \label{thm:fundamental}
  If $\typeof{\ctx}{\expr}{\synR}$ and $\allsat$, then $\expinterp{\R}{\psubexpr{\eraseann{\expr}}[\env]}$.
\end{theorem}
\begin{proof}
  By rule induction on $\typeof{\ctx}{e}{\synR}$, using the aforementioned lemmas.
\end{proof}

As usual, semantic type soundness follows more or less directly from the
fundamental theorem of logical relations, once we expand a few definitions.

\begin{theorem}[Semantic Type Soundness]\label{thm:type-soundness}
  If $\typeof{\empctx}{\expr}{\R}$, then either $\eraseann{\expr}$ diverges or
  $\neestepsto{\eraseann{\expr}}{v}$ and $\valinterp{\R}{v}[n][\empenv]$ for all $n$.
\end{theorem}
\begin{proof}
  By \Cref{thm:fundamental} with $\ctx = \empctx$ and $\env = \empenv$, we have
  $\expinterp{\R}{\eraseann{\expr}}[n][\empenv]$ for all $n$.

  Terms must either diverge, evaluate to a value, or get stuck (i.e., evaluate to an irreducible
  non-value), so we begin by showing that the third possibility cannot occur for well-typed terms.
  Suppose that $\neestepsto{\eraseann{\expr}}{\unexpr'}[k]$ for some irreducible non-value
  $\unexpr'$. By head reduction (\Cref{lem:hdrd}) and $\expinterp{\R}{\eraseann{\expr}}[k+1]$, we
  know that either (1) there exists $\unexpr''$ such that
  $\neestepsto{\eraseann{\expr}}{\unexpr''}[k+1]$, or (2) there exists a value $v$ such that
  $\neestepsto{\eraseann{\expr}}{v}[m]$ for $m<k+1$. By determinacy, the first contradicts the
  hypothesis that $\unexpr'$ is irreducible, and the latter contradicts the irreducibility of
  values.

  It remains only to show that if $\neestepsto{\eraseann{\expr}}{v}[k]$, then
  $\valinterp{\R}{v}[n][\empenv]$ for all $n$. By determinacy and head reduction applied to
  $\expinterp{\R}{\eraseann{\expr}}[n+k][\empenv]$, we obtain $\valinterp{\R}{v}[n+k-k][\empenv]$,
  completing the proof.
\end{proof}

\begin{corollary}[Soundness for Booleans]\label[corollary]{cor:bool}
  If $\typeof{\empctx}{\expr}{\res[\Bool][\posprop][\negprop][\obj]}$ then either $\eraseann{\expr}$
  diverges or $\neestepsto{\eraseann{\expr}}{\true}$ or
  $\neestepsto{\eraseann{\expr}}{\false}$.
\end{corollary}
\begin{proof}
  By \Cref{thm:type-soundness} it suffices to show that
  $\valinterp{\res[\Bool][\posprop][\negprop][\obj]}{v}[1][\empenv]$ implies $v=\true$ or
  $v=\false$. Recalling that $\Bool \coloneq \uniontype{\true}{\false}$, this follows immediately by
  unfolding the definition of the logical relation.
\end{proof}

\subsection{Lean Formalization}
\label{ssec:lean}

The definition of the $\lamtrl$ type system and operational semantics
(\cref{sec:formalism}) and all the definitions and theorems in our semantic type
soundness proof (\cref{sec:logical-relations-soundness}) have been fully
formalized in the Lean proof assistant. Our formalization builds in Lean version
4.30.0 (released May 2026) and does not rely on any libraries beyond what is
built into the standard Lean distribution; it is approximately 3,600 lines of
code, including comments and blank lines. The results rely only on the standard
classical axioms \texttt{[propext, Classical.choice, Quot.sound]} widely used in
Lean.

The original iteration of our Lean code was auto-formalized from a draft of our
accompanying technical report using Claude Fable 5, and has been maintained and
improved by a combination of AI and human effort as our work has continued to
evolve. The authors have carefully compared the Lean definitions and theorem
statements with the corresponding definitions and statements in this paper.

The only noteworthy difference between our Lean formalization and our paper
definitions is that the formalization uses a de Bruijn representation of binding
indexed by the context extensions of the form $\bindvar$, enabled by our
improved treatment of contexts described in \cref{ssec:formalism:wf-judgments}.

\section{Related Work}
\label{sec:related_work}

As we have discussed in \Cref{ssec:icfp-wrong}, the $\lambda_{TR}$ core calculus
of Typed Racket~\citep{tobin-hochstadt2010} and subsequent work extending or
directly modeled after $\lambda_{TR}$ \directworkrefs{} must address various
challenges inherent in building a restricted dependent type system for a
general-purpose programming language. Unfortunately, the published metatheory in
these papers---spanning from 2010 to 2021---fails to adequately address these
issues, leading to ill-defined formalisms and failures of type soundness. Our
present work both clarifies and resolves these issues.

Over the past two decades, many other related type systems have been designed
and implemented, often aimed at bringing type systems to other untyped
programming languages, or at extending conventional type systems to support
restricted predicates in the form of refinement types. Although we have focused
primarily on the $\lambda_{TR}$ family of work, a number of these other systems
also suffer from related problems, as we have already briefly mentioned in
\cref{sec:logical-relations-soundness}.

\paragraph{Definitional circularity}

\citeauthor{flanagan2006}'s \emph{Hybrid Type Checking} \citep{flanagan2006}
features a type system whose definitions of the typing, subtyping, and closing
substitution judgments are defined mutually with the implication judgment, which
is used in subtyping of refinement types. This is problematic because
\emph{consistent substitutions}, which are defined in terms of the syntactic
typing judgment, appear in negative position in the rule for the implication
judgment, meaning that the definition of the type system is not obviously
well-founded. This issue was noticed by \citet{knowles2010}, who fix that
definition by bootstrapping it via a denotational interpretation, in somewhat of
a similar spirit to our logical relation.

\citeauthor{rondon2008}'s work on \emph{Liquid Types} \citep{rondon2008} has been built upon by a
wide variety of subsequent systems
\cite{vazou2013,vazou2015bounded,vazou2018reflection,vekris2015}. Much of the
metatheory of that subsequent work directly references and relies upon the
metatheory of \citet{rondon2008}, which features an identical circularity to
\citet{flanagan2006}.

In a related line of work, \citet{chugh2012} develop a different solution to
this circularity problem, by developing an infinite stratified hierarchy of
typing derivations and semantic types. This avoids circularity, at the cost of
significant new metatheoretical complexity.

Additionally, all of these systems have static typing rules which refer to the
full operational semantics of the language, as part of their specification of
subtyping for refinements. For example, the \RULE{IMP} rule of
\citet[Figure 4]{knowles2010} states:
\[
\inferrule{
  \forall \sigma.\
  \text{if }
  \vdash \sigma : E
  \text{ and }
  \sigma(s) \rightsquigarrow^{*} \mathsf{true}
  \text{ then }
  \sigma(t) \rightsquigarrow^{*} \mathsf{true}
}{
  E \vdash s \Rightarrow t
}
\]
Although such rules are formally well-defined, they are not remotely
implementable. Practical systems such as Liquid
Haskell~\cite{vazou-liquid-haskell} instead work in a similar fashion to how our
formalism is presented: there is some syntactic proof system (like our
$\proves{\ctx}{\prop}$), often dispatched to an external solver, that has a
\emph{semantic} correctness criterion along the lines of the above rule. The
confusion between defining the system itself and giving its specification is
similar to the problems with the \textit{remove} and \textit{restrict}
metafunctions, as described previously.

\paragraph{Ill-scoped rules}

\Citet{vekris2015} present a liquid type system for JavaScript, based on the
metatheory presented by \citet{rondon2008}. In keeping with the theme of many
other errors described in the present paper, their \RULE{R-Let} rule (in Figure
9) is ill-scoped because its conclusion has a type that may mention the
\texttt{let}-bound variable, which is no longer in scope outside the
\texttt{let} itself.
\[
\inferrule{
  G \vdash M_1 \mathbin{:\!\!:} T_1 \\
  G, x{:}T_1 \vdash M_2 \mathbin{:\!\!:} T_2
}{
  G \vdash
  \mathbf{let}\ x = M_1\ \mathbf{in}\ M_2
  \mathbin{:\!\!:} T_2
}
\]

\paragraph{Restricting expressiveness}

Unlike the systems discussed above, many formalizations of type systems for
untyped languages lie \emph{outside} the dependent type spectrum. For
JavaScript, \citet{chaudhuri2017} and \citet{bierman2014} present two different
typed variants, Flow and TypeScript, respectively. \citet{chaudhuri2017}
specially treats each runtime type-test expression, such as \texttt{typeof x ===
"number"}, their function types do not contain information about type
refinements, and type refinements propagate only through expressions and within
an \texttt{if} statement, thereby avoiding the challenges faced in more
dependently-typed languages including $\lambda_{TR}$.

\Citet{bierman2014} similarly avoid these challenges by employing a
non-dependent system which disallows expressions of any form from appearing in
types. Flow and TypeScript later independently incorporated mechanisms inspired
by the latent propositions of $\lambda_{TR}$
\citep{tobin-hochstadt2010}. However, they restrict the result type of functions
to only refer to the parameters of that function, whereas $\lambda_{TR}$'s
propositions are more expressive, allowing references to any variables in scope.

\citet{castagna2022} incorporate features of occurrence typing into a
set-theoretic type system. The greater expressiveness of the underlying type
system allows type-checking examples \emph{without} type dependency. However,
the issue of ill-typed immediate steps nonetheless arises, and the system
includes a set of additional typing rules to resolve the problems this causes
for their syntactic, subject-reduction-based soundness proof.

\section{Conclusion}
\label{sec:conclusion}

Restricted dependent type systems are an attractive proposition for many
language designers, as they allow types to express rich properties while
remaining fully automated. However, the design of these systems often imposes a
tricky set of requirements on these systems' metatheory. As we have seen, this
can lead to type systems that are ill-defined, whose metatheory contains
significant errors, and whose metatheory is extremely complex.

We address these issues for the core calculus of Typed Racket, $\lambda_{TR}$,
as originally developed by \citet{tobin-hochstadt2010}. To that end, we develop
a new core calculus, $\lamtrl$, which incorporates the main features of
$\lambda_{TR}$ as well as the \emph{equality proposition} extension later
introduced by \citet{kent2016}. Although the main ideas of the $\lamtrl$
calculus are the same as those of $\lambda_{TR}$, our revision incorporates a
number of key improvements and corrections.

First, to ensure well-scopedness throughout the type system, we distinguish two
kinds of context entries: those that introduce a new variable, and those that
introduce new typing information about an existing variable
(\cref{ssec:formalism:wf-judgments}). Next, we reformulate the \textsf{restrict}
and \textsf{remove} metafunctions of $\lambda_{TR}$ as judgments, in order to
correct a well-foundedness issue in the definition of the type system. Moreover,
we add two subtyping rules (\RULE{S-Expl} and \RULE{SO-Expl}) in order to
eliminate a discrepancy between the behaviors of the formalism and the Typed
Racket implementation.

More significantly for Typed Racket programmers, the $\lamtrl$ formalism corrects a flaw in the
definition of the object erasure operation of $\lambda_{TR}$, by introducing a
new \emph{bottom object} $\botobj$ to the calculus and mutually defining
\emph{positive} and \emph{negative} erasure operations to account for
contravariance in function subtyping. This oversight in $\lambda_{TR}$ is
responsible for a genuine unsoundness in the implementation of Typed Racket, as confirmed by the developers.

Our presentation of $\lamtrl$ also incorporates some minor improvements to the
\RULE{T-Let} rule, the \RULE{T-If} rule, and a clarification of the
\RULE{L-Transport} rule of \citet{kent2016}. We also discovered some unrelated
small bugs in the Typed Racket implementation through the course of this work,
including errors in the parsing of complex type expressions.

Most importantly, we correct the type soundness proof of
\citet{tobin-hochstadt2010} by replacing it with a semantic type soundness proof for
$\lamtrl$ using step-indexed logical relations; this both clarifies the issues
at stake and simplifies the metatheory itself. We have proven type soundness on
paper as well as in Lean; both proofs are available in the supplementary
materials. The Lean proof uses a de Bruijn representation of variables which is
made possible by our clarification of the binding structure of contexts.

As we have discussed in \Cref{sec:related_work}, establishing syntactic type
soundness for systems like occurrence typing has often required complex and
novel metatheoretical machinery as well as the care for scoping necessitated by dependency; this combination in turn invites mistakes of the kind we analyze here. Although semantic methods such as step-indexed
logical relations generally have a reputation for being complex, we believe that
they are actually significantly simpler than syntactic type soundness in this
part of the design space.

Building on this work, we believe it would be straightforward to extend the
present calculus and metatheory to accommodate other extensions to Typed Racket,
ranging from refinement types \citep{kent2016} to structures and methods
\citep{fu2021}. We also believe that adopting this approach would significantly
simplify, and some cases correct, the metatheory of a number of similar systems.

\clearpage
\bibliographystyle{ACM-Reference-Format}
\bibliography{references}
\ifdefined\LambdaTRIncludeAppendix
  \ifdefined\LambdaTRIncludeAppendix
  \clearpage
  \appendix
  \begingroup
  \def\AppendixFinish{\endgroup}
  \RenewDocumentCommand{\trlabel}{m}{\label{appendix:#1}}
  \section{Appendix for ``Revisiting Soundness for Occurrence Typing, Semantically''}
  \let\section\subsection
\else
  \documentclass[11pt]{article}
  \usepackage[top=1.5cm, bottom=1.5cm, outer=2cm, inner=2cm, heightrounded, marginparwidth=2.5cm, marginparsep=2cm]{geometry}
  
  \newtheorem{lemma}{Lemma}
  \newtheorem{theorem}[lemma]{Theorem}
  \newtheorem{corollary}[lemma]{Corollary}
  \theoremstyle{definition}
  \RequirePackage{todonotes}
  \NewDocumentCommand{\carlonote}{m}{\todo[inline,color=blue!30]{#1}}
  \NewDocumentCommand{\ffnote}{m}{\todo[inline,color=purple!30]{#1}}
  \NewDocumentCommand{\samnote}{m}{\todo[inline,color=green!30]{#1}}
  \begin{document}
  \title{Supplementary Material for\\ ``Revisiting Soundness for Occurrence Typing, Semantically''}
  \maketitle
  \def\AppendixFinish{\end{document}}
\fi

We include the full definition of the $\lamtrl$ formalism, including some
figures omitted in the main paper. Theorems and proofs for semantic type
soundness are located in \Cref{appendix:sec:proof}.

\def\Abs{\abs}
\newcommand{\sep}{\begin{center} $\ast$~$\ast$~$\ast$~$\ast$~$\ast$~$\ast$ \end{center}}
\NewDocumentCommand{\lstpair}{m O{#1}}{\listof{#1}}
\section{Syntax}
\showsyntaxfigure
\showruntimesyntaxfigure
\showerasefigure
\clearpage
\section{Type system}
\showctxdomainfigure
\showwffigure
\showproofsystemfigure
\showrestrictremovefigure
\showsubtypingfigure
\showtypingrulesfigure
\showbasetypesfigure
\showobjectsubstitutionfigure
\clearpage
\section{Operational semantics}
\showexecopfigure
\showreductionfigure
\clearpage
\section{Logical relations}
\showlogicalrelationsfigure
\clearpage
\section{Proof}\label{appendix:sec:proof}
\paragraph{Presuppositions.} We make the following presuppositions:
\presuppositions

\begin{lemma}[Determinacy of One-Step Reduction]\label{appendix:lem:one-step-determ}
  If $\eestepsto{\unexpr}{\unexpr_1}$ and $\eestepsto{\unexpr}{\unexpr_2}$, then
  $\unexpr_1 = \unexpr_2$.
\end{lemma}
\begin{proof}
  By simultaneous case analysis on the derivations of
  $\eestepsto{\unexpr}{\unexpr_1}$ and $\eestepsto{\unexpr}{\unexpr_2}$.
\end{proof}

\begin{lemma}[Determinacy of Reduction to Values]\label{appendix:lem:red-val-determ}
  If $\neestepsto{\unexpr}{v}[m]$ and $\neestepsto{\unexpr}{w}[n]$, then
  $m = n$ and $v = w$.
\end{lemma}
\begin{proof}
  By \Cref{appendix:lem:one-step-determ} and irreducibility of values.
\end{proof}

\begin{lemma}[Well-formedness Weakening]\label{appendix:lem:wf-weakening} Assuming $\wfctx[{\ctx'}]$ and
  $\dom{\ctx} \subseteq \dom{\ctx'}$, the following simultaneously hold:
  \begin{enumerate}
  \item if $\wftype{\tyA}$, then $\wftype[\ctx']{\tyA}$;
  \item if $\wfprop{\prop}$, then $\wfprop[\ctx']{\prop}$;
  \item if $\wfres{\R}$, then $\wfres[\ctx']{\R}$;
  \item if $\wfobj{\obj}$, then $\wfobj[\ctx']{\obj}$.
  \end{enumerate}
\end{lemma}
\begin{proof}
  \newcommand{\wfpropp}{\wfprop[\ctx']}
  \newcommand{\wfobjp}{\wfobj[\ctx']}
  \newcommand{\wftypep}{\wftype[\ctx']}
  Assume $\wfctx[{\ctx'}]$ and $\dom{\ctx} \subseteq \dom{\ctx'}$.
  By simultaneous rule induction on the derivations of the well-formedness
  judgments appearing in the premises of the four clauses.

  \paragraph{Clause 1.} Case on the derivation of $\wftype{\tyA}$.
  \begin{caselist}
  \item $\tyA \in \{\top,\bot,\true,\false,\Nat,\Str\}$.
    The last rules have no premises.
    By reapplying the same rules.
  \item $\tyA = \uniontype{\tyA_1}{\tyA_2}$.
    The premises of the last rule are
    \begin{align*}
      \wftype{\tyA_1} \\
      \wftype{\tyA_2}.
    \end{align*}
    By the mutual IH,
    $\wftype[\ctx']{\tyA_1}$ and $\wftype[\ctx']{\tyA_2}$.
    Then by re-applying the same rule,
    $\wftype[\ctx']{\uniontype{\tyA_1}{\tyA_2}}$
  \item $\tyA = \funty{x}{\tyB}{\R}$. The premises of the last rule are
    \begin{align*}
      \wftype{\tyB} \\
      \wfres[\ctx \exta \bindvar[x][\tyB]]{\R}.
    \end{align*}
    Because of the assumption of
    the lemma $\dom{\ctx} \subseteq \dom{\ctx'}$,
    $\dom{\ctx \exta \bindvar[x][\tyB]} \subseteq \dom{\ctx' \exta
      \bindvar[x][\tyB]}$. Then by the mutual IH on the premise,
    $\wfres[\ctx' \exta \bindvar[x][\tyB]]{\R}$. By reapplying the same rule,
    $\wftype[\ctx']{\funty{x}{\tyB}{\R}}$.

  \end{caselist}

  \paragraph{Clause 2.} Case on the derivation of $\wfprop{\prop}$.
  \begin{caselist}
  \item $\prop \in \{\TT,\FF\}$.
    The last rules have no premises.
    By reapplying the same rules.
  \item $\prop = \typeprop{x}{\tyA}$.
    The premises of the last rule are
    \begin{align*}
      x \in \dom{\ctx} \\
      \wftype{\tyA}.
    \end{align*}
    Because $\dom{\ctx} \subseteq \dom{\ctx'}$, $x \in \dom{\ctx'}$.
    By the mutual IH, $\wftypep{\tyA}$.
    By reapplying the same rule, $\wfpropp{\typeprop{x}{\tyA}}$
  \item $\prop = \nottypeprop{x}{\tyA}$.
    The premises of the last rule are
    \begin{align*}
      x \in \dom{\ctx} \\
      \wftype{\tyA}.
    \end{align*}
    The case follows the reasoning analogous to that in the previous case.
  \item $\prop = \eqprop{x}{y}$.
    The premises of the last rule are
    \begin{align*}
      x \in \dom{\ctx} \\
      y \in \dom{\ctx}.
    \end{align*}
    Because $\dom{\ctx} \subseteq \dom{\ctx'}$, $x \in \dom{\ctx'}$ and $y \in \dom{\ctx'}$.
    By reapplying the same rule, $\wfpropp{\eqprop{x}{y}}$.
  \item $\prop = \conjprop{\prop_1}{\prop_2}$.
    The premises of the last rule are
    \begin{align*}
      \wfprop{\prop_1} \\
      \wfprop{\prop_2}.
    \end{align*}
    By the mutual IH,
    \begin{align*}
      \wfpropp{\prop_1} \\
      \wfpropp{\prop_2}.
    \end{align*}
    By reapplying the same rule, $\conjprop{\prop_1}{\prop_2}$.
  \item $\prop = \disjprop{\prop_1}{\prop_2}$.
    The premises of the last rule are
    \begin{align*}
      \wfprop{\prop_1} \\
      \wfprop{\prop_2}.
    \end{align*}
    The case follows the reasoning analogous to that in the previous case.
  \end{caselist}

  \paragraph{Clause 3.} Case on the derivation of $\wfres{\R}$.
  \begin{caselist}
  \item $\R = \res[\tyA][\posprop][\negprop][\obj]$.
    The premises of the last rule are
    \begin{align*}
      \wftype{\tyA} \\
      \wfprop{\posprop} \\
      \wfprop{\negprop} \\
      \wfobj{\obj}.
    \end{align*}
    By the mutual IH,
    \begin{align*}
      \wftypep{\tyA} \\
      \wfpropp{\posprop} \\
      \wfpropp{\negprop} \\
      \wfobjp{\obj}.
    \end{align*}
    By reapplying the same rule, $\wfres[\ctx']{\res[\tyA][\posprop][\negprop][\obj]}$.
  \end{caselist}

  \paragraph{Clause 4.} Case on the derivation of $\wfobj{\obj}$.
  \begin{caselist}
  \item $\obj = x$.
    The premise of the last rule is
    \begin{align*}
      x \in \dom{\ctx}.
    \end{align*}
    Because $\dom{\ctx} \subseteq \dom{\ctx'}$, $x \in \dom{\ctx'}$.
    By reapplying the same rule,
    $\wfobjp{x}$.
  \item $\obj = \botobj$.
    The last rule has no premises.
    By reapplying the same rule.
  \item $\obj = \topobj$.
    The last rule has no premises.
    By reapplying the same rule.
  \end{caselist}
\end{proof}

\begin{lemma}[Subtyping for Erasure] \label{appendix:lem:subtyping-erase} The following hold simultaneously:
  \begin{enumerate}
  \item if $\wftype{\tyA}$, then $\subtype{\ctx}{\negerase{\tyA}}{\tyA}$ and $\subtype{\ctx}{\tyA}{\poserase{\tyA}}$.
  \item if $\wfres{\R}$,then $\subtype{\ctx}{\negerase{\R}}{\R}$ and $\subtype{\ctx}{\R}{\poserase{\R}}$.
  \item Assuming $\wfprop{\prop}$, the following two statements hold: 1. if $\proves{\ctx}{\negerase{\prop}}$ then $\proves{\ctx}{\prop}$;
    2.if $\proves{\ctx}{\prop}$ then $\proves{\ctx}{\poserase{\prop}}$.
  \item if $\wfobj{\obj}$, $\subtype{\ctx}{\negerase{\obj}}{\obj}$ and $\subtype{\ctx}{\obj}{\poserase{\obj}}$.
  \end{enumerate}
\end{lemma}
\begin{proof}
  By simultaneous case analysis on the shapes of $\tyA$, $\R$, $\prop$, and $\obj$.
  \paragraph{Clause 1.} Case on the shape of $\tyA$.
  \begin{caselist}
  \item $\tyA \in \{\top,\bot,\true,\false,\Nat,\Str\}$.
    By definition, $\negerase{\tyA} = \tyA = \poserase{\tyA}$. Then by \myrule{S-Refl},
    $\subtype{\ctx}{\negerase{\tyA}}{\tyA}$ and $\subtype{\ctx}{\tyA}{\poserase{\tyA}}$.
  \item $\tyA = \uniontype{\tyA_1}{\tyA_2}$.
    By the mutual IH, for $i \in \{1,2\}$
    \begin{align}
      \subtype{\ctx}{\negerase{\tyA_i}}{\tyA_i} \\
      \subtype{\ctx}{\tyA_i}{\poserase{\tyA_i}}
    \end{align}
    By \myrule{S-UnionSub} and \myrule{S-UnionSuper},
    \begin{align*}
      \subtype{\ctx}{\uniontype{\negerase{\tyA_1}}{\negerase{\tyA_2}}}{\uniontype{\tyA_1}{\tyA_2}} \\
      \subtype{\ctx}{\uniontype{\tyA_1}{\tyA_2}}{\uniontype{\poserase{\tyA_1}}{\poserase{\tyA_2}}}
    \end{align*}
  \item $\tyA = \funty{y}{\tyB}{\R}$. We need to show $\subtypeneg{(\funty{y}{\tyB}{\R})}$ and $\subtypepos{(\funty{y}{\tyB}{\R})}$.
    \newcommand{\ctxyb}{\ctx \exta \bindvar[y][\tyB]}
    \newcommand{\ctxynegb}{\ctx \exta \bindvar[y][\negerase{\tyB}]}
    By well-formedness, $\wftype{\tyB}$ and $\wfres[\ctx \exta \bindvar[y][\tyB]]{\synR}$.
    By the mutual IH,
    \begin{align}
      \subtypeneg{\tyB} \label{appendix:eq:subtypenegtyb-} \\
      \subtypepos{\tyB} \label{appendix:eq:subtypepostyb-} \\
      \subtypeneg[\ctxyb]{\R} \label{appendix:eq:subtypenegctxybr-} \\
      \subtypepos[\ctxyb]{\R} \label{appendix:eq:subtypeposctxybr-}
    \end{align}
    By \myrule{S-Fun} on \cref{appendix:eq:subtypepostyb-,appendix:eq:subtypenegctxybr-},
    \begin{align*}
      \subtype{\ctx}{\funty{y}{\poserase{\tyB}}{\negerase\R}}{\funty{y}{\tyB}{\R}}
    \end{align*}
    By definition,
    \begin{align*}
      \subtypeneg{(\funty{y}{\tyB}{\R})}
    \end{align*}
    which concludes the first part of our goal.

    By \Cref{appendix:lem:wf-weakening} on $\wfres[\ctxyb]{\R}$ and
    $\dom{\ctxyb} \subseteq \dom{\ctxynegb}$,
    \begin{align*}
      \wfres[\ctxynegb]{\R}.
    \end{align*}
    Then by the mutual IH,
    \begin{align*}
      \subtypepos[\ctxynegb]{\R}.
    \end{align*}
    Together with \cref{appendix:eq:subtypenegtyb-}, applying \myrule{S-Fun} gives
    \begin{align*}
      \subtype{\ctx}{\funty{y}{\tyB}{\R}}{\funty{y}{\negerase{\tyB}}{\poserase\R}}
    \end{align*}
    By definition,
    \begin{align*}
      \subtypepos{(\funty{y}{\tyB}{\R})}
    \end{align*}
    which concludes the second part of our goal.
  \end{caselist}
  \paragraph{Clause 2.} Case on the shape of $\R$.
  The only case is $\R = \res[\tyA][\posprop][\negprop][\obj]$.
  The premises of the last rule are
  \begin{align*}
    \wftype{\tyA} \\
    \wfprop{\posprop} \\
    \wfprop{\negprop} \\
    \wfobj{\obj}.
  \end{align*}
  We need to show,
  \begin{align*}
    \subtype{\ctx}{\negerase{(\res[\tyA][\posprop][\negprop][\obj])}}{\res[\tyA][\posprop][\negprop][\obj]} \\
    \subtype{\ctx}{\res[\tyA][\posprop][\negprop][\obj]}{\poserase{(\res[\tyA][\posprop][\negprop][\obj])}}.
  \end{align*}

  By the mutual IH on the premises respectively,
  \begin{align}
    \subtype{\ctx}{\negerase{\tyA}}{\tyA} \label{appendix:eq:subtyp--2}\\
    \subtype{\ctx}{\tyA}{\poserase{\tyA}} \label{appendix:eq:subtyp--3}\\
    \proves{\ctx \exta \negerase{\posprop}}{\posprop} \label{appendix:eq:prov-exta-neger}\\
    \proves{\ctx \exta \posprop}{\poserase{\posprop}} \label{appendix:eq:prov-exta-pospr}\\
    \proves{\ctx \exta \negerase{\negprop}}{\negprop} \label{appendix:eq:prov-exta-neger-1}\\
    \proves{\ctx \exta \negprop}{\poserase{\negprop}} \label{appendix:eq:prov-exta-negpr}\\
    \subtype{\ctx}{\negerase{\obj}}{\obj} \label{appendix:eq:subtyp--4}\\
    \subtype{\ctx}{\obj}{\poserase{\obj}}\label{appendix:eq:subtyp--5}
  \end{align}
  By \myrule{SR-Result} on \cref{appendix:eq:subtyp--2,appendix:eq:prov-exta-neger,appendix:eq:prov-exta-neger-1,appendix:eq:subtyp--4},
  \begin{align*}
    \subtype{\ctx}{\negerase{(\res[\tyA][\posprop][\negprop][\obj])}}{\res[\tyA][\posprop][\negprop][\obj]}
  \end{align*}
  By \myrule{SR-Result} on \cref{appendix:eq:subtyp--3,appendix:eq:prov-exta-pospr,appendix:eq:prov-exta-negpr,appendix:eq:subtyp--5},
  \begin{align*}
    \subtype{\ctx}{\res[\tyA][\posprop][\negprop][\obj]}{\poserase{(\res[\tyA][\posprop][\negprop][\obj])}}
  \end{align*}
  \paragraph{Clause 3.} Case on the shape of $\prop$.
  \begin{caselist}
  \item $\prop = \typeprop{z}{\tyA}$. The premises of the last rule are
    \begin{align}
      z \in \dom{\ctx} \\
      \wftype{\tyA}\label{appendix:eq:wftypetya-}
    \end{align}
    By the mutual IH on \cref{appendix:eq:wftypetya-},
    \begin{align}
      \subtypeneg{\tyA} \label{appendix:eq:subtypenegtya-}\\
      \subtypepos{\tyA}\label{appendix:eq:subtypepostya-}
    \end{align}

    1. Assume $\proves{\ctx}{\negerase{(\typeprop{z}{\tyA})}}$. If $z = x$,
    then $\negerase{(\typeprop{z}{\tyA})} = \FF$. Hence
    $\proves{\ctx}{\FF}$, and by \myrule{L-False},
    \begin{align*}
      \proves{\ctx}{\typeprop{z}{\tyA}}.
    \end{align*}
    Otherwise, suppose $z \neq x$. Then
    $\negerase{(\typeprop{z}{\tyA})} = \typeprop{z}{\negerase{\tyA}}$.
    By \myrule{L-Sub} on the assumption
    $\proves{\ctx}{\typeprop{z}{\negerase{\tyA}}}$ and
    \cref{appendix:eq:subtypenegtya-},
    \begin{align*}
      \proves{\ctx}{\typeprop{z}{\tyA}}.
    \end{align*}

    2. Assume $\proves{\ctx}{\typeprop{z}{\tyA}}$. If $z = x$, then
    $\poserase{(\typeprop{z}{\tyA})} = \TT$. By \myrule{L-True},
      \begin{align*}
        \proves{\ctx}{\TT}.
      \end{align*}
    Otherwise, suppose $z \neq x$. Then
    $\poserase{(\typeprop{z}{\tyA})} = \typeprop{z}{\poserase{\tyA}}$.
    By \myrule{L-Sub} on the assumption
    $\proves{\ctx}{\typeprop{z}{\tyA}}$ and \cref{appendix:eq:subtypepostya-},
      \begin{align*}
        \proves{\ctx}{\poserase{(\typeprop{z}{\tyA})}}.
      \end{align*}
  \item $\prop = \nottypeprop{z}{\tyA}$. The premises of the last rule are
    \begin{align}
      z \in \dom{\ctx} \\
      \wftype{\tyA}\label{appendix:eq:wftypetya-not}
    \end{align}
    By the mutual IH on \cref{appendix:eq:wftypetya-not},
    \begin{align}
      \subtypeneg{\tyA} \label{appendix:eq:subtypenegtya-not}\\
      \subtypepos{\tyA}\label{appendix:eq:subtypepostya-not}
    \end{align}

    1. Assume $\proves{\ctx}{\negerase{(\nottypeprop{z}{\tyA})}}$. If $z = x$,
    then $\negerase{(\nottypeprop{z}{\tyA})} = \FF$. Hence
    $\proves{\ctx}{\FF}$, and by \myrule{L-False},
    \begin{align*}
      \proves{\ctx}{\nottypeprop{z}{\tyA}}.
    \end{align*}
    Otherwise, suppose $z \neq x$. Then
    $\negerase{(\nottypeprop{z}{\tyA})} = \nottypeprop{z}{\poserase{\tyA}}$.
    By \myrule{L-SubNot} on the assumption
    $\proves{\ctx}{\nottypeprop{z}{\poserase{\tyA}}}$ and
    \cref{appendix:eq:subtypepostya-not},
    \begin{align*}
      \proves{\ctx}{\nottypeprop{z}{\tyA}}.
    \end{align*}

    2. Assume $\proves{\ctx}{\nottypeprop{z}{\tyA}}$. If $z = x$, then
    $\poserase{(\nottypeprop{z}{\tyA})} = \TT$. By \myrule{L-True},
      \begin{align*}
        \proves{\ctx}{\TT}.
      \end{align*}
    Otherwise, suppose $z \neq x$. Then
    $\poserase{(\nottypeprop{z}{\tyA})} = \nottypeprop{z}{\negerase{\tyA}}$.
    By \myrule{L-SubNot} on the assumption
    $\proves{\ctx}{\nottypeprop{z}{\tyA}}$ and \cref{appendix:eq:subtypenegtya-not},
      \begin{align*}
        \proves{\ctx}{\poserase{(\nottypeprop{z}{\tyA})}}.
      \end{align*}
  \item $\prop = \eqprop{z_1}{z_2}$. The premises of the last rule are
    \begin{align*}
      z_1 \in \dom{\ctx} \\
      z_2 \in \dom{\ctx}
    \end{align*}

    1. Assume $\proves{\ctx}{\negerase{(\eqprop{z_1}{z_2})}}$. If
    $z_1 = x$ or $z_2 = x$, then
    $\negerase{(\eqprop{z_1}{z_2})} = \FF$. Hence $\proves{\ctx}{\FF}$,
    and by \myrule{L-False},
    \begin{align*}
      \proves{\ctx}{\eqprop{z_1}{z_2}}.
    \end{align*}
    Otherwise, suppose $z_1 \neq x$ and $z_2 \neq x$. Then
    $\negerase{(\eqprop{z_1}{z_2})} = \eqprop{z_1}{z_2}$. Then the case is
    concluded by the assumption.

    2. Assume $\proves{\ctx}{\eqprop{z_1}{z_2}}$. If $z_1 = x$ or
    $z_2 = x$, then $\poserase{(\eqprop{z_1}{z_2})} = \TT$. By
    \myrule{L-True},
    \begin{align*}
      \proves{\ctx}{\TT}.
    \end{align*}
    Otherwise, suppose $z_1 \neq x$ and $z_2 \neq x$. Then
    $\poserase{(\eqprop{z_1}{z_2})} = \eqprop{z_1}{z_2}$. Then the case is concluded
    by the assumption.
  \item $\prop = \conjprop{\prop_1}{\prop_2}$. The premises of the last rule are
    \begin{align}
      \wfprop{\prop_1}\label{appendix:eq:wfprop-conj-1} \\
      \wfprop{\prop_2}\label{appendix:eq:wfprop-conj-2}.
    \end{align}

    1. Assume $\proves{\ctx}{\negerase{(\conjprop{\prop_1}{\prop_2})}}$.
    Since
    $\negerase{(\conjprop{\prop_1}{\prop_2})}
    = \conjprop{\negerase{\prop_1}}{\negerase{\prop_2}}$, by
    \myrule{L-AndE1} and \myrule{L-AndE2} on the assumption,
    \begin{align}
      \proves{\ctx}{\negerase{\prop_1}}\label{appendix:eq:prov-conj-neg-1} \\
      \proves{\ctx}{\negerase{\prop_2}}\label{appendix:eq:prov-conj-neg-2}.
    \end{align}
    By the mutual IH on \cref{appendix:eq:wfprop-conj-1,appendix:eq:wfprop-conj-2} and
    \cref{appendix:eq:prov-conj-neg-1,appendix:eq:prov-conj-neg-2},
    \begin{align}
      \proves{\ctx}{\prop_1}\label{appendix:eq:prov-conj-1} \\
      \proves{\ctx}{\prop_2}\label{appendix:eq:prov-conj-2}.
    \end{align}
    By \myrule{L-AndI} on \cref{appendix:eq:prov-conj-1,appendix:eq:prov-conj-2},
    \begin{align*}
      \proves{\ctx}{\conjprop{\prop_1}{\prop_2}}.
    \end{align*}

    2. Assume $\proves{\ctx}{\conjprop{\prop_1}{\prop_2}}$. By
    \myrule{L-AndE1} and \myrule{L-AndE2} on the assumption,
    \begin{align}
      \proves{\ctx}{\prop_1}\label{appendix:eq:prov-conj-pos-1} \\
      \proves{\ctx}{\prop_2}\label{appendix:eq:prov-conj-pos-2}.
    \end{align}
    By the mutual IH on \cref{appendix:eq:wfprop-conj-1,appendix:eq:wfprop-conj-2} and
    \cref{appendix:eq:prov-conj-pos-1,appendix:eq:prov-conj-pos-2},
    \begin{align}
      \proves{\ctx}{\poserase{\prop_1}}\label{appendix:eq:prov-conj-poserase-1} \\
      \proves{\ctx}{\poserase{\prop_2}}\label{appendix:eq:prov-conj-poserase-2}.
    \end{align}
    By
    \myrule{L-AndI} on
    \cref{appendix:eq:prov-conj-poserase-1,appendix:eq:prov-conj-poserase-2},
    \begin{align*}
      \proves{\ctx}{\poserase{(\conjprop{\prop_1}{\prop_2})}}.
    \end{align*}
  \item $\prop = \disjprop{\prop_1}{\prop_2}$. The premises of the last rule are
    \begin{align}
      \wfprop{\prop_1}\label{appendix:eq:wfprop-disj-1} \\
      \wfprop{\prop_2}\label{appendix:eq:wfprop-disj-2}.
    \end{align}

    1. Assume $\proves{\ctx}{\negerase{(\disjprop{\prop_1}{\prop_2})}}$.
    Since
    $\negerase{(\disjprop{\prop_1}{\prop_2})}
    = \disjprop{\negerase{\prop_1}}{\negerase{\prop_2}}$, the assumption is
    \begin{align}
      \proves{\ctx}{\disjprop{\negerase{\prop_1}}{\negerase{\prop_2}}}.
      \label{appendix:eq:prov-disj-neg}
    \end{align}
    Let us consider the left branch.

    By \Cref{appendix:lem:wf-weakening} on \cref{appendix:eq:wfprop-disj-1,appendix:eq:wfprop-disj-2}
    \begin{align}
      \wfprop[\ctx \extb \negerase{\prop_1}]{\prop_1} \label{appendix:eq:wfpr-extb-neger}\\
      \wfprop[\ctx \extb \negerase{\prop_1}]{\prop_2} \label{appendix:eq:wfpr-extb-neger-1}.
    \end{align}
    By the mutual IH on \cref{appendix:eq:wfpr-extb-neger}
    \begin{align}
      \proves{\ctx \extb \negerase{\prop_1}}{\negerase{\prop_1}} \metaimpl \proves{\ctx\extb \negerase{\prop_1}}{\prop_1}
      \label{appendix:eq:prov-disj-neg-left}
    \end{align}
    By L-Atom,
    \begin{align*}
      \proves{\ctx \extb \negerase{\prop_1}}{\negerase{\prop_1}}
    \end{align*}
    By applying \cref{appendix:eq:prov-disj-neg-left}
    \begin{align}
      \proves{\ctx\extb \negerase{\prop_1}}{\prop_1}
    \end{align}
    Together with \cref{appendix:eq:wfpr-extb-neger-1},
    applying \myrule{L-OrI1} gives
    \begin{align}
      \proves{\ctx \extb \negerase{\prop_1}}{\disjprop{\prop_1}{\prop_2}}.
      \label{appendix:eq:prov-disj-neg-left-branch}
    \end{align}
    Let us consider the right branch.

    By \Cref{appendix:lem:wf-weakening} on \cref{appendix:eq:wfprop-disj-1,appendix:eq:wfprop-disj-2}
    \begin{align}
      \wfprop[\ctx \extb \negerase{\prop_2}]{\prop_1} \label{appendix:eq:wfpr-extb-neger-2}\\
      \wfprop[\ctx \extb \negerase{\prop_2}]{\prop_2} \label{appendix:eq:wfpr-extb-neger-3}.
    \end{align}
    By the mutual IH on \cref{appendix:eq:wfpr-extb-neger-3}
    \begin{align}
      \proves{\ctx \extb \negerase{\prop_2}}{\negerase{\prop_2}} \metaimpl \proves{\ctx\extb \negerase{\prop_2}}{\prop_2}
      \label{appendix:eq:prov-disj-neg-right}
    \end{align}
    By L-Atom,
    \begin{align*}
      \proves{\ctx \extb \negerase{\prop_2}}{\negerase{\prop_2}}
    \end{align*}
    By applying \cref{appendix:eq:prov-disj-neg-right}
    \begin{align}
      \proves{\ctx\extb \negerase{\prop_2}}{\prop_2}
    \end{align}
    Together with \cref{appendix:eq:wfpr-extb-neger-2},
    applying \myrule{L-OrI2} gives
    \begin{align}
      \proves{\ctx \extb \negerase{\prop_2}}{\disjprop{\prop_1}{\prop_2}}.
      \label{appendix:eq:prov-disj-neg-right-branch}
    \end{align}
    By \myrule{L-OrE} on
    \cref{appendix:eq:prov-disj-neg,appendix:eq:prov-disj-neg-left-branch,appendix:eq:prov-disj-neg-right-branch},
    \begin{align*}
      \proves{\ctx}{\disjprop{\prop_1}{\prop_2}}.
    \end{align*}

    2. Assume $\proves{\ctx}{\disjprop{\prop_1}{\prop_2}}$.
    Let us consider the left branch.

    By \Cref{appendix:lem:wf-weakening} on \cref{appendix:eq:wfprop-disj-1,appendix:eq:wfprop-disj-2}
    \begin{align}
      \wfprop[\ctx \extb \prop_1]{\prop_1} \label{appendix:eq:wfpr-extb-pos-1}\\
      \wfprop[\ctx \extb \prop_1]{\prop_2} \label{appendix:eq:wfpr-extb-pos-2}.
    \end{align}
    By the mutual IH on \cref{appendix:eq:wfpr-extb-pos-1}
    \begin{align}
      \proves{\ctx \extb \prop_1}{\prop_1} \metaimpl \proves{\ctx\extb \prop_1}{\poserase{\prop_1}}
      \label{appendix:eq:prov-disj-pos-left}
    \end{align}
    By L-Atom,
    \begin{align*}
      \proves{\ctx \extb \prop_1}{\prop_1}
    \end{align*}
    By applying \cref{appendix:eq:prov-disj-pos-left}
    \begin{align}
      \proves{\ctx\extb \prop_1}{\poserase{\prop_1}}
      \label{appendix:eq:wfpr-extb-pos-poserase-2}
    \end{align}
    By \cref{appendix:eq:wfpr-extb-pos-2},
    \begin{align*}
      \wfprop[\ctx \extb \prop_1]{\poserase{\prop_2}}
    \end{align*}
    Together with \cref{appendix:eq:wfpr-extb-pos-poserase-2},
    applying \myrule{L-OrI1} gives
    \begin{align}
      \proves{\ctx \extb \prop_1}{\disjprop{\poserase{\prop_1}}{\poserase{\prop_2}}}
      \label{appendix:eq:prov-disj-pos-left-branch}
    \end{align}
    Let us consider the right branch.

    By \Cref{appendix:lem:wf-weakening} on \cref{appendix:eq:wfprop-disj-1,appendix:eq:wfprop-disj-2}
    \begin{align}
      \wfprop[\ctx \extb \prop_2]{\prop_1} \label{appendix:eq:wfpr-extb-pos-3}\\
      \wfprop[\ctx \extb \prop_2]{\prop_2} \label{appendix:eq:wfpr-extb-pos-4}.
    \end{align}
    By the mutual IH on \cref{appendix:eq:wfpr-extb-pos-4}
    \begin{align}
      \proves{\ctx \extb \prop_2}{\prop_2} \metaimpl \proves{\ctx\extb \prop_2}{\poserase{\prop_2}}
      \label{appendix:eq:prov-disj-pos-right}
    \end{align}
    By L-Atom,
    \begin{align*}
      \proves{\ctx \extb \prop_2}{\prop_2}
    \end{align*}
    By applying \cref{appendix:eq:prov-disj-pos-right}
    \begin{align}
      \proves{\ctx\extb \prop_2}{\poserase{\prop_2}}
      \label{appendix:eq:wfpr-extb-pos-poserase-1}
    \end{align}
    By \cref{appendix:eq:wfpr-extb-pos-3}
    \begin{align*}
      \wfprop[\ctx \extb \prop_2]{\poserase{\prop_1}}.
    \end{align*}
    Together with \cref{appendix:eq:wfpr-extb-pos-poserase-1},
    applying \myrule{L-OrI2} gives
    \begin{align}
      \proves{\ctx \extb \prop_2}{\disjprop{\poserase{\prop_1}}{\poserase{\prop_2}}}.
      \label{appendix:eq:prov-disj-pos-right-branch}
    \end{align}
    By \myrule{L-OrE} on the assumption,
    \cref{appendix:eq:prov-disj-pos-left-branch,appendix:eq:prov-disj-pos-right-branch},
    \begin{align*}
      \proves{\ctx}{\poserase{(\disjprop{\prop_1}{\prop_2})}}.
    \end{align*}
  \item $\prop = \TT$. since $\negerase{\TT} = \poserase{\TT} = \TT$, the case holds by \myrule{L-True}.
  \item $\prop = \FF$. since $\negerase{\FF} = \poserase{\FF} = \FF$, the case holds by \myrule{L-False}.
  \end{caselist}
  \paragraph{Clause 4.} Case on the shape of $\obj$.
  \begin{caselist}
  \item $\obj = z$. The premise of the last rule is
    \begin{align}
      z \in \dom{\ctx} \label{appendix:eq:z-in-domctx}
    \end{align}
    If $z = x$, then $\negerase{z} = \botobj$ and $\poserase{z} = \topobj$.
    Then by \myrule{SO-Bot} and \myrule{SO-Top} respectively,
    \begin{align*}
      \subtype{\ctx}{\botobj}{z} \\
      \subtype{\ctx}{z}{\topobj}.
    \end{align*}
    Otherwise, $\negerase{z} = z = \poserase{z}$.
    By \myrule{L-Refl} on \cref{appendix:eq:z-in-domctx} and then \myrule{SO-Equiv}
    \begin{align*}
      \subtype{\ctx}{z}{z}.
    \end{align*}
  \item $\obj = \botobj$. The last rule has no premises. Since
    $\negerase{\botobj} = \botobj = \poserase{\botobj}$,  by \myrule{SO-Bot}
    \begin{align*}
      \subtype{\ctx}{\botobj}{\botobj}.
    \end{align*}
  \item $\obj = \topobj$. The last rule has no premises. Since
    $\negerase{\topobj} = \topobj = \poserase{\topobj}$, by \myrule{SO-Top}
    \begin{align*}
      \subtype{\ctx}{\topobj}{\topobj}.
    \end{align*}
  \end{caselist}
\end{proof}
\begin{lemma}
  [Weakening for closing substitution]\label{appendix:lem:psubwk} The following hold simultaneously:
    \begin{enumerate}
    \item if $\env(y) = v$, then $\envw(y) = v$ for any variable $y \neq x$;
    \item if $\valinterp{\tyA}{v}$, then $\valinterp{\tyA}{v}[n][\envw]$
    \item if $\propinterp{\prop}$, then $\propinterp{\prop}[n][\envw]$;
    \item if $\expinterp{\R}{\unexpr}$, then $\expinterp{\R}{\unexpr}[n][\envw]$.
  \end{enumerate}
\end{lemma}
\begin{proof}
  All statements follow the definition of closing substitution and logical relations.
\end{proof}

\begin{lemma}
  [Strengthening for parallel substitution]\label{appendix:lem:psubstr} The following hold simultaneously:
  \begin{enumerate}
  \item if $\envw(y) = v$ and $y \neq x$, then $\env(y) = v$;
  \item if $\valinterp{\tyA}{v}[n][\envw]$ and $\notfreein{x}{\tyA}$, then $\valinterp{\tyA}{v}[n]$;
  \item if $\propinterp{\prop}[n][\envw]$ and $\notfreein{x}{\prop}$, then $\propinterp{\prop}[n]$;
  \item if $\expinterp{\R}{\unexpr}[n][\envw]$ and $\notfreein{x}{\R}$, then $\expinterp{\R}{\unexpr}[n]$.
  \end{enumerate}
\end{lemma}
\begin{proof}
  All statements follow the definition of closing substitution and logical relations.
\end{proof}

\begin{lemma}[Renaming for Closing Substitution]\label{appendix:lem:psubrenaming} The following hold simultaneously:
  \begin{enumerate}
    \item if $x \neq y$, $\env(x) = v$ and $\env(x) = \env(y)$ , then $\env(y) = v$;
    \item if $\valinterp{\tyA}{v}$ and $\env(x) = \env(y)$, then $\valinterp{\substR{\tyA}{x}{y}}{v}$;
    \item if $\propinterp{\prop}$ and $\env(x) = \env(y)$, then $\propinterp{(\substprop{\prop}{x}{y})}$;
    \item if $\expinterp{\R}{\unexpr}$ and $\env(x) = \env(y)$, then $\expinterp{\substR{\R}{x}{y}}{\unexpr}$.
  \end{enumerate}
\end{lemma}
\begin{proof}
  All statements follow the definition of closing substitution and logical relations.
\end{proof}

\begin{lemma}[Head Expansion]\label{appendix:lem:hdxp} \hfill
  \begin{enumerate}
    \item
      if $\expinterp{\R}{\unexpr'}[m]$ and $\neestepsto{\unexpr}{\unexpr'}[k]$, then $\expinterp{\R}{\unexpr}[m+k]$.
    \item
      if $\valinterp{\R}{v}[m]$ and $\neestepsto{\unexpr}{v}[k]$, then $\expinterp{\R}{\unexpr}[m+k]$.
  \end{enumerate}
\end{lemma}
\begin{proof}
  The results follow immediately by unfolding the definition of $\expinterp{\R}{\unexpr}$
\end{proof}

\begin{lemma}[Head Reduction]\label{appendix:lem:hdrd} if $\expinterp{\R}{\unexpr}$,
  then either:
  \begin{enumerate}
  \item $\hdrdpropa$;
  \item $\hdrdpropb$.
  \end{enumerate}
\end{lemma}
\begin{proof}
  The result follows immediately by unfolding the definition of $\expinterp{\R}{\unexpr}$.
\end{proof}

\begin{lemma}[Monotonicity of Logical Relations]\label{appendix:lem:srlmono} $\forall m, n $ such that $m \le n$, then the following hold:
  \begin{enumerate}
    \item \label{appendix:lem:srlmono-exp} if $\expinterp{\R}{\unexpr}$, then $\expinterp{\R}{\unexpr}[m]$.
    \item \label{appendix:lem:srlmono-res} if $\valinterp{\R}{v}$, then $\valinterp{\R}{v}[m]$.
    \item \label{appendix:lem:srlmono-type} if $\valinterp{\tyA}{v}$, then $\valinterp{\tyA}{v}[m]$.
    \item \label{appendix:lem:srlmono-prop} if $\propinterp{\prop}$, then $\propinterp{\prop}[m]$.
    \item \label{appendix:lem:srlmono-ctx} if $\ctxinterp{\ctx}{\env}$, then $\ctxinterp{\ctx}{\env}[m]$.
  \end{enumerate}
\end{lemma}
\begin{proof}
  First, let us case on $m$.  If $m = 0$, then all clauses trivially hold.
  Otherwise, $m > 0$. Then we prove all clauses simultaneously by induction on
  $n$, with an inner arbitrary $m \leq n$.

  The base case $n = 0$ holds vacuously because $n = 0$ contradicts $m \leq n$ where $m > 0$.
  We move on to the inductive case $n > 0$.
  \paragraph{For Clause~\ref{appendix:lem:srlmono-exp},} there are two cases.
  \begin{caselist}
  \item $\expinterp{\R}{\unexpr} = \expinterp{\R}{\unexpr'}[n-1]$ and $
    \eestepsto{\unexpr}{\unexpr'}$.  Because $m - 1 \le n - 1$, we can apply
    Component~\ref{appendix:lem:srlmono-exp}  of the mutual IH to it. Together with
    $\expinterp{\R}{\unexpr'}[n-1]$, this gives
    $\expinterp{\R}{\unexpr'}[m-1]$. By definition, $\expinterp{\R}{\unexpr}[m]$
  \item $\expinterp{\R}{\unexpr} = \valinterp{\R}{v}$ and $\unexpr = v$. This case is a consequence of Clause~\ref{appendix:lem:srlmono-res}.
  \end{caselist}
  \paragraph{For Clause~\ref{appendix:lem:srlmono-res},} Clause~\ref{appendix:lem:srlmono-type} gives $\valinterp{A}{v}[m]$, and Clause~\ref{appendix:lem:srlmono-prop}
  gives
  $\propertyC{v}{}{\env}{m}$.
  Then, by definition, $\valinterp{\R}{v}[m]$
  \paragraph{For Clause~\ref{appendix:lem:srlmono-type},} we case on the shape of $\tyA$.
  \begin{caselist}
  \item $\tyA \in \{ \true, \false, \Nat, \Str, \top\}$. The clause holds trivially.
  \item $\tyA = \bot$. There is a contradiction. Therefore, the clause holds vacuously.
  \item $\tyA = \funty{x}{\tyB}{\R}$. By definition, $\forall k \le n, w.\: \valinterp{\tyB}{w}[k] \metaimpl \expinterp{\R}{\ap{v}{w}}[k][\extenv{\env}{x}{w}]$.
    Assume $j \le m$ and $\valinterp{\tyB}{w}[j]$. Together with $j \le m \le n$, this gives $\expinterp{\R}{\ap{v}{w}}[j][\extenv{\env}{x}{w}]$. In other words,
    $\valinterp{\funty{x}{\tyB}{\R}}{v}[m]$.
  \item $\tyA = \uniontype{\tyB}{\tyC}$.
    $\valinterp{\tyB}{v}[n] \metaor \valinterp{\tyC}{v}[n]$. Proceed with those
    two cases. In the first case, by Component ~\ref{appendix:lem:srlmono-type} of the mutual IH,
    $\valinterp{\tyB}{v}[{m}]$. By definition,
    $\valinterp{\uniontype{\tyB}{\tyC}}{v}[{m}]$. In the other case,
    $\valinterp{\uniontype{\tyB}{\tyC}}{v}[{m}]$ follows reasoning analogous to
    that in the previous case.
  \end{caselist}
  \paragraph{For Clause~\ref{appendix:lem:srlmono-prop},} we case on the shape of $\prop$.
  \begin{caselist}
  \item $\prop = \TT$. The case trivially holds.
  \item $\prop = \eqprop{x}{y}$. By definition,
    $\propinterp{\eqprop{x}{y}}[m]$
  \item $\prop = \typeprop{x}{\tyA}$ and $\valinterp{\tyA}{v}$. By Clause \ref{appendix:lem:srlmono-type}, $\valinterp{\tyA}{v}[m]$. Therefore,
    $\propinterp{\typeprop{x}{\tyA}}[m]$
  \item $\prop = \nottypeprop{x}{\tyA}$ and $\notvalinterp{\tyA}{v}[l]$.
    Assume, for contradiction, that $\propinterp{\nottypeprop{x}{\tyA}}[m]$ does not hold, i.e. there exists $\betw{j}[1][{m}]$,
    \begin{align}
      \valinterp{\tyA}{v}[j] \label{appendix:eq:valinterptyavl-}
    \end{align}
    Since $j \le m \le n$, \cref{appendix:eq:valinterptyavl-} contradicts the premise $\notvalinterp{\tyA}{v}[l]$.
    Therefore, the contraction shows $\propinterp{\nottypeprop{x}{\tyA}}[m]$.
  \item $\prop = \conjprop{\prop_1}{\prop_2}$ and $\propinterp{\prop_1} \metaand \propinterp{\prop_2}$. By the mutual IH,
    $\propinterp{\prop_1}[m] \metaand \propinterp{\prop_2}[m]$. Therefore, $\propinterp{\conjprop{\prop_1}{\prop_2}}[m]$.
  \item $\prop = \disjprop{\prop_1}{\prop_2}$ and $\propinterp{\prop_1} \metaor \propinterp{\prop_2}$. Proceed with the two cases.
    In the first case, by the mutual IH, $\propinterp{\prop_1}[m]$.  Therefore, $\propinterp{\disjprop{\prop_1}{\prop_2}}[m]$.
    The other case follows reasoning analogous to that in the first case.
  \end{caselist}
  \paragraph{For Clause~\ref{appendix:lem:srlmono-ctx},} we case on the shape of $\ctx$, and then use the
  mutual IH accordingly.
\end{proof}

\NewDocumentCommand{\lemclause}{m}{\Cref{appendix:lem:prfsound}(\ref{appendix:#1})}
\begin{theorem}[Correctness of Subtyping and Proof Systems]\label[theorem]{appendix:lem:prfsound}
  Given $\allsat$, the following statements hold simultaneously:
  \begin{enumerate}
  \item \label{appendix:lem:semressub}
    If $\unexpr$ is closed, $\subtype{\ctx}{\R_a}{\R_b}$, and $\expinterp{\R_a}{\unexpr}$, then $\expinterp{\R_b}{\unexpr}$.
  \item \label{appendix:lem:semsub} If $ \subtype{\ctx}{\tyA}{\tyB}$ and $\valinterp{\tyA}{v}$, then $\valinterp{\tyB}{v}$.
  \item \label{appendix:lem:restrict} If $\restrictrel{\ctx}{\tyA}{\tyB}{\tyC}$, $ \valinterp{\tyA}{v}$, and $\valinterp{\tyB}{v}$, then $\valinterp{\tyC}{v}$.
  \item \label{appendix:lem:remove} if $\removerel{\ctx}{\tyA}{\tyB}{\tyC}$, $ \valinterp{\tyA}{v}$, and $\notvalinterp{\tyB}{v}$, then $\valinterp{\tyC}{v}$
  \item \label{appendix:lem:real1}
    if $\proves{\ctx}{\prop}$, then $\propinterp{\prop}$.
  \end{enumerate}
\end{theorem}
\begin{proof}
  \NewDocumentCommand{\bymih}{m m m}{%
    \StrSubstitute{#1}{,}{,appendix:}[\AppendixRefs]%
    Since \cref{appendix:\AppendixRefs}
    \IfSubStr{#1}{,}{are subderivations}{is a subderivation}, we can apply
    Component \ref{appendix:#2} of the mutual IH to
    \IfSubStr{#1}{,}{them respectively}{it}. Together
    with {#3}, this gives}

  By simultaneous rule induction on the derivations of the mutually defined judgments appearing in the premises of each statement.
  We proceed by cases on the last rule of each derivation.
  \paragraph{Cases for $\subtype{\ctx}{\R_a}{\R_b}$.} Assume the premise $\expinterp{R_a}{\unexpr}$ of the clause.
  By \Cref{appendix:lem:hdrd}, there are two cases:

  1. ${\hdrdpropa}$. \propaconseq.

  2. $\hdrdpropb[\R_a][v_o][k]$. The subgoal is to show
  $\valinterp{\semR_b}{v_o}[n-k]$. Let us proceed by rule induction
  on $\subtype{\ctx}{\R_a}{\R_b}$. The only rule is \myrule{SR-Result}. Let
  $R_a = \res[\tyC_a][\possynprop[a]][\negsynprop[a]][\obj_a]$ and
  $R_b = \res[\tyC_b][\possynprop[b]][\negsynprop[b]][\obj_b]$.  The premises
  are
  \begin{align}
      \subtype{\ctx}{\tyC_a}{\tyC_b}\label{appendix:eq:subtyp-exta-typepr-1} \\
      \proves{\ctx \extb \possynprop[a]}{\possynprop[b]} \label{appendix:eq:prov-exta-typepr}\\
      \proves{\ctx \extb \negsynprop[a]}{\negsynprop[b]}\label{appendix:eq:prov-exta-typepr-2} \\
      \subtype{\ctx}{\obj_a}{\obj_b} \label{appendix:eq:prov--2}
  \end{align}
  $\valinterp{\semR_a}{v_o}[n-k]$ denotes that
  \begin{align}
    & \propertyA[y_a]{\obj_a}{v_o}  \label{appendix:eq:-v_0-=}\\
    &\valinterp{\tyC_a}{v_o}[n-k] \label{appendix:eq:semr-a}\\
    &\propertyC{v_o}{a}{\env}{n-k} \label{appendix:eq:metaif-u_0-neq}
  \end{align}

  By \Cref{appendix:lem:srlmono} on $\allsat$ and $n - k \le n$,
  \begin{align}
    \ctxinterp{\ctx}{\env}[n-k] \label{appendix:eq:ctxinterpctxenvn-k-}
  \end{align}
  \bymih{eq:subtyp-exta-typepr-1}{lem:semsub}{\cref{appendix:eq:ctxinterpctxenvn-k-,appendix:eq:semr-a}}
  \begin{align}
    \valinterp{\tyC_b}{v_0}[n-k]\label{appendix:eq:semr-}
  \end{align}

  if $\obj_a \neq \topobj$, from \cref{appendix:eq:prov--2}, then either
  \begin{align}
    \obj_b = \topobj \label{appendix:eq:obj_b-=-nullo}
  \end{align}
  or
  \begin{align}
    \obj_b = y_b \label{appendix:eq:o_a-neq-nullo} \\
    \proves{\ctx}{\eqprop{y_b}{y_a}} \label{appendix:eq:prov--5}
  \end{align}
  \bymih{eq:prov--5}{lem:real1}{\cref{appendix:eq:ctxinterpctxenvn-k-}}
  \begin{align*}
    \propinterp{(\eqprop{y_b}{y_a})}[n-k]
  \end{align*}
  By definition,
  \begin{align*}
    \env{(\obj_a)} = \env({\obj_b})
  \end{align*}
  Since $\obj_a \neq \topobj$, \cref{appendix:eq:-v_0-=} gives
  \begin{align}
    v_o = \env({y_a}) = \env({y_b})
  \end{align}

  Together with \cref{appendix:eq:obj_b-=-nullo}, this gives
  \begin{align}
    \propertyA[y_b]{\obj_b}{v_o} \label{appendix:eq:extenv-=-extenv}
  \end{align}

  Then proceed by cases on \cref{appendix:eq:metaif-u_0-neq}.
  \begin{caselist}
  \item \( v_o \notveq \false \) and $\propinterp{\possemprop[a]}$
    Together with \cref{appendix:eq:ctxinterpctxenvn-k-}, by definition, this gives
    \begin{align}
      \sat{\env}{\ctx \extb \posprop[a]}[n-k] \label{appendix:eq:satp-}
    \end{align}
    \bymih{eq:prov-exta-typepr}{lem:real1}{\cref{appendix:eq:satp-}}
    \begin{align*}
      \propinterp{\possemprop[b]}[n-k]
    \end{align*}
  \item $v_o \veq \false$ and $\propinterp{\negprop[a]}$.
    Together with \cref{appendix:eq:ctxinterpctxenvn-k-}, by definition, this gives
    \begin{align}
      \sat{\env}{\ctx \extb \negprop[a]}[n-k] \label{appendix:eq:satn-}
    \end{align}
    \bymih{eq:prov-exta-typepr-2}{lem:real1}{\cref{appendix:eq:satn-}}
    \begin{align*}
      \propinterp{\negprop[b]}[n-k]
    \end{align*}
  \end{caselist}

  Therefore, we have proved
  \begin{align}
    \propertyC{v_o}{b}{\env}{n-k} \label{appendix:eq:satp--3}
  \end{align}

  By \cref{appendix:eq:semr-,appendix:eq:extenv-=-extenv,appendix:eq:satp--3},$\valinterp{\synR_b}{v_o}[n-k]$. Together with $\neestepsto{\unexpr}{v_o}[k]$, by \Cref{appendix:lem:hdxp}, $\expinterp{\synR_b}{\unexpr}[n]$.
  \paragraph{Cases for $\subtype{\ctx}{\tyA}{\tyB}$.}
  \newcommand{\bdf}[1]{By definition of $\valinterp{#1}{\_}$}
  \begin{caselist}
  \item \myrule{S-Refl}. $\tyA = \tyB$. $\valinterp{\tyA}{v}$ from the premise.
  \item \myrule{S-Top}. $\tyB = \top$. By definition, $\valinterp{\top}{v}$.
  \item \myrule{S-Bot}. $\tyA = \bot$. The premise $\valinterp{\bot}{v}$ leads to a contradiction. Therefore, this case vacuously holds.
  \item \myrule{S-Trans}. The premises of the last rule and the clause are
    \begin{align}
      \subtype{\ctx}{\tyA}{\tyC}\label{appendix:eq:subtypectxtyatyc-1} \\
      \subtype{\ctx}{\tyC}{\tyB}\label{appendix:eq:subtypectxtyctyb-2} \\
      \valinterp{\tyA}{v} \label{appendix:eq:subtyp-valint-}
    \end{align}
    \bymih{eq:subtypectxtyatyc-1}{lem:semsub}{$\allsat$ and \cref{appendix:eq:subtyp-valint-}}
    \begin{align}
      \valinterp{\tyC}{v} \label{appendix:eq:valinterptybu-1}
    \end{align}
    \bymih{eq:subtypectxtyctyb-2}{lem:semsub}{$\allsat$ and \cref{appendix:eq:valinterptybu-1}}
    \begin{align*}
      \valinterp{\tyB}{v}
    \end{align*}
  \item \myrule{S-UnionSub}. $A = \uniontype{\tyA_1}{\tyA_2}$. The premises of the last rule and the clause are
    \begin{align}
      \subtype{\ctx}{\tyA_1}{\tyB} \label{appendix:eq:subtypectxtya}\\
      \subtype{\ctx}{\tyA_2}{\tyB} \label{appendix:eq:subtypectxtya_2tyb} \\
      \valinterp{\uniontype{\tyA_1}{\tyA_2}}{v} \label{appendix:eq:valint-}
    \end{align}
    \cref{appendix:eq:valint-} gives
    \begin{align*}
      \valinterp{\tyA_1}{v} \lor \valinterp{\tyA_2}{v}
    \end{align*}
    Proceed with the cases.
    \begin{caselist}
    \item $\valinterp{\tyA_1}{v}$.
      \bymih{eq:subtypectxtya}{lem:semsub}{$\allsat$ and $\valinterp{\tyA_1}{v}$}
      \begin{align*}
        \valinterp{\tyB}{v}
      \end{align*}
    \item $\valinterp{\tyA_2}{v}$.
      \bymih{eq:subtypectxtya_2tyb}{lem:semsub}{$\allsat$ and $\valinterp{\tyA_2}{v}$}
      \begin{align*}
        \valinterp{\tyB}{v}
      \end{align*}
    \end{caselist}
    In both cases, we have proved $\valinterp{\tyB}{v}$.
  \item \myrule{S-UnionSuper1}. $B = \uniontype{\tyB_1}{\tyB_2}$. The premises of the last rule and the clause are
    \begin{align}
      \subtype{\ctx}{\tyA}{\tyB_1} \label{appendix:eq:subtypectxtya1} \\
      \valinterp{\tyA}{v}\label{appendix:eq:valinterptyav-}
    \end{align}
    \bymih{eq:subtypectxtya1}{lem:semsub}{$\allsat$ and \cref{appendix:eq:valinterptyav-}}
    \begin{align*}
      \valinterp{\tyB_1}{v}
    \end{align*}
    By definition, $\valinterp{(\uniontype{\tyB_1}{\tyB_2})}{v}$
  \item \myrule{S-UnionSuper2}. $B = \uniontype{\tyB_1}{\tyB_2}$. The premise is
    \begin{align}
      \subtype{\ctx}{\tyA}{\tyB_2} \label{appendix:eq:subtypectxtya2tyb2}
    \end{align}
    We follow similar reasoning here to
    conclude $\valinterp{\tyB_2}{v}$. Then
    $\valinterp{(\uniontype{\tyB_1}{\tyB_2})}{v}$ naturally follows.
  \item \myrule{S-Fun}. $A = \funty{x}{\tyA_1}{\synR_a}, B = \funty{x}{\tyB_1}{\synR_b}$, and the premises of the last rule are
    \newcommand{\sfungoal}{\forall k \le n, w.\: \: \allowbreak{} \valinterp{\tyB_1}{w}[k] \metaimpl \expinterp{\semR_b}{\ap{v}{w}}[k][\envw]}
    \newcommand{\ctxextall}{\ctx \exta \typeprop{x}{\tyB_1}
      \extaall{\typeprop{\topobj_i}{{\tyB}'_i}}}
    \newcommand{\envextall}{\extenvall{\envw}{\topobj_{i}}{v_{i}}}
    \begin{align}
      \subtype{\ctx}{\tyB_1}{\tyA_1} \label{appendix:eq:subtyp}\\
      \subtype{\ctx \exta \bindvar[{x}][{\tyB_1}]}{\synR_a}{\synR_b} \label{appendix:eq:subtyp-exta-typepr}
    \end{align}
    Given the premise $\valinterp{(\funty{x}{\tyA_1}{\synR_a})}{v}$ of the clause, by
    definition, $v = \abs{x}[\tyA]{\semexpr_b} $ or $v = \op$. In other
    words,
    \begin{align}
      \forall k \le n, w.\: \metaif\valinterp{\tyA_1}{w}[k] \metaimpl \expinterp{\semR_a}{\ap{v}{w}}[k][\envw] \label{appendix:eq:forall-u.-meta}
    \end{align}
    Then we need to show $\valinterp{(\funty{x}{\tyB_1}{\semR_b})}{v}$. In other words, we need to show $\sfungoal$.
    Let $w$ be an arbitrary value, $k$ be a natural number and $k \le n$. We assume
    \begin{align}
      \valinterp{\tyB_1}{w}[k] \label{appendix:eq:semr-1--1}
    \end{align}
    \bymih{eq:subtyp}{lem:semsub}{$\allsat$ and \cref{appendix:eq:semr-1--1}}
    \begin{align}
      \valinterp{\tyA_1}{w}[k] \label{appendix:eq:valinterptya_1wn-1-}
    \end{align}
    By applying \cref{appendix:eq:forall-u.-meta} to $k \le n$ and \cref{appendix:eq:valinterptya_1wn-1-}
    \begin{align}
      \expinterp{\semR_a}{\ap{v}{w}}[k][\envw] \label{appendix:eq:semr-1-}
    \end{align}

    By \Cref{appendix:lem:psubwk} on \cref{appendix:eq:valinterptya_1wn-1-},
    \begin{align}
      \valinterp{\tyA_1}{w}[k][\envw] \label{appendix:eq:valint--1}
    \end{align}
    By applying \Cref{appendix:lem:srlmono} on $\allsat$ and $k \le n$,
    \begin{align*}
      \ctxinterp{\ctx}{\env}[k]
    \end{align*}
    Together with \cref{appendix:eq:valint--1}, this gives
    \begin{align}
      \ctxinterp{\ctx \exta \bindvar[{x}][{\tyA_1}]}{\envw}[k]\label{appendix:eq:satenvwctx-exta-type}
    \end{align}

    \bymih{eq:subtyp-exta-typepr}{lem:semressub}{\cref{appendix:eq:semr-1-,appendix:eq:satenvwctx-exta-type}}
    $\expinterp{\semR_b}{\ap{v}{w}}[k][\envw]$.
    In other words, we have proved

    $\sfungoal$
  \end{caselist}
  This completes the cases for the judgment $\subtype{\ctx}{\tyA}{\tyB}$,
  thereby establishing Statement \ref{appendix:lem:semsub}.

  \paragraph{Cases for $\restrictrel{\ctx}{\tyA}{\tyB}{\tyC}$.}
  \begin{caselist}
  \item \myrule{Restrict-Base}. The premises are
    \begin{align*}
      \tyA, \tyB \in \{\true,\false,\Nat,\Str\} \\
      \tyB \neq \tyA
    \end{align*}
    By definition, we know that if $\valinterp{\tyA}{v}$ and
    $\tyA \in \{\true,\false,\Nat,\Str\}$, then $v = \true, \false, \nat $ and $ \str$
    respectively. Because $\tyA \neq \tyB$, $\negate \valinterp{\tyB}{v}$, which
    contradicts $\valinterp{\tyB}{v}$. Similarly, If $\valinterp{\tyB}{v}$,
    $\negate \valinterp{\tyA}{v}$, which contradicts $\valinterp{\tyA}{v}$. In
    both cases, we have derived a contraction. Henceforth, the case holds
    vacuously
  \item \myrule{Restrict-BaseFun}. $B = \funty{x}{C}{\R}$. The premise is
    $\tyA \in \{\true,\false,\Nat,\Str\}$.  By definition, if
    $\valinterp{\tyA}{v}$, then $\negate \valinterp{\funty{x}{C}{\R}}{v}$, which
    contradicts the premise $\valinterp{\funty{x}{C}{\R}}{v}$. Henceforth, the
    case holds vacuously.
  \item \myrule{Restrict-Union}. The premises of the last rule are
    \begin{align}
      \restrictrel{\ctx}{\tyA}{\tyB}{\tyC} \label{appendix:eq:restrict-union-left}\\
      \restrictrel{\ctx}{\tyD}{\tyB}{\tyE} \label{appendix:eq:restrict-union-right}
    \end{align}
    We also know $\valinterp{\uniontype{\tyA}{\tyD}}{v}$ and $\valinterp{\tyB}{v}$.
    We need to show
    \begin{align*}
      \valinterp{(\uniontype{\tyC}{\tyE})}{v}
    \end{align*}

    By definition, $\valinterp{\uniontype{\tyA}{\tyD}}{v}$ gives
    \begin{align*}
      \valinterp{\tyA}{v} \metaor \valinterp{\tyD}{v}
    \end{align*}

    Proceed by cases.
    \begin{caselist}
    \item \label{appendix:cccc1}
      $\valinterp{\tyA}{v}$.
      \bymih{eq:restrict-union-left}{lem:restrict}{$\allsat$, $\valinterp{\tyA}{v}$, and $\valinterp{\tyB}{v}$}
      \begin{align}
        \valinterp{\tyC}{v}
      \end{align}

      Then by definition,
      \begin{align}
        \valinterp{(\uniontype{\tyC}{\tyE})}{v}
      \end{align}
    \item $\valinterp{\tyD}{v}$. $\valinterp{(\uniontype{\tyC}{\tyE})}{v}$
      follows by reasoning analogous to that in Case \ref{appendix:cccc1}.
    \end{caselist}
  \item \myrule{Restrict-Sym}
    The premise of the last rule is $\restrictrel{\ctx}{\tyB}{\tyA}{\tyC}$.
    Then by the mutual IH.
  \item \myrule{Restrict-Sub}.
    The case holds by the given premise $\valinterp{\tyA}{v}$.
  \item \myrule{Restrict-Other}.
    The case holds by the given premise $\valinterp{\tyB}{v}$.
  \end{caselist}
  \paragraph{Cases for $\removerel{\ctx}{\tyA}{\tyB}{\tyC}$.}

  \begin{caselist}
  \item \myrule{Remove-Union}.
    $\tyA = \uniontype{\tyC}{\tyD}$ and, the premises of the clause and the last rule are
    \begin{align}
      \valinterp{\uniontype{\tyC}{\tyD}}{v} \label{appendix:eq:valint--notv} \\
      \notvalinterp{\tyB}{v} \label{appendix:eq:notvalinterptybv--} \\
      \removerel{\ctx}{\tyC}{\tyB}{\tyE_1} \label{appendix:eq:remove-union-left}\\
      \removerel{\ctx}{\tyD}{\tyB}{\tyE_2} \label{appendix:eq:remove-union-right}
    \end{align}
    Then we need to show
    \begin{align*}
      \valinterp{(\uniontype{\tyE_1}{\tyE_2})}{v}
    \end{align*}

    By definition, \cref{appendix:eq:valint--notv} gives
    \begin{align*}
      \valinterp{\tyC}{v} \metaor \valinterp{\tyD}{v}
    \end{align*}

    Proceed by cases.
    \begin{caselist}
    \item  \label{appendix:sec:caseaaa}
      $\valinterp{\tyC}{v}$.
      \bymih{eq:remove-union-left}{lem:remove}{$\allsat$, $\valinterp{\tyC}{v}$, \cref{appendix:eq:notvalinterptybv--}}
      \begin{align}
        \valinterp{\tyE_1}{v}
      \end{align}

      Then by definition,
      \begin{align}
        \valinterp{(\uniontype{\tyE_1}{\tyE_2})}{v}
      \end{align}
    \item
      $\valinterp{\tyD}{v}$. The case holds by reasoning analogous to Case \ref{appendix:sec:caseaaa}.
    \end{caselist}
  \item \myrule{Remove-Sub}.
    The premises of the last rule and the clause  are
    \begin{align}
      \valinterp{\tyA}{v} \label{appendix:eq:valint--notv1} \\
      \notvalinterp{\tyB}{v} \label{appendix:eq:notv-subtyp-} \\
      \subtype{\ctx}{\tyA}{\tyB} \label{appendix:eq:remove-sub-subtype}
    \end{align}

    \bymih{eq:remove-sub-subtype}{lem:semsub}{$\allsat$ and \cref{appendix:eq:valint--notv1}}
    \begin{align*}
      \valinterp{\tyB}{v}
    \end{align*}
    which contradicts the premise $\notvalinterp{\tyB}{v}$. Therefore the case vacuously holds.
  \item \myrule{Remove-Other}. This case holds by the premise $\valinterp{\tyA}{v}$.
  \end{caselist}
  \paragraph{Cases for $\proves{\ctx}{\prop}$.}
  \def\ctxextzero{\ctx' \cdot \ctxele[x][\tyA][[]]}
  \begin{caselist}
  \item L-True. By definition, $\propinterp{\TT}$.

  \item L-False. The premise of the last rule is
    \begin{align}
      \proves{\ctx}{\FF}
      \label{appendix:eq:provesctxff-}
    \end{align}
    \bymih{eq:provesctxff-}{lem:real1}{$\allsat$} $\propinterp{\FF}$, which does
    not hold. Therefore the case vacuously holds.

  \item L-Atom. The premise of the last rule is $\lookup{\ctx}{\prop}$. Given $\sat{\env}{\ctx}$, by definition,
    $\propinterp{\prop}$
  \item L-Ann. The premise of the last rule is $\lookup{\ctx}{(\bindvar)}$. Given $\sat{\env}{\ctx}$, by definition,
    $\valinterp{\tyA}{\env(x)}$. Then by definition, $\propinterp{(\typeprop{x}{\tyA})}$
  \item L-AndI. The premises of the last rule are
    \begin{align}
      \proves{\ctx}{\prop_1} \label{appendix:eq:prov-prov-} \\
      \proves{\ctx}{\prop_2}\label{appendix:eq:provesctxprop_2-}
    \end{align}
    \bymih{eq:prov-prov-,eq:provesctxprop_2-}{lem:real1}{$\allsat$} $\propinterp{\prop_1}$ and
    $\propinterp{\prop_2}$ respectively. Then by definition,
    $\propinterp{\conjprop{\prop_1}{\prop_2}}$
  \item \label{appendix:sec:1} L-OrI1. The premise of the last rule is
    \begin{align}
      \label{appendix:eq:provesctxprop_1-}
      \proves{\ctx}{\prop_1}
    \end{align}
    \bymih{eq:provesctxprop_1-}{lem:real1}{$\allsat$}
    $\propinterp{\prop_1}$. Then by definition,
    $\propinterp{\disjprop{\prop_1}{\prop_2}}$.
  \item L-OrI2. The case holds by reasoning analogous to that in Case \ref{appendix:sec:1}.
  \item \label{appendix:sec:2} L-AndE1. The premise of the last rule is
    \begin{align}
      \label{appendix:eq:prov--6}
      \proves{\ctx}{\conjprop{\prop_1}{\prop_2}}
    \end{align}
    \bymih{eq:prov--6}{lem:real1}{$\allsat$} $\propinterp{\conjprop{\prop_1}{\prop_2}}$.
    Then by definition,
    $\propinterp{\prop_1}$.
  \item L-AndE2. The case holds by reasoning analogous to that in Case \ref{appendix:sec:2}.
  \item \myrule{L-OrE}. The premises of the last rule are
    \begin{align}
      \wfprop{\prop} \label{appendix:eq:wfpr-prov-} \\
      \proves{\ctx}{\disjprop{\prop_1}{\prop_2}} \label{appendix:eq:prov-2}\\
      \proves{\ctx \extb \prop_1}{\prop} \label{appendix:eq:prov-extb-prop_1pr}\\
      \proves{\ctx \extb \prop_2}{\prop} \label{appendix:eq:prov-extb-prop_2pr}
    \end{align}
    \bymih{eq:prov-2}{lem:real1}{$\allsat$}
    $\propinterp{(\disjprop{\prop_1}{\prop_2})}$. By definition,
    $\propinterp{\disjprop{\prop_1}{\prop_2}}$, i.e.
    $\propinterp{\prop_1} \metaor \propinterp{\prop_2}$. Let us proceed with the two cases.

    In the first case, $\propinterp{\prop_1}$ and $\ctxinterp{\ctx}{\env}$. Then by definition
    \begin{align}
      \sat{\env}{\ctx \extb \prop_1}
      \label{appendix:eq:satenvctx-extb-synpr}
    \end{align}
    \bymih{eq:prov-extb-prop_1pr}{lem:real1}{\cref{appendix:eq:satenvctx-extb-synpr}}
    $\propinterp{\prop}$.

    In the other case, $\propinterp{\prop_2}$ and
    $\ctxinterp{\ctx}{\env}$. Then by definition,
    \begin{align}
      \sat{\env}{\ctx \extb \prop_2}
      \label{appendix:eq:satenvctx-extb-synpr-2}
    \end{align}
    \bymih{eq:prov-extb-prop_2pr}{lem:real1}{\cref{appendix:eq:satenvctx-extb-synpr-2}}
    $\propinterp{\prop}$.
  \item L-Sub: The premises of the last rule are
    \begin{align}
      \proves{\ctx}{\typeprop{x}{\tyA}} \label{appendix:eq:prov-3} \\
      \subtype{\ctx}{\tyA}{\tyB} \label{appendix:eq:lem1lsub0}
    \end{align}
    \bymih{eq:prov-3}{lem:real1}{$\allsat$}
    $\propinterp{\typeprop{x}{\tyA}}$. By definition,
    \begin{align}
      \valinterp{\tyA}{\env(x)} \label{appendix:eq:lem1lsub2}
    \end{align}
    \bymih{eq:lem1lsub0}{lem:semsub}{$\allsat$ and \cref{appendix:eq:lem1lsub2}}
    \begin{align}
      \valinterp{\tyB}{\env(x)} \label{appendix:eq:lem1lsub1}
    \end{align}
    By definition, $\propinterp{\typeprop{x}{\tyB}}$.
  \item L-SubNot. The premises of the last rule are
    \begin{align}
      \proves{\ctx}{\nottypeprop{x}{\tyA}} \label{appendix:eq:prov-} \\
      \subtype{\ctx}{\tyB}{\tyA} \label{appendix:eq:lem1lnotsub0}
    \end{align}

    Assume, for contradiction, that
    $\metanot \propinterp{\nottypeprop{x}{\tyB}}$, i.e. there exists
    $\betw{l}[1][n]$ such that
    \begin{align}
      \valinterp{\tyB}{\env(x)}[l]\label{appendix:eq:valinterptybenvxl-}
    \end{align}
    \bymih{eq:prov-}{lem:real1}{$\allsat$}
    $\propinterp{\nottypeprop{x}{\tyA}}$. By definition,
    \begin{align}
      \notvalinterp{\tyA}{\env(x)} \label{appendix:eq:metan-semr-}
    \end{align}
    Because $l \le n$, by \Cref{appendix:lem:srlmono} on $\allsat$,
    \begin{align}
      \ctxinterp{\ctx}{\env}[l] \label{appendix:eq:ctxinterpctxenvl-}
    \end{align}
    \bymih{eq:lem1lnotsub0}{lem:semsub}{\cref{appendix:eq:ctxinterpctxenvl-,appendix:eq:valinterptybenvxl-}}
    \begin{align}
      \valinterp{\tyA}{\env(x)}[l] \label{appendix:eq:lem1lnotsub1}
    \end{align}
    Because $l \le n$,  this contradicts \cref{appendix:eq:metan-semr-}.
  \item L-Bot. The premises of the last rule is
    \begin{align}
      \proves{\ctx}{\typeprop{x}{\bot}}\label{appendix:eq:prov--7}
    \end{align}
    \bymih{eq:prov--7}{lem:real1}{$\allsat$} $\propinterp{\typeprop{x}{\bot}}$. By definition,
    $\valinterp{\bot}{\env(x)}$ leads to a contradiction.
    Therefore, this case holds vacuously.
  \item \myrule{L-Refl}. The premise of the last rule is $x \in \dom{\ctx}$. Since $\allsat$,
    $x \in \dom{\env}$. Then By definition,
    $\propinterp{(\eqprop{x}{x})}$.
  \item \myrule{L-Transp}. The premises of the last rule are
    \begin{align}
      \wfprop[{\ctx \exta \bindvar[{z}][{\top}]}]{\prop} \\
      \proves{\ctx}{\substprop{\prop}{z}{x}} \label{appendix:eq:provesctxprop-} \\
      \proves{\ctx}{\eqprop{x}{y}} \label{appendix:eq:prov--1}
    \end{align}
    \bymih{eq:prov--1}{lem:real1}{$\allsat$}
    \begin{align*}
      \propinterp{\eqprop{x}{y}}
    \end{align*}
    By definition,
    \begin{align}
      \env(x) = \env(y) \label{appendix:eq:satpenveqpropo_1o_2-}
    \end{align}
    \bymih{eq:provesctxprop-}{lem:real1}{$\allsat$}
    \begin{align}
      \propinterp{\substprop{\prop}{z}{x}}\label{appendix:eq:satpenvprop-}
    \end{align}
    By applying \Cref{appendix:lem:psubrenaming} to \cref{appendix:eq:satpenveqpropo_1o_2-,appendix:eq:satpenvprop-},
    $\propinterp{\substprop{\prop}{z}{y}}$.
  \item \myrule{L-Restrict}.
    The premises of the last rule are
    \begin{align}
      \proves{\ctx}{\typeprop{x}{\tyA}} \label{appendix:eq:prov--3}\\
      \proves{\ctx}{\typeprop{x}{\tyB}} \label{appendix:eq:prov--4}\\
      \restrictrel{\ctx}{\tyA}{\tyB}{\tyC} \label{appendix:eq:restrictrel-}
    \end{align}

    \bymih{eq:prov--3,eq:prov--4}{lem:real1}{$\allsat$}
    \begin{align*}
      \propinterp{\typeprop{x}{\tyA}} \\
      \propinterp{\typeprop{x}{\tyB}}
    \end{align*}

    By definition,
    \begin{align*}
      \valinterp{\tyA}{\env(x)} \\
      \valinterp{\tyB}{\env(x)}
    \end{align*}

    Let $v = \env(x)$.
    \begin{align}
      \valinterp{\tyA}{v} \label{appendix:eq:valinterppsubtyav-}\\
      \valinterp{\tyB}{v} \label{appendix:eq:valinterppsubtybv-}
    \end{align}
    Then our goal is to show $\propinterp{\typeprop{x}{\tyC}}$.

    \bymih{eq:restrictrel-}{lem:restrict}{$\allsat$ and \cref{appendix:eq:valinterppsubtyav-,appendix:eq:valinterppsubtybv-}}
    \begin{align*}
      \valinterp{\tyC}{v}
    \end{align*}

    By definition,
    \begin{align*}
      \propinterp{\typeprop{x}{\tyC}}
    \end{align*}
  \item \myrule{L-Remove}.
    The premises of the last rule are
    \begin{align}
      \proves{\ctx}{\typeprop{x}{\tyA}} \label{appendix:eq:prov--31}\\
      \proves{\ctx}{\nottypeprop{x}{\tyB}} \label{appendix:eq:prov--41}\\
      \removerel{\ctx}{\tyA}{\tyB}{\tyC} \label{appendix:eq:removerel-}
    \end{align}

    \bymih{eq:prov--31,eq:prov--41}{lem:real1}{$\allsat$}
    \begin{align*}
      \propinterp{\typeprop{x}{\tyA}} \\
      \propinterp{\nottypeprop{x}{\tyB}}
    \end{align*}

    By definition,
    \begin{align*}
      \valinterp{\tyA}{\env(x)} \\
      \notvalinterp{\tyB}{\env(x)}
    \end{align*}

    Let $v = \env(x)$.

    \begin{align}
      \valinterp{\tyA}{v} \label{appendix:eq:valinterppsubtyav-1}\\
      \notvalinterp{\tyB}{v} \label{appendix:eq:valinterppsubtybv-1}
    \end{align}

    Then our goal is to show $\propinterp{\typeprop{v}{\tyC}}$

    \bymih{eq:removerel-}{lem:remove}{$\allsat$ and \cref{appendix:eq:valinterppsubtyav-1,appendix:eq:valinterppsubtybv-1}}
    \begin{align*}
      \valinterp{\tyC}{v}
    \end{align*}

    By definition,
    \begin{align*}
      \propinterp{\typeprop{x}{\tyC}}
    \end{align*}
  \end{caselist}
\end{proof}

\begin{lemma}[Erasure Preserves Logical Relations] \hfill \label{appendix:lem:erase-prev-rel} Given $\allsat$, the following hold simultaneously:
  \begin{enumerate}
    \item \label{appendix:lem:erprev-type} if  $\wftype{\tyA}$ and $\valinterp{\tyA}{v}$, then $\valinterp{\osubst{\tyA}[x][\topobj]}{v}$.
    \item \label{appendix:lem:erprev-prop} if $\wfprop{\prop}$ and $\propinterp{\prop}$, then $\propinterp{\osubst{\prop}[x][\topobj]}$.
  \end{enumerate}
\end{lemma}
\begin{proof} We prove each clause separately.
  \paragraph{Clause 1.} Assume $\allsat$ and $\wftype{\tyA}$. By \Cref{appendix:lem:subtyping-erase}, $\subtype{\ctx}{\tyA}{\poserase{\tyA}}$. Then by \Cref{appendix:lem:prfsound}, $\valinterp{\poserase{\tyA}}{v}$.

  \paragraph{Clause 2.} Assume $\allsat$, $\propinterp{\prop}$ and $\proves{\ctx}{\prop}$. Because of the presuppositions, $\wfprop{\prop}$. By \Cref{appendix:lem:subtyping-erase}, $\proves{\ctx \exta \prop}{\poserase{\prop}}$. Since $\propinterp{\prop}$, by definition $\ctxinterp{\ctx \exta \prop}{\env}$. Then by \Cref{appendix:lem:prfsound}, $\propinterp{\poserase\prop}$.
\end{proof}

\begin{lemma}[Object Substitution Preserves Logical Relations]\label{appendix:lem:obj-subst-logrel}
  Given $\allsat$, if $\wftype{\tyC}, \wfprop{\posprop}, \wfprop{\negprop}, \wfobj{\obj_r}$,
  $\valinterpres{\res[\tyC][\posprop][\negprop][\obj_r]}{v_o}$, and either $\obj = \topobj$ or $\env(y) \alphaeq \env(x)$ where $\obj = y$ for some $y$, then
  $\valinterpres{\osubst{\res[\tyC][\posprop][\negprop][\obj_r]}[x][\obj]}{v_o}$.
\end{lemma}
\begin{proof}
  Assume $\allsat$, $\propertyA[y]{\obj}{\env(x)}$, and
  $\valinterpres{\res[C][\posprop][\negprop][\obj_r]}{v_o}$.
  By definition, we have
  \begin{align}
    \propertyA[y']{\obj_r}{v_o} \label{appendix:eq:obj-subst-logrel-propa}\\
    \valinterp{C}{v_o} \label{appendix:eq:obj-subst-logrel-val}\\
    \propertyC{v_o}{}{\env}{n} \label{appendix:eq:obj-subst-logrel-propc}
  \end{align}
  To show the substituted result interpretation, it suffices to prove the
  following three obligations:
  \begin{enumerate}
  \item $\propertyA{\osubst{\obj_r}[x][\obj]}{v_o}$.
  \item if $\valinterp{C}{v_o}$, then
    $\valinterp{\osubst{C}[x][\obj]}{v_o}$.
  \item if $\propinterp{\prop}$, then
    $\propinterp{\osubst{\prop}[x][\obj]}$.
  \end{enumerate}

  Let us prove the first obligation by cases using
  $\propertyA[y]{\obj}{\env(x)}$ and \cref{appendix:eq:obj-subst-logrel-propa}.
  \begin{itemize}
  \item $\obj_r = \topobj$, or $\obj = \topobj$ and $\obj_r = x$. Then
    $\osubst{\obj_r}[x][\obj] = \topobj$.
  \item $\obj_r = y'$, $y' \neq x$, and
    $v_o \alphaeq \env(y')$. Then
    $\osubst{y'}[x][\obj] = y'$, and the obligation follows from
    \cref{appendix:eq:obj-subst-logrel-propa}.
  \item $\obj = y$, $\obj_r = x$, $\env(x) \alphaeq \env(y)$, and
    $v_o \alphaeq \env(x)$. Then $\osubst{x}[x][y] = y$ and
    $\env(x) = \env(y) = v_o$. In other words, we have proved
    \begin{align}
      \propertyAsnd[y]{\osubst{\obj_r}[x][\obj]}{v_o}
    \end{align}
  \end{itemize}

  For the second obligation, assume $\valinterp{C}{v_o}$. If
  $\obj = \topobj$, then by \Cref{appendix:lem:erase-prev-rel},
  \begin{align*}
    \valinterp{\poserase{C}}{v_o}
  \end{align*}
  If $\obj \neq \topobj$, then by $\propertyA[y]{\obj}{\env(x)}$, we know that
  $\obj = y$ and $\env(x) = \env(y)$ for some $y \in \dom{\env}$. Then by
  \Cref{appendix:lem:psubrenaming},
  \begin{align*}
    \valinterp{\osubst{C}[x][y]}{v_o}
  \end{align*}
  Therefore, by definition,
  \begin{align*}
    \valinterp{\osubst{C}[x][\obj]}{v_o}
  \end{align*}

  For the third obligation, assume $\propinterp{\prop}$. The proof is
  analogous to the second obligation. By replacing $\valinterp{C}{v_o}$ with
  $\propinterp{\prop}$, and using the
  proposition clauses of \cref{appendix:lem:erase-prev-rel,appendix:lem:psubrenaming} accordingly,
  \begin{align*}
    \propinterp{\osubst{\prop}[x][\obj]}
  \end{align*}
  Thus all three obligations hold. Therefore, by definition,
  \begin{align*}
    \valinterpres{\osubst{\res[C][\posprop][\negprop][\obj_r]}[x][\obj]}{v_o}.
  \end{align*}
\end{proof}

\begin{theorem}[The Fundamental Theorem of Logical Relations] \label[theorem]{appendix:thm:fundamental}\hfill
  if $\allsat$ and $\typeof{\ctx}{\expr}{\synR}$, then $\expinterp{\R}{\psubexpr{\eraseann{\expr}}[\env]}$
\end{theorem}
\begin{proof}
  Proceed by rule induction on $\typeof{\ctx}{e}{\synR}$.
  \begin{caselist}
  \item
    \myrule{T-BaseVal}.
    $e = \eraseann{e} = \bv$.

    Case on $\bv$.
    \begin{caselist}
    \item
      $\bv = \nat$. $\R = \res[\Nat][\TT][\FF][\topobj]$.
      By definition,
      \begin{align*}
        \valinterp{\Nat}{\nat}
      \end{align*}
      Since $\nat \neq \false$, we need to show
      \begin{align*}
        \propinterp{\TT}
      \end{align*}
      which holds trivially by definition.
      Then by definition,
      \begin{align*}
        \valinterp{\res[\Nat][\TT][\FF][\topobj]}{\nat}
      \end{align*}
      By definition,
      \begin{align*}
        \expinterp{\res[\Nat][\TT][\FF][\topobj]}{\nat}
      \end{align*}
    \item $\bv = \str$. This case follows reasoning analogous to that in the previous case.
    \item $\bv = \true$. This case follows reasoning analogous to that in the first case.
    \item $\bv = \false$. This case follows reasoning analogous to that in the first case.
    \end{caselist}
  \item
    \newcommand{\curenv}{\envele{x}{w}}
    \myrule{T-PrimOp}. $e = \eraseann{e} = \op$.

    Case on $\op$.
    \begin{caselist}
    \item
      $\op = \isnatop$ and $\R = \res[\tyA][\TT][\FF][\topobj]$ where $\tyA = \predfun{\Nat}$

      Now we need to show $\valinterp{\predfun{\Nat}}{\op}$. In other words,
      assume $\valinterp{\top}{w}[k]$ for any $k \le n$ and $w$, we need to show
      $\expinterp{\predres{\Nat}}{\ap{\op}{w}}[k][\envele{x}{w}]$. if $n \le 1$,
      the case trivially holds. Otherwise, $0 \le k \le n$.
      if $k = 0$, the result follows trivially.
      Otherwise, let us proceed by casing on
      $\eestepsto{\ap{\op}{w}}{\unexpr'}$. The only candidate is $\stepsto{\ap{\op}{w}}{\unexpr'}$ for
      which the only case is \myrule{E-PrimOp}. Therefore,
      $\unexpr' = u = \execop{\isnatop}{w}$.

      By definition, $u = \true$ or $u = \false$. Then it follows that
      \begin{align}
        \valinterp{\Bool}{u}[k-1] \label{appendix:eq:valinterpboolun-2-}
      \end{align}

      if $u =_{\alpha} \true$, then by the definition of $\execop{\isnatop}{w}$, $w = \nat$ for some $\nat$.
      Then by definition, $\valinterp{\Nat}{\curenv(x)}[k-1][\curenv]$. By definition,
      \begin{align}
        \propinterp{\typeprop{x}{\Nat}}[k-1][\curenv]\label{appendix:eq:prop-1cur-}
      \end{align}
      Similarly, if $u =_{\alpha} \false$, then by the definition of $\execop{\isnatop}{w}$, $w \neq \nat$ for any $\nat$.
      Then by definition, $\notvalinterp{\Nat}{\curenv(x)}[m][k-1][\curenv]$. By definition,
      \begin{align*}
        \propinterp{\nottypeprop{x}{\Nat}}[k-1][\curenv]
      \end{align*}
      Together with \cref{appendix:eq:prop-1cur-}, this gives
      \begin{align*}
        \propertyC{u}{\typeprop{x}{\Nat}}[\nottypeprop{x}{\Nat}]{\curenv}{k-1}
      \end{align*}
      Together with \cref{appendix:eq:valinterpboolun-2-}, this gives
      \begin{align*}
        \valinterp{\res[\Bool][\typeprop{x}{\Nat}][\nottypeprop{x}{\Nat}][\topobj]}{u}[k-1][\curenv]
      \end{align*}
      Since $\eestepsto{\ap{\op}{w}}{u}$, by \Cref{appendix:lem:hdxp}
      \begin{align*}
        \expinterp{\predres{\Nat}}{\ap{\op}{w}}[k][\curenv]
      \end{align*}
    \item $\op = \isstrop$. This case follows reasoning analogous to that in the previous case.
    \item $\op = \notop$. This case follows reasoning analogous to that in the first case.
    \item $\op = \addop$. The case follows similar yet simpler reasoning, as the propositions are trivial.
    \end{caselist}
  \item \myrule{T-Var}.
    $e = \eraseann{e} = x $ and $R = \res[\tyA][\nottypeprop{x}{\false}][\typeprop{x}{\false}][x]$. The
    premise of the last rule is $\proves{\ctx}{\typeprop{x}{\tyA}}$. We need to
    show
    $\expinterp{\res[\tyA][\nottypeprop{x}{\false}][\typeprop{x}{\false}][x]}{\psubexpr{x}[{\env}]}$.
    Since $\allsat$, then $\env(x) = v$ for some $v$.
    By applying \lemclause{lem:real1} to $\allsat$ and $\proves{\ctx}{\typeprop{x}{\tyA}}$,
    \begin{align*}
      \propinterp{\typeprop{x}{\tyA}}
    \end{align*}
    By definition
    \begin{align}
      \valinterp{\tyA}{v} \label{appendix:eq:valinterppsubtyav--1}
    \end{align}

    If $v \neq \false$,
    \begin{align*}
      \notvalinterp{\false}{v}
    \end{align*}
    By definition
    \begin{align}
      \propinterp{\nottypeprop{x}{\false}}\label{appendix:eq:sems--2}
    \end{align}

    Otherwise, $v = \false$, and
    \begin{align*}
      \valinterp{\false}{x}
    \end{align*}
    By definition
    \begin{align}
      \propinterp{\typeprop{x}{\false}}\label{appendix:eq:sems--3}
    \end{align}

    By $\env(x) = v$, \cref{appendix:eq:valinterppsubtyav--1,appendix:eq:sems--2,appendix:eq:sems--3},
    \begin{align*}
      \valinterp{\res[{\tyA}][\nottypeprop{x}{\false}][\typeprop{x}{\false}][x]}{v}
    \end{align*}
    By definition, $\expinterp{\res[{\tyA}][\nottypeprop{x}{\false}][\typeprop{x}{\false}][v]}{\psubexpr{x}}$.
  \item \myrule{T-Abs}.
    $\R = {\res[\funty{x}{\tyA}{\synR}][\TT][\FF][\topobj]}, e =
    {\abs{x}{\expr_b}}$ and
    $\unexpr = \eraseann{e} = \unabs[x][{\eraseann{\expr_b}}]$. The premise of
    the last rule is
    \begin{align}
      \typeof{\ctxa}{\expr_b}{\R} \label{appendix:eq:typeofctxaexpr_br-}
    \end{align}

    We need to show
    $\forall k \le n, w.\valinterp{\tyA}{w}[k] \metaimpl
    \expinterp{\R}{\ap{\psubexpr{\unexpr}[\env]}{w}}[k][\envw]$.

    Let $w$ be an arbitrary value, $k$ be a natural number, $k < n$ and assume $\valinterp{\tyA}{w}[k]$.

    By \Cref{appendix:lem:psubwk},
    \begin{align*}
      \valinterp{\tyA}{w}[k][\envw]
    \end{align*}
    By definition,
    \begin{align*}
      \propinterp{\typeprop{x}{\tyA}}[k][\envw]
    \end{align*}
    By \Cref{appendix:lem:srlmono},
    \begin{align}
      \propinterp{\typeprop{x}{\tyA}}[k-1][\envw] \label{appendix:eq:prop-}
    \end{align}
    Similarly, by \Cref{appendix:lem:srlmono} on $\allsat$,
    \begin{align*}
      \ctxinterp{\ctx}{\env}[k-1]
    \end{align*}
    Together with \cref{appendix:eq:prop-}, this gives
    \begin{align}
      \ctxinterp{\ctxa}{\envw}[k-1] \label{appendix:eq:satenvwctxa-}
    \end{align}
    By IH on \cref{appendix:eq:satenvwctxa-,appendix:eq:typeofctxaexpr_br-},
    \begin{align}
      \expinterp{\R}{(\psubexpr{\eraseann{\expr_b}}[\envw])}[k-1][\envw] \label{appendix:eq:expint-1envw-}
    \end{align}

    By \myrule{E-Beta},
    $\stepsto{\ap{\psubexpr{(\unabs[x][{\eraseann{\expr_b}}])}[\env]}{w}}{\psubexpr{\eraseann{\expr_b}}[\envw]}$. Together
    with \cref{appendix:eq:expint-1envw-}, this allows us to apply \Cref{appendix:lem:hdxp},
    yielding
    \begin{align}
      \expinterp{\R}{\ap{\psubexpr{(\unabs[x][{\eraseann{\expr_b}}])}[\env]}{w}}[k][\envw] \label{appendix:eq:expint-1envw-1}
    \end{align}

    Then by universal introduction,
    $\forall k \le n, w.\: \valinterp{\tyA}{w}[k] \metaimpl \expinterp{\R}{(\ap{\psubexpr{\unabs[x][({\eraseann{\expr_b}})]}}{w})}[k][\envw]$
  \item
    \myrule{T-If}. $\expr = {\If{\expr_1}{\expr_2}{\expr_3}}, \eraseann{\expr} = {\If{\eraseann{\expr_1}}{\eraseann{\expr_2}}{\eraseann{\expr_3}}}$ and $\R = \IfR$. Let $\unexpr_1= {\eraseann{\expr_1}}, \unexpr_2 = {\eraseann{\expr_2}}, \unexpr_3 = {\eraseann{\expr_3}}$, and
    the premises of the last rule are
    \begin{align}
      \typeof{\ctx}{e_1}{\IfCondR} \label{appendix:eq:type} \\
      \typeof{\ifbrctx \ifbrposprop}{e_2}{\IfR} \label{appendix:eq:type-exta-extb}\\
      \typeof{\ifbrctx \ifbrnegprop}{e_3}{\IfR} \label{appendix:eq:type-exta-extb-1}
    \end{align}

    By IH on $\allsat$ and \cref{appendix:eq:type},
    $\expinterp{\IfCondR}{\psubexpr{\eraseann{\expr_1}}}$.
    By \Cref{appendix:lem:hdrd}, there are two cases.

    \begin{itemize}
    \item $\hdrdpropa[\psubexpr{\unexpr_1}]$.  \propaconseqsub[\psubexpr{\unexpr_1}][\unexpr'_1]{\psubexpr{\If{{\unexpr_1}}{\unexpr_2}{\unexpr_3}}}{\If{\unexpr'_1}{\psubexpr{\unexpr_2}}{\psubexpr{\unexpr_3}}}[\IfR]

    \item $\hdrdpropb[\IfCondR][v_1][k][\psubexpr{\unexpr_1}]$.
    \NewDocumentCommand{\condenv}{}{\env}
    By definition,
    \begin{align}
      \valinterp{\tyA}{v_1}[n-k] \label{appendix:eq:semr-meta-v_1} \\
      \propertyC{v_1}{c}{\env}{n-k} \label{appendix:eq:metaif-v_1-neq}
    \end{align}

    Now let us proceed with cases on \cref{appendix:eq:metaif-v_1-neq}.

    \begin{caselist}
    \item
      $v_1 \neq \false$ and $\propinterp{\ifbrposprop}[n-k]$.
      Together with $\allsat$, this gives
      \begin{align}
        \sat{\condenv}{\ifbrctx{\ifbrposprop}} \label{appendix:eq:satc--1}
      \end{align}
      By IH on \cref{appendix:eq:satc--1,appendix:eq:type-exta-extb},
      \begin{align*}
        \expinterp{\IfR}{\psubexpr{\unexpr_2}[\condenv]}
      \end{align*}
      Then by \Cref{appendix:lem:srlmono},
      \begin{align}
        \expinterp{\IfR}{\psubexpr{\unexpr_2}[\condenv]}[n-k] \label{appendix:eq:expint-}
      \end{align}
      Since $0 \le k < n$, $0 \le n - k - 1 < n$. Then
      By \Cref{appendix:lem:srlmono} on \cref{appendix:eq:expint-},
      \begin{align}
        \expinterp{\IfR}{\psubexpr{\unexpr_2}[\condenv]}[n-k-1] \label{appendix:eq:expint-if-1}
      \end{align}
      By \Cref{appendix:lem:hdrd} on \cref{appendix:eq:expint-if-1}, there are two cases.
      \begin{itemize}
      \item $\hdrdpropa[\psubexpr{\unexpr_2}][n-k-1]$.
        Since
        $\neestepsto{\psubexpr{\If{{\unexpr_1}}{\unexpr_2}{\unexpr_3}}}{\psubexpr{\unexpr_2}}[k+1]$,
        \cref{appendix:eq:expint-if-1} and \Cref{appendix:lem:hdxp} give
        \begin{align*}
          \expinterp{\IfR}{\psubexpr{\If{{\unexpr_1}}{\unexpr_2}{\unexpr_3}}}
        \end{align*}
      \item
        In the other case,
        $\neestepsto{\psubexpr{\unexpr_2}}{v_2}[j]$ and
        \begin{align*}
          \valinterp{\IfR}{v_2}[n-k-1-j]
        \end{align*}

        Since
        $\neestepsto{\If{\psubexpr{\unexpr_1}}{\psubexpr{\unexpr_2}}{\psubexpr{\unexpr_3}}}{v_2}[k+1+j]$,
        together with the value interpretation above, by applying \Cref{appendix:lem:hdxp}
        \begin{align*}
          \expinterp{\IfR}{\psubexpr{\If{\unexpr_1}{\unexpr_2}{\unexpr_3}}}
        \end{align*}
      \end{itemize}
    \item $v_1 = \false$. This case follows reasoning analogous to that in the
      previous case, using the premise that types $\unexpr_3$ at $\IfR$. The IH
      for $\unexpr_3$ is first weakened to index $n-k-1$
      to account for the \myrule{E-IfFalse} step. If
      $\neestepsto{\psubexpr{\unexpr_3}}{v_3}[j]$, then the whole expression
      reaches $v_3$ in $k+1+j$ steps.
    \end{caselist}
    \end{itemize}
  \item
    \RenewDocumentCommand{\appfunrngres}{}{\R'}
    \myrule{T-App}. $\expr = \ap{\expr_1}{\expr_2}$, $\R' = \appfunrngresunfold$, and $\R = \appres$. Let $\unexpr_1= {\eraseann{\expr_1}}, \unexpr_2 = {\eraseann{\expr_2}}$. The premises are
    \begin{align}
      &\typeof{\ctx}{\expr_1}{\appfunres} \label{appendix:eq:type-} \\
      &\typeof{\ctx}{\expr_2}{\appargres} \label{appendix:eq:type--1}
    \end{align}
    The goal is to show $\expinterp{\appres}{\psubexpr{\ap{\unexpr_1}{\unexpr_2}}}$.

    By IH on $\allsat$ and \cref{appendix:eq:type-,appendix:eq:type--1} respectively,
    \begin{align}
      \expinterp{\appfunres}{\psubexpr{\unexpr_1}} \label{appendix:eq:semr--2}\\
      \expinterp{{\appargres}}{\psubexpr{\unexpr_2}} \label{appendix:eq:semr--3}
    \end{align}

    By applying \Cref{appendix:lem:hdrd} on \cref{appendix:eq:semr--2}, there are two cases.
    \begin{caselist}
    \item $\hdrdpropa[\psubexpr{\unexpr_1}]$.
      \propaconseqsub[\psubexpr{\unexpr_1}][\unexpr'_1]{\psubexpr{\app{\unexpr_1}{\unexpr_2}}}{\app{\unexpr'_1}{\psubexpr{\unexpr_2}}}[\appres]
    \item $\hdrdpropb[\appfunres][v_1][k][\psubexpr{\unexpr_1}]$.
      \newcommand{\appfunenv}{\env}
      \newcommand{\argenv}{\env}
      \newcommand{\allenv}{\env}

      By definition,
      \begin{align*}
        \valinterp{\appfunres}{v_1}[n-k]
      \end{align*}
      which gives
      \begin{align}
        \valinterp{\appfunty}{v_1}[n-k] \label{appendix:eq:semr--1}
      \end{align}
      By \cref{appendix:eq:semr--1},
      \begin{align}
        \forall i \le n-k, w'.\: \valinterp{\tyB}{w'}[i] \metaimpl \expinterp{\appfunrngres}{\ap{v_1}{w'}}[i][{\extenv{\env}{x}{w'}}]   \label{appendix:eq:forall-w.-meta}
      \end{align}
      Now let us move on to the argument of the application.
      By applying \Cref{appendix:lem:srlmono} to \cref{appendix:eq:semr--3},
      \begin{align}
        \expinterp{{\appargres}}{\psubexpr{\unexpr_2}}[n-k] \label{appendix:eq:semr--33}
      \end{align}
      By applying \Cref{appendix:lem:hdrd} to \cref{appendix:eq:semr--33}, there are two cases.
      \begin{caselist}
        \item $\hdrdpropa[\psubexpr{\unexpr_2}][n-k]$.
          \propaconseqsub[\psubexpr{\unexpr_2}][\unexpr'_2]{\app{v_1}{\psubexpr{\unexpr_2}}}{\app{v_1}{\unexpr'_2}}[\appres][n-k]
          Since
          $\neestepsto{\psubexpr{\app{\unexpr_1}{\unexpr_2}}}{\app{v_1}{\psubexpr{\unexpr_2}}}[k]$,
          by \Cref{appendix:lem:hdxp}
          \begin{align*}
            \expinterp{\appres}{\psubexpr{\ap{\expr_1}{\expr_2}}}
          \end{align*}
        \item $\hdrdpropb[\appargres][v_2][j][\psubexpr{\unexpr_2}][n-k]$.
          By definition, this gives
          \begin{align}
            \propertyA{\obj_2}{v_2}[\argenv] \label{appendix:eq:-v_2-=} \\
            \valinterp{\tyB}{v_2}[n-k-j] \label{appendix:eq:semr--1a}
          \end{align}
          Because $n - k - j \le n -k$, applying \cref{appendix:eq:forall-w.-meta} to \cref{appendix:eq:semr--1a} gives us
          \begin{align}
            \expinterp{\appfunrngres}{\ap{v_1}{v_2}}[n-k-j][{\extenv{\env}{x}{v_2}}]\label{appendix:eq:semr-k-j-1}
          \end{align}
          By \Cref{appendix:lem:hdrd} on \cref{appendix:eq:semr-k-j-1}, there are two cases
          \begin{itemize}
            \item $\hdrdpropa[\ap{v_1}{v_2}][n-k-j]$.
              \propaconseqsub[\ap{v_1}{v_2}][\unexpr'_o]{\psubexpr{\app{\unexpr_1}{\unexpr_2}}}{\unexpr'_o}[\appres][n-k-j]
              Since
              $\neestepsto{\psubexpr{\app{\unexpr_1}{\unexpr_2}}}{\ap{v_1}{v_2}}[k+j]$,
              by \Cref{appendix:lem:hdxp}
              \begin{align*}
                \expinterp{\appres}{\psubexpr{\ap{\expr_1}{\expr_2}}}
              \end{align*}
            \item $\hdrdpropb[\appfunrngres][v_o][i][\ap{v_1}{v_2}][n-k-j]$.
              \newcommand{\finalidx}{n-k-j-i}
              \newcommand{\totalsteps}{k+j+i}
              By definition,
              \begin{align}
                \valinterpres{\appfunrngresunfold}{v_o}[\finalidx][{\enva[v_2]}]
                \label{appendix:eq:app-body-result-val}
              \end{align}

              By \cref{appendix:eq:-v_2-=}, $\propertyA{\obj_2}{v_2}[\enva[v_2]]$.
              Together with \cref{appendix:eq:app-body-result-val}, applying
              \Cref{appendix:lem:obj-subst-logrel} gives
              \begin{align}
                \valinterpres{\osubst{\appfunrngresunfold}[x][\obj_2]}{v_o}[\finalidx][\enva[v_2]]
                \label{appendix:eq:app-subst-body-result-val-ext}
              \end{align}
              Because $x$ does not occur free anymore, by \Cref{appendix:lem:psubstr},
              \begin{align}
                \valinterpres{\osubst{\appfunrngresunfold}[x][\obj_2]}{v_o}[\finalidx]
                \label{appendix:eq:app-subst-body-result-val}
              \end{align}
              Since
              $\neestepsto{\psubexpr{\app{\unexpr_1}{\unexpr_2}}}{v_o}[\totalsteps]$,
              by \Cref{appendix:lem:hdxp},
              \begin{align*}
                \expinterp{\osubst{\appfunrngresunfold}[x][\obj_2]}{\psubexpr{\app{\unexpr_1}{\unexpr_2}}}
              \end{align*}
              .
          \end{itemize}
        \end{caselist}
    \end{caselist}
  \item
    \NewDocumentCommand{\letboundres}{}{\res[\tyA][\posprop][\negprop][\obj]}
    \NewDocumentCommand{\letbodyctx}{}{\ctx \exta \bindvar \exta \osubst{(\eqprop{x}{z})}[z][\obj] \exta \left(\disjprop{\conjprop{\nottypeprop{x}{\false}}{\posprop}}{\conjprop{\typeprop{x}{\false}}{\negprop}}\right)}
    \myrule{T-Let}. $\expr = \Let{\expr_1}{\expr_b}$, $\eraseann{\expr} =
    \Let{\eraseann{\expr_1}}{\eraseann{\expr_b}}$, and the result is
    $\osubst{\R}[x][\obj]$. Let $\unexpr_1 = \eraseann{\expr_1}$ and
    $\unexpr_b = \eraseann{\expr_b}$. The premises are
    \begin{align}
      \typeof{\ctx}{\expr_1}{\letboundres} \label{appendix:eq:let-type-bound}\\
      \typeof{\letbodyctx}{\expr_b}{\R} \label{appendix:eq:let-type-body}
    \end{align}
    The goal is to show
    $\expinterp{\osubst{\R}[x][\obj]}{\psubexpr{\Let{\unexpr_1}{\unexpr_b}}}$.

    By IH on $\allsat$ and \cref{appendix:eq:let-type-bound},
    \begin{align}
      \expinterp{\letboundres}{\psubexpr{\unexpr_1}} \label{appendix:eq:let-exp-bound}
    \end{align}
    By \Cref{appendix:lem:hdrd} on \cref{appendix:eq:let-exp-bound}, there are two cases.
    \begin{caselist}
    \item $\hdrdpropa[\psubexpr{\unexpr_1}]$.
      \propaconseqsub[\psubexpr{\unexpr_1}][\unexpr'_1]{\psubexpr{\Let{\unexpr_1}{\unexpr_b}}}{\Let{\unexpr'_1}{\psubexpr{\unexpr_b}}}[\osubst{\R}[x][\obj]]
    \item $\hdrdpropb[\letboundres][v][k][\psubexpr{\unexpr_1}]$.
      By definition,
      \begin{align}
        \propertyA{\obj}{v} \label{appendix:eq:let-property-a}\\
        \valinterp{\tyA}{v}[n-k] \label{appendix:eq:let-val-bound}\\
        \propertyC{v}{}{\env}{n-k} \label{appendix:eq:let-property-c}
      \end{align}
      By \Cref{appendix:lem:srlmono} on $\allsat$,
      \begin{align}
        \ctxinterp{\ctx}{\env}[n-k] \label{appendix:eq:let-sat-mono}
      \end{align}
      Since \cref{appendix:eq:let-val-bound,appendix:eq:let-sat-mono},
      by definition,
      \begin{align}
        \ctxinterp{\ctxa}{\enva}[n-k] \label{appendix:eq:ctxint-k-}
      \end{align}

      By \cref{appendix:eq:let-property-a}, there are two cases
      \begin{itemize}
        \item  $\obj = \topobj$. $\osubst{(\eqprop{x}{z})}[z][\obj] = \topobj$. Then by definition, $\propinterp{\TT}[n-k][\enva]$
        \item  $\propertyAsnd[y]{\obj}{v}$. Then \Cref{appendix:lem:psubwk}, $\enva(y) = v$. By definition, $\propinterp{\eqprop{x}{y}}[n-k][\enva]$
      \end{itemize}
      In both cases, we have proved
      $\propinterp{\osubst{(\eqprop{x}{z})}[z][\obj] }[n-k][\enva]$. Together
      with \cref{appendix:eq:ctxint-k-,appendix:eq:let-property-c}, by definition, this gives
      \begin{align}
        \ctxinterp{\letbodyctx}{\enva}[n-k]
        \label{appendix:eq:let-sat-body}
      \end{align}
      Since $0 \le k < n$, $0 \le n - k - 1 < n$.
      By IH on \cref{appendix:eq:let-sat-body,appendix:eq:let-type-body},
      \begin{align}
        \expinterp{\R}{\psubexpr{\unexpr_b}[\enva]}[n-k][\enva]
        \label{appendix:eq:let-exp-body}
      \end{align}
      By \Cref{appendix:lem:srlmono},
      \begin{align}
        \expinterp{\R}{\psubexpr{\unexpr_b}[\enva]}[n-k-1][\enva]
        \label{appendix:eq:let-exp-body-1}
      \end{align}
      Let $i \le n-k-1$ and $v_o$ be such that
      $\neestepsto{\psubexpr{\unexpr_b}[\enva]}{v_o}[i]$. By
      \cref{appendix:eq:let-exp-body-1},
      \begin{align}
        \valinterpres{\R}{v_o}[n-k-1-i][\enva] \label{appendix:eq:let-val-body}
      \end{align}
      Since $\enva(x) = v$, \cref{appendix:eq:let-property-a} gives
      $\propertyA{\obj}{\enva(x)}[\enva]$. Together with
      \cref{appendix:eq:let-val-body}, applying \Cref{appendix:lem:obj-subst-logrel} gives
      \begin{align}
        \valinterpres{\osubst{\R}[x][\obj]}{v_o}[n-k-1-i][\enva]
        \label{appendix:eq:let-body-val-subst-ext}
      \end{align}
      Because $x$ does not occur free anymore, by \Cref{appendix:lem:psubstr},
      \begin{align}
        \valinterpres{\osubst{\R}[x][\obj]}{v_o}[n-k-1-i]
        \label{appendix:eq:let-body-val-subst}
      \end{align}
      By \myrule{E-Let},
      $\stepsto{\Let{v}{\psubexpr{\unexpr_b}}}{\psubexpr{\unexpr_b}[\enva]}$.
      Therefore,
      \begin{align*}
        \neestepsto{\psubexpr{\Let{\unexpr_1}{\unexpr_b}}}{\psubexpr{\unexpr_b}[\enva]}[k+1]
      \end{align*}
      Together with $\neestepsto{\psubexpr{\unexpr_b}[\enva]}{v_o}[i]$,
      \begin{align}
        \neestepsto{\psubexpr{\Let{\unexpr_1}{\unexpr_b}}}{v_o}[k+1+i]
        \label{appendix:eq:let-step-body-value}
      \end{align}
      equivalently, any evaluation through the body takes $i+k+1$ steps.
      Together with
      \cref{appendix:eq:let-body-val-subst,appendix:eq:let-step-body-value}, applying \Cref{appendix:lem:hdxp} gives
      \begin{align*}
        \expinterp{\osubst{\R}[x][\obj]}{\psubexpr{\Let{\unexpr_1}{\unexpr_b}}}
      \end{align*}
    \end{caselist}
  \item \myrule{T-Subsume}. The premise are
    \begin{align}
      \typeof{\ctx}{e}{\synR_1} \label{appendix:eq:type-subtyp-} \\
      \subtype{\ctx}{\synR_1}{\synR_2} \label{appendix:eq:subtyp-}
    \end{align}
    By IH on \cref{appendix:eq:type-subtyp-},
    \begin{align}
      \expinterp{\R_1}{\psubexpr{e}} \label{appendix:eq:semr--4}
    \end{align}
    By \Cref{appendix:lem:semressub} on $\allsat$, \cref{appendix:eq:subtyp-,appendix:eq:semr--4},
    \begin{align*}
      \expinterp{\R_2}{\psubexpr{e}}
    \end{align*}
  \end{caselist}
\end{proof}

\begin{corollary}[Syntactically Well-Typed Closed Terms are Semantically Well-Typed]\label[corollary]{appendix:cor:closed-ftlr}
  If $\typeof{\empctx}{\expr}{\R}$ then $\expinterp{\R}{\eraseann{\expr}}[n][\empenv]$ for all $n$.
\end{corollary}
\begin{proof}
  By \Cref{appendix:thm:fundamental} with $\ctx = \empctx$ and $\env = \empenv$.
\end{proof}

\begin{theorem}[Semantic Type Soundness]\label[theorem]{appendix:thm:type-soundness}
  If $\typeof{\empctx}{\expr}{\R}$, then either $\eraseann{\expr}$ diverges or
  $\neestepsto{\eraseann{\expr}}{v}$ and $\valinterp{\R}{v}[n][\empenv]$ for all $n$.
\end{theorem}
\begin{proof}
  By \Cref{appendix:cor:closed-ftlr}, we have $\expinterp{\R}{\eraseann{\expr}}[n][\empenv]$ for all $n$.

  Terms must either diverge, evaluate to a value, or get stuck (i.e., evaluate to an irreducible
  non-value), so we begin by showing that the third possibility cannot occur for well-typed terms.
  Suppose that $\neestepsto{\eraseann{\expr}}{\unexpr'}[k]$ for some irreducible non-value
  $\unexpr'$. By head reduction (\Cref{appendix:lem:hdrd}) and $\expinterp{\R}{\eraseann{\expr}}[k+1]$, we
  know that either (1) there exists $\unexpr''$ such that
  $\neestepsto{\eraseann{\expr}}{\unexpr''}[k+1]$, or (2) there exists a value $v$ such that
  $\neestepsto{\eraseann{\expr}}{v}[m]$ for $m<k+1$. By determinacy, the first contradicts the
  hypothesis that $\unexpr'$ is irreducible, and the latter contradicts the irreducibility of
  values.

  It remains only to show that if $\neestepsto{\eraseann{\expr}}{v}[k]$, then
  $\valinterp{\R}{v}[n][\empenv]$ for all $n$. By determinacy and head reduction applied to
  $\expinterp{\R}{\eraseann{\expr}}[n+k][\empenv]$, we obtain $\valinterp{\R}{v}[n+k-k][\empenv]$,
  completing the proof.
\end{proof}

\begin{corollary}[Soundness for Booleans]
  If $\typeof{\empctx}{\expr}{\res[\Bool][\posprop][\negprop][\obj]}$ then either $\eraseann{\expr}$
  diverges or $\neestepsto{\eraseann{\expr}}{\true}$ or
  $\neestepsto{\eraseann{\expr}}{\false}$.
\end{corollary}
\begin{proof}
  By \Cref{appendix:thm:type-soundness} it suffices to show that
  $\valinterp{\res[\Bool][\posprop][\negprop][\obj]}{v}[1][\empenv]$ implies $v=\true$ or
  $v=\false$. Recalling that $\Bool \coloneq \uniontype{\true}{\false}$, this follows immediately by
  unfolding the definition of the logical relation.
\end{proof}

\AppendixFinish

\fi
\end{document}